\documentclass[11pt,reqno]{amsart}
\usepackage[foot]{amsaddr}
\usepackage[margin=1.2in]{geometry}
\usepackage{graphicx}
\usepackage{mathrsfs,amsmath,amssymb}
\usepackage{subfig,hyperref,xcolor,enumerate,enumitem}
\usepackage{cite, bm}
\usepackage{cleveref}
\usepackage{tikz,multirow}
\usepackage{xcolor,cancel,amsthm}
\newcommand{\dd}{\mathrm{d}}
\newcommand{\bU}{\bar{U}}

\newcommand{\cS}{\mathcal{S}}
\newcommand{\smax}{s_{\mathrm{ max}}}
\newcommand{\qand}{\quad \mbox{and} \quad}

\DeclareMathOperator{\tr}{tr}

\mathchardef\ordinarycolon\mathcode`\:
\mathcode`\:=\string"8000
\begingroup \catcode`\:=\active
  \gdef:{\mathrel{\mathop\ordinarycolon}}
\endgroup

\newtheorem{lem}{Lemma}[section]
\newtheorem{defn}{Definition}[section]
\newtheorem{prop}{Proposition}[section]

\begin{document}

 % \title[Black holes and their evolutions]{Black holes and their evolutions: \\ lessons from bifurcation theory
 % }

  \title[Black hole evolutions]{Black hole evolutions: \\ lessons from bifurcation theory
 }

%\title[Black hole evolution and bifurcation theory]{Black hole evolution and bifurcation theory
%}

\author{Ivan Booth$^{1,2,*}$, Graham Cox$^1$, Chiamaka Mary Okpala$^2$}

\address{$^1$ Department of Mathematics and Statistics, Memorial University, NL, Canada}
\address{$^2$ Department of Physics and Physical Oceanography, Memorial University, NL, Canada}

\thanks{* \texttt{ibooth@mun.ca} (corresponding author)}

\begin{abstract} 
We explore the role of the stability operator in regulating the evolution of marginally outer trapped surfaces (MOTSs). In 2005, Andersson, Mars and Simon showed that if the stability operator is invertible,  then the time evolution of a MOTS is unique. Here we focus on  moments at which the stability operator is not invertible.  Understanding MOTSs as analogous to fixed points, and the stability operator as a linearization of the system of MOTS-defining equations, bifurcation theory can be used to classify possible non-unique evolutions. MOTS pair creation/annihilation is an example of a saddle-node bifurcation but other possibilities can occur, including pitchfork and transcritical bifurcations. Analytical and numerical tools are used to identify examples of the various bifurcations in a variety of spacetimes. To help analyze those results, we define a generalized MOTS stability operator and discuss the (partial) barrier properties of unstable MOTSs. Spherically symmetric examples are given in the Reissner--Nordstr\"om--de Sitter spacetime and axisymmetric examples are studied in Reissner--Nordstr\"om and Weyl-distorted Schwarzschild. 
This application of bifurcation theory is very general and so applies to any theory of gravity containing MOTSs (or some generalization thereof). Possible bifurcations of those structures are constrained in any such theory in the same way.
\end{abstract}

\maketitle

\setcounter{tocdepth}{1}
\tableofcontents

\section{Introduction}

A marginally outer trapped surface 
(MOTS) is a closed, spacelike two-surface $\mathcal{S}$ for which the expansion vanishes,
\begin{equation}
    \left. \theta_\ell \right|_\mathcal{S} = 0 \, , 
\end{equation}
where $\ell$ is any future outward-oriented null normal to $\mathcal{S}$. 
Apparent horizons are the best known and most important examples of MOTSs. However, over the last few years, it has become clear that most MOTSs are 
not apparent horizons.
Large, likely infinite, families of MOTSs have been found
in the interior of apparent horizons 
not only during black hole mergers \cite{Pook-Kolb:2021jpd,Booth:2021sow,Pook-Kolb:2021gsh}
but also in even the simplest exact solutions \cite{Booth:2020qhb,Hennigar:2021ogw,Booth:2022vwo,Sievers:2023zng}. These interior MOTSs often have complicated, self-intersecting geometries and some even have non-spherical topologies. During the numerical time evolution of a merger, as well as during exact solution phase space evolutions (such as changing the charge in Reissner--Nordstr\"om), these exotic MOTSs also evolve and, notably, 
engage in pair creations and annihilations. For example, the well-known ``jumps'' of 
apparent horizons seen in both numerical 
mergers \cite{CADEZ1974449,Matzner:1995ib,Anninos:1994gp,thornburg}
and exact-solution spacetimes \cite{Booth:2005b,bendov} are
better understood as MOTS pair creations.  
Further,  during a merger, the final dissolution of the original two black holes in the interior of the new final black hole can be understood to occur when their apparent horizons annihilate in meetings with interior MOTS \cite{Pook-Kolb:2021jpd,Pook-Kolb:2021gsh}.

\subsection{MOTS stability operator: classifying local apparent horizons}
\label{sec:introapp}

\begin{figure}
    \centering
\includegraphics[width=0.9\textwidth]{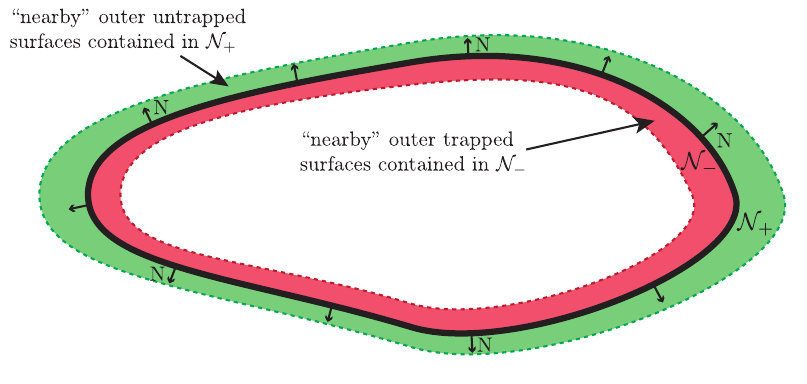}
    \caption{A strictly stable MOTS $\mathcal{S}$ has a two-sided neighbourhood $\mathcal{N}$
    %with outer untrapped surfaces in the outside green region $\mathcal{N}_+$ and outer trapped surfaces in the inside red region $\mathcal{N}_-$. Further,  
    such that any outer trapped (untrapped) surface contained in $\mathcal{N}$ is necessarily entirely contained in the inside $\mathcal{N}_-$ (outside $\mathcal{N_+}$). $\mathcal{S}$ is thus a barrier between outer trapped and untrapped surfaces that are contained in $\mathcal{N}$. }
    \label{fig:stability}
\end{figure}

With this multitude of new MOTSs, the study of interior black hole geometries and their dynamics appears to explode in complexity. However, there is a mathematical tool that imposes some order. Consider a MOTS $\mathcal{S}$ embedded in a spacelike 
slice $\Sigma_t$, 
with ${N}$ the outward-pointing unit normal to $\mathcal{S}$ in $\Sigma_t$ (Figure \ref{fig:stability}). The \emph{MOTS stability operator},  introduced by Andersson, Mars and Simon in \cite{Andersson:2005,Andersson:2007fh}, is defined by\footnote{The stability operator, and hence the stability classification of $\cS$, depends on the choice of spacelike slice $\Sigma_t$; see Appendix~\ref{sec:stabop} and in particular Figure~\ref{fig:zoom} for further discussion.}
\begin{equation}
    L_{\mathcal S} \psi := \delta_{\psi N} \theta_\ell
\end{equation}
for any function $\psi$ on $\cS$ . This measures the rate of change of the 
expansion if $\mathcal{S}$ is deformed by an amount $\psi$ in the $N$
direction. Equivalently, 
it is the linearization of $\theta_\ell$ when $\mathcal{S}$ is deformed by an amount $\epsilon \psi$ in the $N$ direction:
\begin{equation}
\left. \theta_{\ell} \right|_{\mathcal{S} + \epsilon \psi N} =
 (L_\mathcal{S} \psi) \, \epsilon
+ \mathrm{O}(\epsilon^2)\; . 
\end{equation} 
Like the closely related stability operator 
for minimal surfaces in Riemannian geometry, this is also a linear elliptic operator (though in general it is not self-adjoint). As such it has a discrete spectrum, consisting of isolated eigenalues with finite multiplicity. Moreover, the eigenvalue with the smallest real part, the principal eigenvalue $\lambda_o$, is non-degenerate and real. 

If $\lambda_o >0$ then $\mathcal{S}$ is said to be \emph{strictly stable}. In this case there exists a two-sided neighbourhood $\mathcal{N}$ of $\mathcal{S}$ such that any outer trapped surface contained in $\mathcal{N}$ must be entirely contained in the inner part $\mathcal{N}_-$, and any outer untrapped surface contained in $\mathcal{N}$ must be entirely contained in the outer part $\mathcal{N}_+$ \cite{Andersson:2005,Andersson:2007fh}. This is illustrated in Figure \ref{fig:stability}.

Hence, strictly stable MOTSs form local barriers between an outside region of outer untrapped surfaces and an inside region of outer trapped surfaces. As such, they are sensible local replacements for apparent horizons, where we recall that the apparent horizon is the boundary of the weakly outer trapped region in $\Sigma_t$\cite{Hawking:1973}. In \cite{Andersson:2007} it was shown that for asymptotically flat initial data the apparent horizon is the unique outermost MOTS; in particular, it is smooth and stable. Previous results in \cite{Hawking:1973,KH97} show that \emph{if} the apparent horizon is assumed to be sufficiently smooth, then it must be a stable MOTS.

Unstable MOTSs (with $\lambda_o<0$) flip this distinction: they have outer trapped surfaces outside and outer untrapped surfaces inside\footnote{In Section 
\ref{sec:inout} we will see that, while there is still a partial separation of trapped versus untrapped regions, an unstable MOTS is not an impermeable barrier. There can be (un)trapped surfaces in $\mathcal{N}$ that straddle it.}. This opposite behaviour for stable versus unstable MOTSs is illustrated in Figure \ref{fig:mtt}.

This understanding of strictly stable MOTSs as local apparent horizons matches intuition.
During a non-rotating, head-on black hole merger, the only stable MOTSs are the two initial apparent horizons and the final apparent horizon\cite{Pook-Kolb:2021jpd,Pook-Kolb:2021gsh}. All other MOTSs have one or more negative eigenvalues and so should not be thought of as horizons.
Similarly, in standard black hole exact solutions, such as the Kerr--Newman--de Sitter family, only the outermost black hole horizons are stable MOTSs.
The MOTSs foliating interior and cosmological horizons are unstable \cite{Booth:2020qhb,Hennigar:2021ogw,Booth:2022vwo,Sievers:2023zng}.

{\bf Other notions of stability:} While the preceding formulation is key to our modern understanding of local apparent horizons,  there are other, closely related  ideas.
Almost 40 years ago,  Newman \cite{Newman_1987}  introduced a notion of MOTS stability  to localize Hawking's horizon topology thereom\cite{Hawking:1973}. In modern language he showed that, given energy conditions, a strictly stable MOTS is necessarily diffeomorphic to a sphere. Unstable MOTSs are not restricted in their topologies.

A few years later, Hayward\cite{Hayward:1994} introduced \emph{trapping horizons}: three-dimensional surfaces that are foliated by MOTS. In particular
\emph{future outer trapping horizons (FOTHs)} \cite{Hayward:1994} were intended as local replacements for time-evolved apparent horizons. For Hayward, ``outer'' is a version of stability though with deformations taken in the future inward null direction $l_-$: $\delta_{l_-} \theta_\ell < 0$ while ``future'' means $\theta_{l_-} < 0$. Then, a FOTH has
fully trapped surfaces ``just'' inside and untrapped surfaces ``just'' outside. 
 This null direction stability  is stronger than spacelike stability with respect to a particular surface: an ``outer'' MOTS is strictly stable with respect to all foliations. 
%However strict stability with respect to a spacelike direction $\mathcal{N}$ does not guarantee stability with respect to $l_-$ (or even all other spacelike directions). 
From a spacetime perspective this is a more invariant notion. However, dynamical black hole studies are almost always numerical with the MOTSs found with respect to a particular foliation. 
For such work, the Andersson, Mars, Simon notion of stability is a better fit as the observations are generated by a particular foliation. For further discussion of more general notions of stability, see Appendix \ref{sec:stabop}. 

Finally, just before \cite{Andersson:2005},  Ashtekar and Krishnan introduced dynamical horizons\cite{AshtekarKrishnan2002,Ashtekar:2004a}. They weren't explicitly defined to be stable but this followed almost immediately  from their definition.

\subsection{MOTS stability operator: governing evolution}

The stability operator is  more than just a tool for classifying MOTSs.  
If spacetime is deformed, it also determines how $\mathcal{S}$ must respond in order to remain a MOTS.

Consider a MOTS $\mathcal{S}$ with stability operator $L_{\mathcal{S}}$ in a time slice $\Sigma_{t_0}$.
It was shown in \cite{Andersson:2005,Andersson:2007fh,Andersson:2008up} that if zero is not an eigenvalue of $L_{\mathcal{S}}$, so it is invertible, then $\mathcal{S}$ can be evolved uniquely in time into a \emph{marginally outer trapped tube (MOTT)} for some finite interval $(t_a, t_b)$ 
containing $t_o$, with the evolved $S_t \subset \Sigma_t$. Note that the unique evolution is guaranteed into both the past and future, since $t_a < t_o < t_b$.
Strictly stable MOTSs not only evolve uniquely (their principal eigenvalue is positive and so all other eigenvalues are also non-zero) but their associated MOTTs are 
necessarily null or spacelike and (given appropriate energy conditions) non-contracting in area (Figure \ref{fig:mtt}).

It is important to note that the proof of unique evolution applies to any variation of
the spacetime containing the MOTS. For example, it would equally well cover variations of the charge parameter $q$ in Reissner--Nordstr\"om, a deformation parameter in the Weyl-distorted 
Schwarzschild solutions \cite{Pilkington:2011aj} or indeed any other perturbation of 
the spacetime. In any of these cases (to which 
we will return) the existence of a MOTS with invertible stability operator in an initial spacetime 
implies that in any perturbation of that spacetime there will also be a perturbed version 
of that MOTS.

If an eigenvalue vanishes, then the MOTS can have more interesting 
evolutions. Starting with the simple case of
spherically symmetric dust collapse, it was 
observed in \cite{Booth:2005b} (and is illustrated in 
Figure \ref{fig:mtt}) that at points where the 
principal eigenvalue vanishes, MOTSs can be 
created and annihilated in pairs. In the figure, dust accretes
onto an existing black hole. From $t_1$ to $t_2$ the principal 
eigenvalue of each MOTS is positive and so the MOTT expands. However at $t_2$ the infalling dust reaches sufficient density that a new MOTS forms outside of the original one. This has 
vanishing principal eigenvalue and so can bifurcate into 
outer and inner MOTSs. The outer MOTS has positive principal 
eigenvalue and is expanding. The inner one has negative principal eigenvalue. Hence, it  no longer forms a boundary with outer untrapped surfaces outside and outer trapped surfaces inside. Instead the outer trapped region is now outside while the outer untrapped region is inside. This MOTS contracts until $t_6$ when it meets the original horizon, at which point they annihilate with vanishing principal eigenvalue (a non-unique evolution into the past with no evolution into the future). At this point the original 
black hole has lost its distinct identity and so we say that it 
has \emph{dissolved} in the new, larger one. 

\begin{figure}
    \centering
    \includegraphics[width=0.9\linewidth]{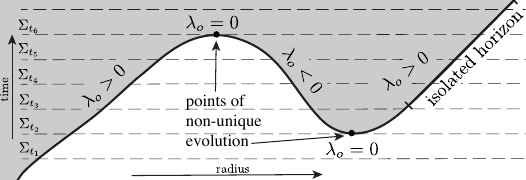}
    \caption{Cartoon spacetime diagram of the evolution of a spherical MOTT as ingoing matter forces it to expand. The MOTS bifurcates at the points where the principal eigenvalue vanishes, and so its evolution is non-unique. In contrast to the strictly stable case, it only extends into one of the past/future. At the end of the evolution it settles down into equilibrium as an isolated horizon. As per Section \ref{sec:introapp}, when the MOTS is strictly stable the outer trapped region is inside and the untrapped region is outside. }
    \label{fig:mtt}
\end{figure}

To date, the most detailed tracking of MOTSs during binary mergers was conducted in the series of papers \cite{Pook-Kolb:2018igu,Pook-Kolb:2019iao,Pook-Kolb:2019ssg, Pook-Kolb:2021gsh,Booth:2021sow,Pook-Kolb:2021jpd,kastha2026cuspformationmergingblack}. These studied head-on 
collisions of non-rotating black holes and found an extensive, intricate
series of MOTS pair creations and annihilations that mediated the 
evolution from the two initial apparent horizons into that of the final 
black hole. In contrast to the examples considered in Figure \ref{fig:mtt}, these MOTSs had exotic, non-symmetric geometries
and many even had (multiple) self-intersections. Further, creations and annihilations in these
cases occurred not only for a vanishing principal eigenvalue, but also 
for vanishing higher eigenvalues: these events don't just occur between stable and 
unstable MOTS but can also happen for pairs of unstable ones. 

There were also hints that other things could be happening. 
During the simulations of \cite{Pook-Kolb:2021jpd} it was observed
that it is also possible for eigenvalues to vanish during otherwise
smooth evolutions (see Figure 17 and the associated discussion on page 16 of that reference). This is consistent with the fact that vanishing eigenvalues are 
necessary though not sufficient indicators of
non-unique evolution. 

However, there is another possible explanation of what is happening at that disappearance. 
In \cite{Pook-Kolb:2021jpd},
the authors only searched for axisymmetric MOTS, whereas 
the vanishing 
eigenvalues were associated with non-axisymmetric eigenfunctions. This suggests that they might correspond to new MOTSs
bifurcating off the original one in non-axisymmetric ways. Such events would have been invisible to 
axisymmetric searches. However, 
even if that is the case, there would be something different happening from the now well-known MOTS pair creation and annihilation events. The MOTS which developed the vanishing eigenvalue
did not appear or vanish, but instead
persisted. Any non-unique evolution would then involve
extra MOTSs splitting off from and/or joining the 
original one. 
Hence,
this would have to have been a different kind of evolution.

There are other examples of these kinds of events. 
Another recent paper 
\cite{Hennigar:2021ogw} included a study of MOTSs
in the Reissner--Nordstr\"om spacetime as the charge is varied. In that case the inner horizon 
at $r_{\mathrm{IH}} = m - \sqrt{m^2-q^2}$  persists for any $0 < q \leq m$. This is a 
spherical MOTS and so its eigenfunctions are spherical harmonics $Y_{n p} (\theta, \phi)$ and its eigenvalues may easily be calculated. 
Restricting our attention to 
axisymmetric $Y_{n 0}$, the $n^{\mathrm{th}}$ eigenvalue\footnote{Note that \cite{Hennigar:2021ogw} restricted attention to odd numbered eigenvalues and numbered them accordingly. Hence the formula for $q_n$ appears differently in that paper. } 
vanishes when \cite{Booth:2021sow}
% \begin{equation}
% \frac{Q_N}{M} = \frac{\sqrt{4N^2 + 2N + 1}}{2N^2 + N + 1}  \label{eq:RNvanishX}
% \end{equation}
\begin{equation}
    \frac{q_n}{m} = \frac{2 \sqrt{n^2 + n + 1}}{n^2 + n + 2}  
    = \left\{1, \frac{\sqrt{3}}{2}, \frac{\sqrt{7}}{4}, \frac{\sqrt{13}}{7}, \frac{\sqrt{21}}{11}, \dots \right\} \label{eq:RNvanish}
\end{equation}
and in general for $q_{n+1} < q \leq q_n$
there will be $n$ negative eigenvalues. 
In that paper it was seen (see for example
Figure 12 and the associated discussion in \cite{Hennigar:2021ogw}) that other MOTSs may split from the inner one and 
this splitting occurs at points where an eigenvalue vanishes. These are examples of 
non-unique evolution events that are not pair creations/annihilations.

 \subsection{Overview}

This paper returns to more carefully examine MOTSs with vanishing eigenvalues and the non-unique evolutions of MOTTs from those surfaces. Our core result is that the possible non-unique evolutions are constrained by results from bifurcation theory. This sheds light on the results mentioned at the end of the last section and also motivates the search for other bifurcations. Finding and constructing these examples is technically non-trivial, hence a significant portion of the paper is spent developing the analytic and numerical tools needed to do that work. We proceed as follows. 

In Section \ref{backgroundI},  we review both the basic mathematics of MOTSs as well as to how, in axisymmetry, the MOTS-defining 
equations may be reformulated as a coupled pair of ordinary differential equations which can be solved by standard numerical methods. These are the equations that we use to identify MOTSs in the various spacetimes that we study. We also define the MOTS stability operator and use it to characterize when a MOTS forms a boundary between trapped and untrapped regions.

Section \ref{sec:MOTSevolve} introduces the tools of bifurcation theory to the study of MOTS evolution, emphasizing the role of the stability operator and the connection to finite-dimensional dynamical systems. We see that MOTS pair creations/annihilations are examples of saddle-node birfurcations but other possibilites include transcritical and pitchfork bifurcations.  We also discuss the extent to which our results depend on the choice of spacelike slice through $\cS$.

%\ivan{In particular, we start with a review of one-dimensional bifurcations} and describe the three normal forms (saddle-node, transcritical, pitchfork) that arise in our examples, then explain how this applies to the infinite-dimensional phase space of MOTS bifurcations.

\begin{table}
\begin{tabular}{|l||l|c|c|c|}
    \hline
    \multirow{2}{*}{Bifurcation} & \multirow{2}{*}{MOTS with vanishing eigenvalue} & \multirow{2}{*}{Symmetry} & Stable & \multirow{2}{*}{Section} \\ & & & MOTS? & \\ \hline \hline
    \multirow{2}{*}{saddle-node} & \multirow{2}{*}{i) RN extremal} & \multirow{2}{*}{spherical} & \multirow{2}{*}{Yes} & \ref{sec:RNdS_SN} \\
    & & & &  \ref{sec:RN_innerSN} \\ \hline
    \multirow{6}{*}{transcritical} & \multirow{2}{*}{i) SdS ``ultramassive'' extremal} & \multirow{2}{*}{spherical} & \multirow{2}{*}{Yes} & \multirow{2}{*}{\ref{sec:RNdStrans}} \\
    & & &  & \\ 
    \cline{2-5}
    & \multirow{2}{*}{ii) RN inner horizon} & spherical  & \multirow{2}{*}{No}  & \multirow{2}{*}{\ref{sec:RN_innertrans}} \\
    & & and axial &  &  \\
    \cline{2-5}
    & \multirow{2}{*}{iii) Weyl--Schwarz Killing horizon} & \multirow{2}{*}{ axial}& \multirow{2}{*}{Yes} & \multirow{2}{*}{\ref{sec:weyltrans}} \\
    & & & & \\ \hline
    \multirow{6}{*}{pitchfork} & \multirow{2}{*}{i) RNdS triple extremal} & \multirow{2}{*}{spherical} & \multirow{2}{*}{Yes} & \multirow{2}{*}{\ref{sec:RNdS_pitch}} \\
    & & & & \\ 
    \cline{2-5}
    & \multirow{2}{*}{ii) RN inner horizon} & spherical  & \multirow{2}{*}{No} & \multirow{2}{*}{\ref{sec:RN_innerpf}} \\
    & & and axial &  & \\ 
    \cline{2-5}
    & \multirow{2}{*}{iii) Weyl-Schwarz Killing horizon} & \multirow{2}{*}{ axial} & \multirow{2}{*}{No} & \multirow{2}{*}{\ref{sec:weylpitch}} \\
    & & & &  \\ \hline
\end{tabular}
\caption{Summary of the bifurcation examples in this paper. The
second column is the ``central'' MOTS for which the stability operator is not invertible and around which the bifurcation occurs. The third indicates the symmetry of the MOTSs that appear during the bifurcation: either all spherical, some spherical and some axial or all only axial. The fourth indicates whether there are any strictly stable MOTSs involved in the process while the fifth is the section where more details can be found. }
     \label{tab:MOTSexamples}
 \end{table}

The bulk of the paper is Sections \ref{sec:sphere_stab}--\ref{sec:weyl}, which present a variety of examples of the different types of bifurcations in exact solution spacetimes. Prior to this paper only saddle-node bifurcations (pair creation/annihilations) had been observed. Here we present examples of saddle-node, transcritical and pitchfork bifurcations to show that they occur in a wide range of situations. 
Specifically:
\begin{itemize}
    \item Section \ref{sec:sphere_stab} 
considers spherically symmetric bifurcations in Reissner--Nordstr\"om--de Sitter spacetimes, where the inner, outer and cosmological horizons merge or split as the mass, charge and cosmological constant are varied. In these simple cases all calculations can be done analytically.  

    \item Section \ref{sec:axiRN} partially departs from spherical symmetry by studying bifurcations from the inner horizon of Reissner--Nordstr\"om. The background spacetime is spherically symmetric but most of the MOTSs studied are only axisymmetric. 
This section builds from the work of \cite{Hennigar:2021ogw}, extending those results with new insights provided by bifurcation theory. Almost all of the bifurcations considered in this sections are associated with vanishing non-principal eigenvalues. 
\item Section \ref{sec:weyl} examines bifurcations in the Weyl--Schwarzschild spacetimes, focusing on the case of quadrupole distortions. This is still an exact solution but this time neither the spacetime nor the MOTSs are spherically symmetric: we only have axisymmetry. Both principal and non-principal eigenvalue bifurcations are considered. 
\end{itemize}
The types of bifurcations including the spacetime, symmetries of constituent MOTSs and whether or not there is a stable MOTS involved in the process are summarized in Table \ref{tab:MOTSexamples}

In Appendix \ref{sec:stabop} we define a generalized stability operator and establish some properties which, in particular, are very useful in stability calculations. In Appendix \ref{app:bif} we review finite-dimensional bifurcation theory and describe the three normal forms (saddle-node, transcritical, pitchfork) that arise in our examples. Finally, in Appendix \ref{sec:numerics} we review the numerical methods used in this paper.

\section{Background: Marginally outer trapped surfaces}
\label{backgroundI}

\subsection{Definitions}
\label{sec:AH}

Let $(M,g_{ab}, \nabla_a)$ be a four-dimensional time-oriented spacetime
and let $(\mathcal{S},q_{AB}, \mathcal{D}_A)$ be a
 closed, spacelike two-dimensional surface immersed in $M$. It is allowed for $\mathcal{S}$ to have
 self-intersections. 

The normal bundle to $\mathcal{S}$ is two-dimensional and can be
spanned by a dyad of two future-pointing null vectors $(l^+, l^-)$, which we assume can
be meaningfully labeled as outward- and inward-oriented, respectively. We cross-normalize
so that $l^+ \cdot l^- =  -1$. This leaves a degree of rescaling freedom in the
definition of the null vectors. Rescaling by 
\begin{equation}
    (l^+, l^-)  \rightarrow (e^\beta l^+, e^{-\beta} l^-) \label{gauge}
\end{equation} 
for an arbitrary function $\beta$ results in a pair of vectors that also satisfy our scaling condition. 

\subsubsection{Null expansion}

The 
outward/inward null expansions are
\begin{equation}
    \theta_{l^\pm} = q^{ab} \nabla_a l^\pm_b \, , 
    \label{eq:MOTS4D}
\end{equation}
where
\begin{equation}
    q^{ab} := e^a_A e^b_B q^{AB} = g^{ab} + l_+^a l_-^b +  l_-^a l_+^b
\end{equation}
is the push-forward of the inverse of the  induced metric $q_{AB}$ on $\mathcal{S}$. Note that the position of $\pm$ as a sub/superscript doesn't convey meaning but switches up and down as is notationally convenient.

We say that $\mathcal{S}$ is an \emph{outer trapped surface} if $\theta_{l^+}< 0$, a \emph{marginally outer trapped surface (MOTS)} if $\theta_{l^+}=0$, and
%. Encompassing both of these 
a \emph{weakly outer trapped surface} if $\theta_{l^+} \leq 0$.
%
%has $\theta_{l^+} \leq 0$. 
An \emph{outer untrapped surface} has  $\theta_{l^+} >0$ and a \emph{weakly outer untrapped surface} has 
$\theta_{l^+} \geq 0$\footnote{Not all $\mathcal{S}$ are included in this classification, since $\theta_{l^+}$ may change sign
over the surface.} .
The sign of $\theta_{l^+}$ is unchanged if we replace $l^+$ by any other outward-pointing future null vector: such a vector must be of the form $\ell = e^\mu l^+$ for some positive function $e^\mu$, therefore
\begin{equation}
\label{theta:rescale}
    \theta_\ell = e^\mu \theta^{l_+}
    \; .
\end{equation}
This means that  the above classifications are invariant under such rescalings. We will see below (in Section \ref{sec:HGSO} and Appendix \ref{sec:stabop}) how $\mu$ can be chosen to simplify our calculations.

\subsubsection{Rotation}
\label{ssec:rotation}
Expressed relative to the dyad $(l^+, l^-)$, the connection on the normal bundle $T^\perp \mathcal{S}$ is  
\begin{equation}
\omega^{l_+}_A = - e_A^b l^-_c \nabla_b l_+^c \; . 
\label{eq:omega}
\end{equation}
For an axisymmetric MOTS with rotational Killing vector field $\frac{\partial}{\partial \phi}$ and area form $\sqrt{q}$,
\begin{equation}
J = \int_\mathcal{S} \sqrt{q} \left(\frac{\partial}{\partial \phi} \right)^A \omega^{l_+}_A
\end{equation}
is the angular momentum associated with that rotation \cite{Brown:1992br, Ashtekar:1999a, Ashtekar:2000b, Ashtekar:2001jb , Booth:2007}.  This vanishes if
\begin{equation}
    \omega^{l_+}_A = D_A F
    \label{eq:totalder1}
\end{equation}
for some function $F$. We thus say that a MOTS is \emph{non-rotating} if (\ref{eq:totalder1}) holds.  When $\mathcal{S}$  has spherical topology this is equivalent to the curvature of the normal bundle vanishing, $\Omega = d \omega^{l_+} = 0$.

\subsubsection{$(3+1)$ versions}
In mathematical and numerical relativity it is common to work in a $(3+1)$ formalism,
where the spacetime is foliated by spacelike slices $(\Sigma_t, h_{ij}, D_i)$ for some
time parameter $t$. These have a 
future-pointing unit timelike normal $\hat{u}^a$ and extrinsic curvature
\begin{equation}
    K_{ij} = e_i^a e_j^b \nabla_a \hat{u}_b 
\end{equation}
where $e_i^a$ is the pullback operator into $\Sigma_t$. 

If one restricts attention to surfaces $\mathcal{S}$ contained in a slice $\Sigma_t$, the general form of the future-oriented null normals is \begin{equation}
    l^a_+ = e^\beta (\hat{u}^a + N^a)  \quad \mbox{and} \quad  l^a_- = \frac{e^{-\beta}}{2} \left( \hat{u}^a - N^a \right) 
    \label{eq:lplm}
\end{equation}
where $N^a = e^a_i N^i$ is the push-forward of the outward-oriented unit normal
to $\mathcal{S}$ in $\Sigma_t$ and $\beta$ is an arbitrary function. Then
\begin{equation}
    \theta_{l^+} =e^\beta \left( q^{ij} K_{ij} + D_i N^i   \right) \, ,
\label{eq:MOTS3D}
\end{equation}
where we can also write $q^{ij} K_{ij} = \tr K - K_{ij} N^i N^j$, with $\tr K = h^{ij} K_{ij}$ being the mean curvature of $\Sigma_t$. This shows that $\theta_{l^+}$ depends only on surface quantities $q_{ij}$, $K_{ij}$ and the in-surface normal $N^i$.

\subsection{The MOTSodesic method for locating MOTSs}
\label{sec:MOTSodesic}

While it is straightforward to identify MOTSs defined by constant coordinate parameters in 
static or stationary spacetimes (such 
as the $r=2m$ horizon in Schwarzschild) it is not so simple to identify them in general. For these non-trivial cases, 
identifying a MOTS involves solving the partial differential equation $\theta_\ell = 0$\cite{thornburg, Pook-Kolb:2018igu}.

For an axisymmetric MOTS in an axisymmetric spacetime, however, this PDE can be recast as a pair of coupled ordinary differential equations: the so-called 
MOTSodesic equations\cite{Booth:2021sow,Booth:2022vwo}. 
Their derivation is straightforward for the non-rotating axisymmetric cases that we encounter in this paper. See \cite{Booth:2022vwo} for more complicated cases, including rotating black hole 
spacetimes of arbitrary dimension.

We assume that the slices $(\Sigma_t, h_{ij},  D_i)$
and $K_{ij}$ share the axisymmetry of the spacetime
and so, in a spherical coordinate system 
$(r,\theta, \phi)$, are functions only of $(r,\theta)$. We further assume that both the metric and extrinsic curvature
take  block-diagonal form with no
$\dd r\dd \phi$ or $\dd \theta \dd \phi$ components:
\begin{equation}
h_{ij} = \left[
\begin{array}{ccc}
 \tilde{h}_{rr} & \tilde{h}_{r\theta} & 0 \\
 \tilde{h}_{r\theta} & \tilde{h}_{\theta \theta} & 0 \\
 0 & 0 & \mathcal{R}^2 \\
\end{array} 
\right] \; \;  \mbox{and} \;
K_{ij} = \left[ \begin{array}{ccc}
 \tilde{K}_{rr} & \tilde{K}_{r\theta} & 0 \\
 \tilde{K}_{r\theta} & \tilde{K}_{\theta \theta} & 0 \\
 0 & 0 & K_{\phi \phi}\\
\end{array} 
\right] \; . \label{eq:hblock}
\end{equation}
% With these assumptions, a circle of constant $r$ and $\theta$ has circumference  $C = 2 \pi \mathcal{R}$. Then the 
% smoothness of $\Sigma_t$ requires that
% \begin{equation}
% \lim_{\theta \rightarrow 0} \mathcal{R}= \lim_{\theta \rightarrow \pi} \mathcal{R} = 0 \; . 
% \end{equation}
%
%
As shown  in Figure \ref{fig:MOTSo}, we think of
$(\Sigma_t, h_{ij},  D_i)$ as the rotation of a 
two-dimensional half-plane geometry
$(\tilde{\Sigma}_t, \tilde{h}_{\tilde{\imath}\tilde{\jmath}}, \tilde{D}_{\tilde{\imath}})$ (where $\tilde{\imath}, \tilde{\jmath}$ run over $(r,\theta)$) with respect to the azimuthal angle $\phi$. An arclength parameterized curve $\gamma(s) = (P(s), \Theta(s))$ in $\tilde{\Sigma}$ then rotates into an $(s,\phi)$ parameterized surface $\mathcal{S}(s, \phi) = (P(s), \Theta(s), \phi)$ in $\Sigma$ (Figure \ref{fig:MOTSo} again). 

\begin{figure}
    \includegraphics[width=0.8\linewidth]{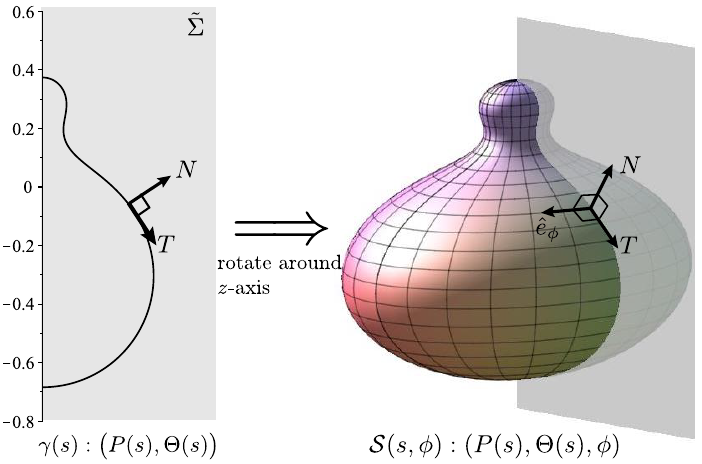}
    \caption{Rotation of a curve $\gamma(s)$ in the $\tilde{\Sigma}$ half-plane into a surface $\mathcal{S}(s,\phi)$ in $\Sigma$.}
    %Given a metric $h_{ij}$ of the form (\ref{eq:hblock}), the normal $N^i$ to $\mathcal{S}$ is $N^i = e^i_{\tilde{\jmath}} N^{\tilde{\jmath}}$, the push-forward to $T \Sigma$ of the unit normal to $\gamma$ in $T \tilde{\Sigma}$.}
    \label{fig:MOTSo}
\end{figure}

Writing derivatives  with respect to arclength with a dot, $T^{\tilde{\imath}} = [\dot{P}, \dot{\Theta}]$ is the unit tangent to $\gamma(s)$ in $\tilde{\Sigma}$ and the left-hand unit normal is
\begin{equation}   
N^{\tilde{\imath}} = \tilde{\varepsilon}^{\tilde{\imath}}_{\phantom{\imath} \tilde{\jmath}} T^{\tilde{\jmath}} \; ,  \label{eq:N}
\end{equation}
where $\varepsilon_{\tilde{\imath} \tilde{\jmath}}$ is the 
Levi-Civita tensor on $\tilde{\Sigma}$. Along $\gamma$ these form a basis for $T \tilde{\Sigma}$ and, as for any arclength
parameterized curve,
\begin{equation}
    T^{\tilde{\jmath}} \tilde{D}_{\tilde{\jmath}} T^{\tilde{\imath}} = \kappa N^{\tilde{\imath}} \label{eq:TdT}
\end{equation}
for some curvature function $\kappa$. Given such a 
curve, we can calculate 
\begin{equation}
    \kappa = N_{\tilde{\imath}} T^{\tilde{\jmath}} \tilde{D}_{\tilde{\jmath}} T^{\tilde{\imath}} \; . 
\end{equation}

Turning to $\mathcal{S}$, at any point on the 
surface the tangent space is spanned by $(T,N, \hat{e}_\phi)$: the 
push-forwards of the unit tangent and normal vectors 
from $T \tilde{\Sigma}$ along with the unit vector pointing in the rotational 
symmetry direction. Then 
\begin{equation}
    h^{ij} = T^i T^j + N^i N^j + \hat{e}_\phi^i \hat{e}_\phi^j
\end{equation}
and so the mean curvature of $\mathcal{S}$ in 
$\Sigma$ is 
\begin{equation}
    k_N = D_i N^i = h^{ij} D_i N_j = (T^i T^j + \hat{e}_\phi^i \hat{e}_\phi^j) D_i N_j = - \kappa - N^i a^{\hat{e}_\phi}_i\, , 
\end{equation}
where $a^{\hat{e}_\phi}_i$ is the acceleration of $\hat{e}_\phi$, which by direct calculation can be shown to be equal to
\begin{equation}
  a^{\hat{e}_\phi}_i =  -   D_i \ln \mathcal{R} \; . \label{eq:acc}
\end{equation}
Hence, the MOTS condition $\theta_\ell = 0$ can be transformed with the help of (\ref{eq:MOTS3D}) into an expression for $\kappa$ in terms of 
$h_{ij}$ and $K_{ij}$:
\begin{equation}
    \theta_\ell = 0 \quad \Longrightarrow \quad k_N + q^{ij} K_{ij}  = 0 \quad \Longrightarrow \quad 
    \kappa =  -   N^i a_i^{\hat{e}_\phi}  + 
    q^{ij} K_{ij} \label{eq:kappa} \; . 
\end{equation}
Then (\ref{eq:N}), (\ref{eq:TdT}) and (\ref{eq:kappa}) together define a pair
of coupled, second-order, non-linear differential equations for $(P(s),\Theta(s))$ whose solutions 
rotate into surfaces with $\theta_{l^+}=0$:
\begin{align}
    \ddot{P} & = -\tilde{\Gamma}^r_{\tilde{\imath} \tilde{\jmath}} T^{\tilde{\imath}} T^{\tilde{\jmath}} +  \kappa 
    N^{r} \label{eq:PT}\\
\ddot{\Theta} & = -\tilde{\Gamma}^\theta_{\tilde{\imath} \tilde{\jmath}} T^{\tilde{\imath}} T^{\tilde{\jmath}}
+ \kappa  
N^{\theta} \; . \nonumber 
\end{align}
These are the MOTSodesic equations. 

These equations require boundary conditions: 
to be a MOTS, a surface $\mathcal{S}$ must be closed, in addition to having $\theta_{l^+}=0$. This means a MOTSodesic curve will only generate a MOTS if it intersects the $z$-axis orthogonally (and so avoids having a conical singularity there) as in the example of Figure \ref{fig:MOTSo}, or is a closed curve in the half-plane $\tilde\Sigma$ and hence rotates
into a torus (for examples of toroidal MOTS see \cite{Pook-Kolb:2021jpd} or \cite{Sievers:2023zng}). 
In this paper, all MOTS will be deformed from initially spherical ones and so will intersect the $z$-axis.
% However, there is a technical complication in that degeneracies at $\theta = 0, \pi$ mean that (\ref{eq:PT}) are not well-defined there. We address this problem in the next subsection. 
The numerical procedure we used to solve the MOTSodesic equations is presented in Appendix \ref{app:MOTSo}.

\subsection{The stability operator}
\label{sec:slice}

Having defined MOTSs and shown how they can be found numerically, we now study their deformations. In Appendix \ref{sec:stabop} we will define and study a generalized stability operator that measures the variation in $\theta_\ell$ when $\cS$ undergoes an arbitrary normal deformation  within the spacetime $M$. Here we restrict our attention to the case that is most relevant in most applications: deformations that lie in a time slice $(\Sigma_t, h_{ij},  K_{ij} )$, as discussed in Section \ref{sec:AH}. As noted above, the stability operator (and resulting stability classification of $\cS$) depends on the choice of slice; see Appendix~\ref{sec:stabop} for details.

Letting $N$ denote the outward unit normal to $\cS$ in $\Sigma_t$, and $\ell$ a choice of future-pointing outward null vector along $\cS$, the stability operator, $L_{\mathcal S}(\ell, N)$, is given by
\begin{equation}
    L_{\mathcal S}(\ell, N) \psi := \delta_{\psi N} \theta_\ell
\end{equation}
for any function $\psi$ on $\cS$. Writing $\ell = e^\mu l^+$ for some function $\mu$, we obtain (from the general formula (\ref{eq:var}) in Appendix \ref{sec:stabop} with $A=\frac12$ and $B=1$) 
\begin{equation}
    L_{\cS}(\ell,N) = e^\mu \Big( 
    -  \mathcal{D}^2_{l^+} \label{eq:varSigma}
   +  \mathcal{K}  -  \frac{1}{2}  \|\sigma_{l^+}\|^2 - G_{l^+ l^-}  - \frac{1}{2}  G_{l^+ l^+} \,    \Big) \psi  \; ,
\end{equation}
where $\mathcal{K}$ is the Gauss curvature of $\mathcal{S}$, 
$G_{l^+ l^-} = G_{ab} l_+^a l_-^b$ and $G_{l^+ l^+} = G_{ab} l_+^a l_+^b$ for the Einstein tensor $G_{ab}$, and $\|\sigma_{l^+}\|^2 = \sigma_{AB}^{l^+} \sigma^{AB}_{l^+}$
for the outward null shear
\begin{equation} 
\sigma^{l_+}_{AB} = k^{l_+}_{AB} - \frac{1}{2}  \theta_{l^+} q_{AB} = 
e_A^a e_B^b \nabla_A l^+_B - \frac{1}{2} \theta_{l^+} q_{AB} \; . 
\end{equation}
For the derivative terms we have
\begin{equation}
    \mathcal{D}_{l^+}^2 F  = (\mathcal{D}^A - \omega_{l^+}^A)(\mathcal{D}_A - {\omega}_A^{l^+}) F\;,
\end{equation} 
where $\mathcal{D}$ is the covariant derivative on $\mathcal{S}$ and ${\omega}^{l_+}_A$ is the connection on the normal bundle to $\mathcal{S}$ defined  in (\ref{eq:omega}).

From (\ref{eq:varSigma}) we see that $L_\cS(\ell,N)$ is a second-order linear elliptic differential operator. Since $\cS$ is compact, the spectrum of $L_{\mathcal S}(\ell,N)$ consists of isolated eigenvalues of finite multiplicity. While the stability operator is not necessarily self-adjoint, and so may have complex eigenvalues, the principal eigenvalue $\lambda_o$ (the eigenvalue with the smallest real part) is real and simple, so the corresponding eigenfunction is unique up to a scalar multiple. Moreover, any such eigenfunction is nowhere vanishing, so it can be chosen to be positive everywhere \cite{Andersson:2005}. We then say that $\cS$ is \emph{stable} if $\lambda_o \geq 0$ and unstable if $\lambda_o < 0$. The stable case can be further subdivided into \emph{strictly stable} ($\lambda_o > 0$) and \emph{marginally stable} ($\lambda_o = 0$).

%In Appendix \ref{sec:stabop} we will see that the the operator $L_{\cS}(\ell,N)$ depends on the scaling of $\ell$ but its eigenvalues do not. Moreover, if we rescale $N$ (so it no longer has unit length) the eigenvalues will change, but the number of negative and zero eigenvalues will not. So, if we are only interested in the signs of the eigenvalues, the scaling of $\ell$ and $N$ is irrelevant. In Section \ref{sec:HGSO} we will see that this freedom is useful in numerical calculations, for instance when the stability operator has very large coefficients, rendering the calculation of its eigenvalues numerically unstable. By scaling $\ell$ and/or $N$ we obtain an operator that still encodes the stability of $\cS$ but has more manageable coefficients.

We will see in Appendix \ref{sec:rescaling} that the stability does not depend on the scaling of $\ell$, so we can simplify the form of (\ref{eq:varSigma}) by setting $\mu = 0$, which means $\ell = l_+$.  Alternatively, $\mu$ could be chosen to cancel off complicated leading factors that would cause numerical problems. For example, in Section \ref{sec:HGSO} we choose a $\mu$ to eliminate an inconvenient leading exponent that multiplies the entire operator.

%\ivan{Haven't defined non-rotating at this point...}

\subsection{Stability operator for axisymmetric MOTSs}
\label{sec:eigcalc_axi}

In the examples of later sections we will need to calculate stability spectra for MOTSs generated from MOTSodesics, as in Section \ref{sec:MOTSodesic}. Relative to the spherical-type coordinates $(r,\theta, \phi)$ on $\Sigma$, such a MOTS can be parameterized by arclength along 
the meridians $s \in (0,\smax)$  and $\varphi = (- \pi, \pi) $ as 
\begin{equation}
    r = P(s) \, , \quad \theta = \Theta(s) \quad  \mbox{and} \quad \phi = \varphi \, , 
\end{equation}
where $P(s)$ and $\Theta(s)$ are solutions of the MOTSodesic equations (\ref{eq:PT}). The induced metric on this surface is 
\begin{equation}
    q_{AB} \dd x^A \dd x^B = \dd s^2 + \mathcal{R}(s)^2 \dd \varphi^2 
\end{equation}
where $\mathcal{R}^2 (s)  = h_{\phi \phi} \big(
P(s),\Theta(s) \big)$. 
The unit tangent vectors to the surface are 
\begin{equation}
    T = \frac{\partial}{\partial s} \rightarrow  \dot{P} \frac{\partial}{\partial r} + \dot{\Theta} \frac{\partial}{\partial \theta} \quad \mbox{and} \quad \hat{e}_\phi = \frac{1}{\mathcal{R}} \frac{\partial}{\partial \varphi} \rightarrow   \frac{1}{\mathcal{R}} \frac{\partial}{\partial \phi} \, , 
\end{equation}
where dots indicate derivatives with respect to $s$ and arrows indicate the push-forward into the 
full tangent space of $\Sigma$. The outward-oriented unit normal is then 
\begin{equation}
    N^{\tilde{\imath}} \frac{\partial}{\partial x^{\tilde{\imath}}}  = \sqrt{\tilde{h}} \left( 
    \left(\tilde{h}_{r \theta} \dot{P} + \tilde{h}_{\theta \theta} \dot{\Theta} \right) \frac{\partial}{\partial r} - 
    \left(\tilde{h}_{r r} \dot{P}   + \tilde{h}_{r \theta} \dot{\Theta} \right) \frac{\partial}{\partial \theta}
    \right) 
\end{equation}
where $\tilde{h} = \tilde{h}_{rr} \tilde{h}_{\theta \theta} - \tilde{h}^2_{r \theta}$. This was raised from the normal one-form $N_{\tilde{\imath}} \dd x^{\tilde{\imath}} = \sqrt{\tilde{h}} \big( \dot{\Theta} \dd r - \dot{P} \dd \theta \big)$.
% \begin{equation}
%     N_{\tilde{\imath}} \dd x^{\tilde{\imath}} = \sqrt{\tilde{h}} \left( \dot{\Theta} \dd r - \dot{P} \dd \theta \right) \;.
% \end{equation}

Then we need to calculate the components of 
(\ref{eq:varSigma}). Choosing $\ell=l^+$, so that $\mu=0$, we
expand this as 
\begin{equation}
\label{eq:axiL}
   L_{\mathcal{S}}(l^+, N)   = 
    -\mathcal{D}^2 + 2 \omega_{l^+}^A \mathcal{D}_A + \Big( D_A \omega_{l^+}^A- \| \omega_{l^+} \|^2 +   \mathcal{K}
    \ - \frac{1}{2}  \|\sigma_{l^+}\|^2 - G_{\hat{u} \hat{u}}   - G_{\hat{u} N} \Big)  \; .
\end{equation}
% \begin{align}
%    L_{\mathcal{S}}(\ell, N)   =  e^\mu \bigg(
%     -\mathcal{D}^2 + 2 \omega_{l^+}^A \mathcal{D}_A + \Big(& D_A \omega_{l^+}^A- \| \omega_{l^+} \|^2 +   \mathcal{K}
%    \label{eq:axiL}\\
%    & \ - \frac{1}{2}  \|\sigma_{l^+}\|^2 - G_{\hat{u} \hat{u}}   - G_{\hat{u} N} \Big) \bigg)  \; . \nonumber 
% \end{align}
Taking the null normal scaling (\ref{eq:lplm}) with $\beta=0$, from
%(\ref{eq:tom}) and 
the block diagonal form of $K_{ij}$ we get
\begin{equation}
    \omega_{l^+} = \tilde{K}_{TN} \dd s 
    %= (K_{rr} \dot{P}^2 + 2 K_{r \theta} \dot{P} \dot{\Theta} + K_{\theta \theta} \dot{\Theta}^2  )  \dd s
    \, ,   \label{eq:Ktom}
\end{equation}
where we have abbreviated $\tilde{K}_{TN} = \tilde{K}_{\tilde{
    \imath} \tilde{\jmath}} T^{\tilde{
    \imath}} N^{\tilde{
    \imath}}$. The components of the extrinsic 
curvature are also functions of $s$: 
$K_{ij} = K_{ij} (P(s), \Theta(s))$. As 
such, $\omega_{l^+}$ is a total 
derivative and so $L_{\mathcal{S}}(\ell, N)$ has all real eigenvalues and there is a gauge choice for which it is self-adjoint (Proposition \ref{prop:rotation}). However,  
 (\ref{eq:Ktom}) is generally not integrable in 
closed form and so it is computationally inconvenient
to transform $L_{\mathcal{S}} (\ell, N)$ into its self-adjoint form. We thus stay with our current gauge choice and work with this non-zero $\omega_{l^+}$.

We find expressions for the various components of this equation.
% First, 
% the induced metric on the MOTS $\mathcal{S}$ can be written as
% \begin{equation}
%     q_{AB} \dd x^A \dd x^B = \dd s^2 + \mathcal{R}(s)^2 \dd \phi^2 
% \end{equation}
% where $s$ is the arclength along curves of constant $\phi$ and $\mathcal{R} (s)  = \mathcal{R} (P(s), \Theta(s))$.  
For a function $\psi(s, \phi)$, the 
Laplacian is
\begin{equation}
    \mathcal{D}^2 \psi = \psi_{,ss} + \frac{\mathcal{R}_{,s} \psi_{,s}}{\mathcal{R}} + \frac{\psi_{,\phi \phi}}{\mathcal{R}^2}
\end{equation}
where the comma subscripts indicate partial derivatives. 
The other intrinsic term, 
the Gauss curvature, takes the form
\begin{equation}
    \mathcal{K} = - \frac{\mathcal{R}_{,ss}}{\mathcal{R}} \; . 
\end{equation}
Further, from  $\theta_{l^+} = 0$ it is straightforward to show that
\begin{equation}
    \|\sigma_{l^+}\|^2 = 2 (K_{\hat{e}_\phi \hat{e}_\phi} - a_N^{\hat{e}_\phi})^2 \, ,
\end{equation}
where subscripts in the two terms on the right-hand side are again indicating contractions. 

Then the 
stability operator for these 
axisymmetric surfaces becomes
\begin{equation}
     L_{\mathcal{S}} (l^+, N)  \psi  =  
       -\psi_{,ss}  + \mathcal{G}(s) \psi_{,s} + \mathcal{H}(s) \psi
       - \frac{\psi_{,\phi \phi}}{\mathcal{R}(s)^2} \label{eq:LSlN}
\end{equation}
where 
\begin{align}
    \mathcal{G}(s)  = & -\frac{\mathcal{R}_{,s}}{\mathcal{R}} + 2 \tilde{K}_{TN} \,,   \label{eq:cG} \\
    \mathcal{H}(s)  = & \frac{\dd}{\dd s} \tilde{K}_{TN} + \frac{\mathcal{R}_{,s}}{\mathcal{R}} \tilde{K}_{TN} - \tilde{K}_{TN}^2 -\frac{\mathcal{R}_{,ss}}{\mathcal{R}} \label{eq:cH} \\ & \quad - 2 \left(\frac{K_{\phi \phi}}{\mathcal{R}^2}  - \frac{N^{\tilde{\imath}} D_{\tilde{\imath}} \mathcal{R}}{\mathcal{R}}  \right)^2 - G_{\hat{u} \hat{u}}   - G_{\hat{u} N}  \; . \nonumber 
\end{align}

It was shown in \cite{Andersson:2007fh} that the principal eigenfunction of an axisymmetric $L_{\mathcal{S}} (\ell, N)$ is also 
axisymmetric. Hence if we are only interested in whether or not $\mathcal{S}$ is stable, it is  sufficient to study the eigenvalues of axisymmetric eigenfunctions, which are eigenfunctions of the one-dimensional operator
\begin{equation}
     \tilde{L} = 
      -\frac{\dd^2}{\dd s^2}   + \mathcal{G}(s) \frac{\dd}{\dd s} + \mathcal{H}(s) \label{eq:L1d} \; . 
\end{equation}
Away from spherical symmetry
in the examples of Sections
\ref{sec:axiRN} and \ref{sec:weyl},
we generally work with this $\tilde{L}$ to determine stability. 

All of the examples considered in this paper are axisymmetric and we classify 
% When $\cS$ is non-rotating, there is a choice of $\ell$ for which $L_\cS(\ell,N)$ is self-adjoint and hence has real eigenvalues; see Proposition \ref{prop:rotation} for details. Restricting our attention to axisymmetric eigenfunctions of an
% axisymmetric non-rotating MOTS,
our  MOTSs by labelling them as
\emph{$(n_0, n_-)$-unstable} if  $\tilde{L}$ has $n_0$ vanishing eigenvalues and $n_-$ negative eigenvalues. The full operator $L_\mathcal{S} (l^+, N)$ will also have non-axisymmetric eigenfunctions. 
 However, for our purposes the axisymmetric eigenfunctions are the most important and so we focus on them. 

% In the examples below we will restrict our attention to $\tilde{L}$. In addition to determining stability it will also be sufficient to provide examples of the various types of bifurcations. 

% Even with this simplification,  explicitly calculating
% the spectrum  as  we did in Section \ref{sec:sphere_stab}
% is not possible: the eigenfunctions do not have a closed algebraic form. 
% Instead we must approximate the spectrum numerically, as described in Appendix
% \ref{app:spectra}. 

\subsection{MOTSs as boundaries between trapped and untrapped regions}
\label{sec:inout}

As discussed in the introduction, a MOTS $\mathcal{S}$ in a time slice $\Sigma$ is a potential boundary between outer trapped and untrapped regions. We now examine this in more detail. All results are at most slight extensions of Proposition 1 of \cite{Andersson:2005}: the core argument is the same but our applications to unstable MOTSs were not explicitly discussed in that paper.

We restrict our attention to embedded and orientable MOTS so that $\mathcal{S}$ can unambiguously be assigned an inside and an outside. We will call such a MOTS \emph{two-sided}. Using the principal eigenfunction 
$\psi_o>0$  we can construct
a foliation of a neighbourhood of $\mathcal{S}$ as follows. Let $x^a_\mathcal{S} (\vartheta^A)$  be a coordinate parameterization of $\mathcal{S}$ for some set of surface coordinates $\vartheta^A$. Define surfaces $\mathcal{S}_\sigma$ that are parameterized by
       \begin{equation}
    x_\sigma^a (\vartheta^A) = x^a_\mathcal{S} (\vartheta^A) + \sigma \psi_o (\vartheta^A) N^a(\vartheta^A) \; . 
\end{equation}
Then $\{ \sigma, \vartheta^A \}$ defines a coordinate system on some neighbourhood of $\mathcal{S}$ in $\Sigma$, foliated by surfaces of constant $\sigma$, which we denote $\mathcal{S}_\sigma$. We work with this system for the rest of this subsection.
%We split $\mathcal{N}-\mathcal{S}$ into two pieces: the outside $\mathcal{N}_+$ with $\sigma >0$ and inside $\mathcal{N}_-$ with $\sigma <0$.

For small $\sigma$ the outward null expansion of $\mathcal{S}_\sigma$ is
\begin{equation}
    \theta_\sigma^{l_+} \approx \sigma \delta_{\psi_o N} \theta^{l^+}_0 = \sigma L_\mathcal{S} (l_+, N) \psi_o = \sigma \lambda_o \psi_o \; . 
\end{equation}
Thus, for some small neighbourhood $- \epsilon < \sigma < \epsilon$ around a strictly stable
($\lambda_o > 0$) $\cS$, the expansion $\theta^{l_+}_\sigma$ of $\mathcal{S}_\sigma$
is positive for $\sigma >0$ and negative for 
$\sigma < 0$. If $\mathcal{S}$ is unstable ($\lambda_o < 0$) then this relationship is reversed. In either case we let $\mathcal{N}$ be the union of the $\cS_\sigma$ for $- \epsilon < \sigma < \epsilon$, with the outside ($\sigma > 0$) $\mathcal{N}_+$ and inside $\mathcal{N}_-$ ($\sigma < 0$)  as illustrated in Figure \ref{fig:stability}. 

Then we have the following result. 

\begin{lem}
\label{lem:maxmin}
Let $\mathcal{S}$ be a two-sided MOTS with a neighbourhood $\mathcal{N}$ foliated by $\{S_\sigma\}$ as above, and let $\tilde\cS$ be a closed surface in $\mathcal N$.
\begin{enumerate}
    \item If $\tilde\cS$ is weakly outer trapped, then $\sigma|_{\tilde\cS}$ cannot have a maximum $\sigma_*$ such that $\theta^{l_+}_{\sigma_*} >0$.
    \item If $\tilde\cS$ is weakly outer untrapped, then $\sigma|_{\tilde\cS}$ cannot have a minimum $\sigma_*$ such that $\theta^{l_+}_{\sigma_*} < 0$.
    %\item No weakly outer untrapped surface can have a minimum value of $\sigma$ in a region where $\theta^{l_+}_\sigma <0$.
    \end{enumerate}
\end{lem}

The first statement is a part of Proposition 1 of \cite{Andersson:2005}, where it is proved using the strong maximum principle. Here we give a more direct proof, as in Theorem 7.1 of \cite{Andersson:2007fh}.

\begin{proof}
    Let $\tilde{\mathcal{S}}$ be parameterized as $\sigma = H(\vartheta^A)$. The outward-oriented unit normal to $\tilde{\mathcal{S}}$ in $\Sigma$ is 
    \begin{equation}
        N = \alpha (\dd \sigma - \partial_A H \dd \vartheta^A)
    \end{equation}
    for some positive normalization factor $\alpha$, thus the outward null expansion is
    \begin{align}
        \theta^{l_+}_{\tilde{\mathcal{S}}} 
        %& = k_u + k_N \nonumber\\
        & = q^{ij} K_{ij} + q^{ij} D_i N_j \nonumber \\
        & = \left( q^{ij} K_{ij} - q^{ij} \Gamma_{ij}^k N_k \right) - \alpha q^{AB} \partial_A \partial_B H \;. 
    \end{align}
    At any extremum of $H$ on $\cS$ we have $\partial_A H = 0$ and so
    \begin{equation}
        \theta^{l_+}_{\tilde{\mathcal{S}}} = \theta^{l_+}_{\sigma_*} - \alpha q^{AB} \partial_A \partial_B H \; .  
    \end{equation}
    
    If $\tilde\cS$ is weakly outer trapped and $\sigma_*$ is the maximum value of $H$, then $q^{AB} \partial_A \partial_B H \leq 0$ and so
    \begin{equation}
        0 \geq \theta^{l_+}_{\tilde{\mathcal{S}}} \geq \theta^{l_+}_{\sigma_*}
    \end{equation}
    holds at the point where the maximum is achieved.
    Therefore it is not the case that $\theta^{l_+}_{\sigma_*}>0$.

    Similarly, if $\tilde\cS$ is weakly outer untrapped and $\sigma_*$ is the minimum value of $H$, then we have $q^{AB} \partial_A \partial_B H \geq 0$ and so
    \begin{equation}
        0 \leq \theta^{l_+}_{\tilde{\mathcal{S}}} \leq \theta^{l_+}_{\sigma_*}
    \end{equation}
    holds at the point where the minimum is achieved, thus it is not the case that $\theta^{l_+}_{\sigma_*}< 0$.
\end{proof}

For a strictly stable MOTS we then recover Proposition 1(i) from \cite{Andersson:2005}.

\begin{prop}
\label{prop:barrier-s}
    Let $\mathcal{S}$ be a two-sided strictly stable MOTS with a neighbourhood $\mathcal{N}$ as constructed above, and let $\tilde\cS$ be a closed surface in $\mathcal N$.
    \begin{enumerate}
        \item If $\tilde\cS$ is weakly outer trapped, then it must be contained in $\mathcal{N}_- \cup \mathcal{S}$.
        \item If $\tilde\cS$ is weakly outer untrapped, then it must be contained in $\mathcal{N}_+ \cup \mathcal{S}$.
    \end{enumerate}
\end{prop}

\begin{proof}
Strict stability of $\cS$ implies $\theta^{l_+}_{\sigma} >0$ for $\sigma > 0$ and $\theta^{l_+}_{\sigma} < 0$ for $\sigma < 0$. If $\tilde\cS$ is weakly outer trapped, the maximum value of $\sigma|_{\tilde\cS}$ must therefore satisfy $\sigma_* \leq 0$, by Lemma \ref{lem:maxmin}(i). This means $\sigma \leq 0$ everywhere on $\tilde\cS$ and so $\tilde{\mathcal{S}} \in \mathcal{N}_- \cup \mathcal{S}$. Similarly, if $\tilde\cS$ is weakly outer untrapped the minimum must satisfy $\sigma_* \geq 0$, which means $\sigma \geq 0$ on $\tilde\cS$ and so $\tilde{\mathcal{S}} \in \mathcal{N}_+ \cup \mathcal{S}$.
 \end{proof}
 
This means that a strictly stable $\mathcal{S}$ is a barrier between a region $\mathcal{N}_+$ containing outer untrapped surfaces and a region $\mathcal{N}_-$ containing outer trapped surfaces. There are a few important points to keep in mind. 
\begin{enumerate}
    \item Even though $\mathcal{N}$ is only guaranteed to exist as a small neighbourhood, in some cases it will cover large regions of the spacetime. For Schwarzschild in the Painlev\'e--Gullstrand slicing, $\mathcal{N}_+$ is all points with $r>2m$ while $\mathcal{N}_-$ is all points with $r<2m$. For Reissner--Nordstr\"om with the same slicing, $\mathcal{N}_+$ is all points outside the outer horizon and
   $\mathcal{N}_-$ is all points between the inner and outer horizons. 
   \item It is possible for weakly outer (un)trapped MOTS to traverse $\mathcal N$: if they have no local extrema in $\mathcal{N}$ then they cannot violate Lemma \ref{lem:maxmin}. Similarly, a weakly outer trapped surface that entered $\mathcal{N}_+$ from the outside and had a minimum in that region would not violate the proposition, and likewise for a weakly outer untrapped surface that entered $\mathcal N_-$ from the inside. $\cS$ is only a barrier for surfaces that are fully contained in $\mathcal{N}$.
   \item On the other hand, we \emph{can} use Lemma \ref{lem:maxmin} to rule out the possiblity of a weakly outer trapped surface that enters $\mathcal N$ from the inside and has a maximum in $\mathcal N_+$, or a weakly outer untrapped surface that enters $\mathcal N$ from the outside and has a minimum in $\mathcal N_-$.
\end{enumerate}

In the unstable case we have the following analogue of Proposition \ref{prop:barrier-s}.

\begin{prop}
\label{prop:barrier-us}
    Let $\mathcal{S}$ be a two-sided unstable MOTS with a neighbourhood $\mathcal{N}$ as constructed above, and let $\tilde\cS$ be a closed surface in $\mathcal N$.
    \begin{enumerate}
        \item If $\tilde\cS$ is weakly outer trapped, then it is not contained entirely in $\mathcal N_-$.
        \item If $\tilde\cS$ is weakly outer untrapped, then it is not contained entirely in $\mathcal N_+$.
    \end{enumerate}
\end{prop}

\begin{proof}
Instability of $\cS$ implies $\theta^{l_+}_{\sigma} < 0$ for $\sigma > 0$ and $\theta^{l_+}_{\sigma} > 0$ for $\sigma < 0$. If $\tilde\cS$ is weakly outer trapped, the maximum value of $\sigma|_{\tilde\cS}$ must therefore satisfy $\sigma_* \geq 0$, by Lemma \ref{lem:maxmin}(i). This means $\sigma \geq 0$ somewhere on $\tilde\cS$ and so $\tilde{\mathcal{S}}$ is not contained in $\mathcal{N}_-$. Similarly, if $\tilde\cS$ is weakly outer untrapped the minimum must satisfy $\sigma_* \leq 0$, which means $\sigma \leq 0$ somewhere on $\tilde\cS$ and so $\tilde{\mathcal{S}}$ is not contained in $\mathcal{N}_+$.
 \end{proof}

The constraints on $\tilde\cS$ in the unstable case are weaker than in the strictly stable case. While no weakly outer untrapped surface can be contained in $\mathcal{N}_+$ and no weakly outer trapped surface can be contained in $\mathcal{N}_-$, it is possible for such surfaces to straddle $\mathcal{S}$ and so have a maximum value of $\sigma$ in $\mathcal{N}_+$ and a minimum in $\mathcal{N}_-$.

It is easy to see that such surfaces exist in the vicinity of an unstable $\mathcal{S}$ with a vanishing non-principal eigenvalue $\lambda_k$. The corresponding eigenfunction $\psi_k$ will take both positive and negative values of $\mathcal{S}$ (see, for instance, Lemma 4.2 in \cite{BCMsymmetry}). A surface $\tilde{\mathcal{S}}$ defined by
%\begin{equation}
%    x_{\tilde{\mathcal{S}}}^a (\vartheta^A) = x_{\mathal{S}}^a (\vartheta^A) + \left( \alpha \psi_o(\vartheta^A) + b \psi_k (\vartheta^A) \right) N^a (\vartheta^A)
%\end{equation}
%
\begin{equation}
    x_{\tilde{\mathcal{S}}}^a (\vartheta^A) = x_{\mathcal{S}}^a + \Big( \epsilon_1 \psi_o(\vartheta^A) + \epsilon_2 \psi_k(\vartheta^A) \Big) N^a(\vartheta^A)
\end{equation}
% \begin{equation}
%     \sigma_{\tilde{\mathcal{S}}} = \left( \epsilon_1 \psi_o + \epsilon_2 \psi_k \right) N^\sigma 
% \end{equation}
% where $N^\sigma$ is the component of the normal in the $\sigma$ direction, 
has expansion
\begin{equation}
    \theta^{l_+}_{\tilde{\mathcal{S}}} \approx 
    \epsilon_1 L_{\mathcal{S}} (l_+, N) \psi_o + \epsilon_2
L_{\mathcal{S}} (l_+, N) \psi_k = \epsilon_1 \lambda_o \psi_o \, ,   \end{equation}
if $\epsilon_1, \epsilon_2$ are sufficiently small. Choosing the relative sizes of $\epsilon_1$ and $\epsilon_2$ so that $\tilde{\mathcal{S}}$ straddles $\mathcal{S}$, it is clear that $\tilde\cS$ can be made outer trapped, untrapped or marginally trapped with an appropriate choice of $\epsilon_1$. 
We will see many examples of such surfaces in the later sections, as they appear in most bifurcations.

\section{MOTS deformations and bifurcations}
\label{sec:MOTSevolve}

As was first discussed in \cite{Andersson:2005, Andersson:2007fh} with respect to time evolutions, the spectrum of the stability operator constrains the possible evolutions and deformations of a MOTS. 

\subsection{Deformations of MOTSs}
\label{sec:deformations}

For a time slice $(\Sigma, h_{ij}, K_{ij})$, we consider a one-parameter deformation of the geometry:
\begin{equation}
    h_{ij} = h_{ij}(\alpha) \,, \qquad K_{ij} = K_{ij}(\alpha) \,,
\end{equation}
where $\alpha$ is a parameter.
This parameter could correspond to time evolution, as in \cite{Andersson:2005, Andersson:2007fh}, or 
other coordinate translation, but it could also 
involve changing physical parameters such as the charge of a Reissner--Nordstr\"om spacetime (Section \ref{sec:axiRN}), combination of charge, mass and cosmological constant for Reissner--Nordstr\"om--de Sitter (Section \ref{sec:sphere_stab}) or a Weyl parameter (Section \ref{sec:weyl}).

Let $\mathcal S$ be a MOTS in $\Sigma$. Given a small parameter change $\Delta\alpha$, we attempt to deform $\cS$ by a corresponding amount $\psi N^i \Delta\alpha$ in the normal direction such that it remains a MOTS. For this to work the combined deformation must vanish,
\begin{equation}
    \delta \theta_\ell = \delta_{\alpha} \theta_\ell + \delta_{\psi N} \theta_\ell = 0 \, , \label{eq:MOTSvar}
\end{equation}
where $\delta_\alpha \theta_\ell = \frac{\mathrm{d} \theta_\ell}{\mathrm{d} \alpha}$ measures the change of expansion for the fixed surface $\cS$ as $\alpha$ is varied. The function $\psi$ therefore must solve the equation
\begin{equation}
    L_\mathcal{S} (\ell, N) \psi = -\delta_\alpha \theta_\ell \;.
    \label{eq:psieq}
\end{equation}
This has a unique solution 
% \begin{equation}
%     \psi = - L^{-1}_\mathcal{S} (\ell, X) [\delta_\alpha \theta_\ell] \label{eq:MOTSvarsol}
% \end{equation}
if $L_\mathcal{S} (\ell, N)$ is invertible, which suggests that invertibility is a sufficient condition for the required deformation of $\cS$ to exist.

This intuition can be made precise using the implicit function theorem: if $L_\mathcal{S} (\ell, N)$ is invertible, then $\cS$ can be deformed so that it remains a MOTS with respect to the $\alpha$-deformed geometry \cite{Andersson:2005}. Since $L_\mathcal{S} (\ell, N)$ has discrete spectrum, it will be invertible if and only if it has no vanishing eigenvalues. In particular, if the principal eigenvalue is positive then all other eigenvalues also have positive real part and so none vanish. 

This means a bifurcation can only occur at a MOTS $\cS$ for which $L_\mathcal{S} (\ell, N)$ is \emph{not} invertible. In this case the equation (\ref{eq:psieq}) could have multiple solutions, or none at all. Correspondingly, $\cS$ may split into multiple MOTSs, or disappear completely, as $\alpha$ changes.

Bifurcation theory provides tools for  
analyzing such problems and calculating the number and behaviour of nearby MOTSs. The application of such tools to MOTSs is investigated systematically in \cite{BusseyCoxKunduri} (see also \cite{Andersson:2008up,XiePrasad26} for special cases and \cite{Jost} for an application to minimal surfaces). In this paper we provide a number of examples, both analytic and numerical, that illustrate the most common types of bifurcations in physically reasonable spacetimes.

%We next review some fundamentals of bifurcation theory, before returning to MOTSs in Section \ref{sec:BifTheory}. For more details on bifurcation theory see, for example, \cite{strogatz2015} as an introduction or \cite{kuznetsov1998elements,wiggins2006} for a more advanced discussion.

\subsection{Bifurcations of MOTSs}
\label{sec:BifTheory}

In Appendix~\ref{app:bif} we review finite-dimensional bifurcation theory for a family
\begin{equation}
    \dot{y}^i = F^i(y^j,\alpha) 
    \label{eq:Fi}
\end{equation}
of $n$-dimensional dynamical systems depending on a parameter $\alpha$. The discussion there (or in any textbook on bifurcation theory) is reminiscent of the statements in Section~\ref{sec:deformations}. 
In particular, varying $\alpha$
generates a smooth curve of fixed points  of (\ref{eq:Fi}) if the linearization is invertible. Therefore non-invertibility is a necessary condition for a bifurcation. This is not a coincidence as in both cases we are addressing essentially the same problem; this connection will be made explicit in \cite{BusseyCoxKunduri}.

\begin{table}
    \centering
    \begin{tabular}{|l||l|l|}
    \hline
      Concept   & Finite-dimensional & MOTS\\
      \hline \hline
       phase space  & coordinates $y^i$ & parameterized surfaces $x^i(\vartheta^A)$\\
       \hline
       fixed points  & $F^j(y_o^i, \alpha_o) = 0$ & $\theta_\ell ( x_\mathcal{S}^i (\vartheta^A))=0$ \\
       \hline
       linearization  & $y_o^i \rightarrow y_o^i + \Delta y^i $ & $x_S^i \rightarrow  x_\mathcal{S}^i + \epsilon \psi N^i$\\
       \hline
        stability & all eigenvalues of $[DF]^i_j = f^i_{,j}$
         & all eigenvalues of $L_\mathcal{S}(\ell,N)$ \\
        condition & have \textbf{negative real part} & have \textbf{positive real part }\\
         \hline
        bifurcation & $[DF]^i_j$ has a vanishing  & $L_\mathcal{S} (\ell,N)$ has a vanishing \\
        condition & eigenvalue $\lambda_A$ & eigenvalue $\lambda_A$\\
         \hline
        ``direction'' of & eigenvector $v^i_A$  & eigenfunction $\psi_A$  \\ 
        bifurcation & & and normal: $\psi_A N^i$ \\
        \hline
    \end{tabular}
    \caption{Finite-dimensional dynamical system versus MOTS bifurcation theory.  }
    \label{tab:placeholder}
\end{table}

The correspondence between the two problems is summarized in Table \ref{tab:placeholder}. On the MOTS side we consider the phase space of all closed surfaces in $\Sigma$. To those surfaces we can apply the outward expansion operator $\theta_\ell$ and MOTSs are surfaces on which it vanishes. These are the equivalent of fixed points in this infinite-dimensional phase space. 

For concreteness we will write our surfaces (locally) in parameterized form $x^i (\vartheta^A)$. Then nearby surfaces take the form
\begin{equation}
    x^i (\vartheta^A) \approx x^i (\vartheta^A) + \epsilon \psi (\vartheta^A) N^i(\vartheta^A) \; . 
    \label{eq:nearby}
\end{equation}
Here the function $\epsilon \psi$ is playing the role of the variation $\Delta y^i$ from the finite-dimensional case. The stability operator 
$L_\mathcal{S} (\ell, X)$ on a MOTS $\mathcal{S}$ plays the same role as the Jacobian $[DF_o]$ at a fixed point $y_o$.
However, it can be seen in Table \ref{tab:placeholder} that the stability conditions are not the same. To reconcile this, we recall the notion of stability for minimal surfaces, which are a special case of MOTSs in slices that have $K_{ij} = 0$.

To find a surface that minimizes area (say within a given homotopy class), it is natural to evolve it in the direction that decreases area as quickly as possible. This is the negative gradient flow of the area functional, known in the literature as \emph{mean curvature flow}. This gives a dynamical system on the phase space of closed surfaces in $\Sigma$, where the normal velocity of a surface is equal to \emph{minus} its mean curvature. Fixed points of this dynamical system are minimal surfaces, and the linearization is minus the linearization of the mean curvature, that is, $-L_\cS(\ell,N)$. This has negative eigenvalues if the stability operator has positive eigenvalues, so the dynamical notion of stability for the mean curvature flow is in fact consistent with the geometric notion of stability for minimal surfaces.

To conclude this discussion, we note that while MOTSs are solutions of a nonlinear equation, and hence can be studied using tools of bifurcation theory, in general they do not arise as fixed points of a dynamical system. One exception is the spherical MOTSs in Reissner--Nordstr\"om--de Sitter, where there is a naturally associated dynamical system (Section \ref{sec:RNdS_Bif}).

\subsection{Slice dependence}
\label{sec:slicedependence}

Given a MOTS $\cS$ in the spacetime $M$, there are many foliations of $M$ by spacelike slices with $\cS$ being contained in some leaf of the foliation. Here we discuss the extent to which the stability and evolution of $\cS$ depend on the choice of foliation.

Our discussion applies to any one-parameter deformation of the initial data, as described in Section~\ref{sec:deformations}, where we replace the spacetime $M$ with $\Sigma \times \mathbb R$ foliated by slices of constant $\alpha$. Time evolution is a particular case of this general framework, with the spacetime foliated by constant time slices.

%If $\cS \subset \Sigma \times \{\alpha_0\}$ has an invertible stability operator, then it evolves smoothly with respect to any sufficiently small deformation of this constant $\alpha$ foliation.

First suppose $\cS$ is a MOTS in $\Sigma_0 = \Sigma \times \{\alpha_0\}$ and
its stability operator does not have a zero eigenvalue. If we change the foliation slightly, so that the leaves no longer correspond to constant $\alpha$, then in each slice sufficiently close to $\Sigma_0$ there will exist a unique MOTS close to $\cS$. This generates a smooth family of MOTSs in the new foliation.
%
%
%consider deforming to a nearby slice $\tilde\Sigma$. If $L_\cS(\ell,N)$ does not have a zero eigenvalue, then there will exist a nearby MOTS $\tilde\cS \subset \tilde\Sigma$ whose stability operator also does not have a zero eigenvalue. It follows that $\tilde S$ admits a smooth evolution in time, or any other parameter, in this new slicing. 
%
%
This may be different from the family obtained by evolving $\cS$ in the original slicing \cite{Ashtekar:2005,Krishnan2006,Nielsen2011SlicingDependence}, but the picture is qualitatively the same. That is, a change of foliation cannot cause a bifurcation when $\cS$ has an invertible stability operator in $\Sigma_0$.

The more interesting situation is when the stability operator is not invertible. In this case the bifurcation can change type or disappear completely. We will not study this exhaustively but instead give a few examples to show some of the possible outcomes. To simplify the discussion we only consider spherical symmetry, where the initial MOTS $\cS \subset \Sigma_0$ is given by $\{r=r_0\}$.

The expansion is then a function of two variables, $\theta^{l_+}(r,\alpha)$, and the corresponding bifurcation diagram is the zero set of this function. Deforming to a nearby (spherically symmetric) foliation may change the scaling of $l_+$, but by \eqref{theta:rescale} the expansion in this new foliation has the form $\theta_{\ell}(r,\alpha) = e^\mu \theta^{l_+}(r,\alpha)$ and so its zero set is unchanged. This means every MOTS that appears in the original foliation will also appear in the new foliation, hence a spherical change of the foliation can not cause a spherical MOTS to appear or disappear---it can only change the leaf in which it appears. This can affect the types of bifurcation that are observed, as illustrated in Figure~\ref{fig:PFtilt}.

\begin{figure}
\begin{tikzpicture}[scale=0.8]
	\draw[thick,->] (-2.5,-2.5) -- (-2.5,2.5);
	\node at (-2.5,3) {$\alpha$};
	\draw[thin,dashed] (-2,-0.15) -- (2,-0.15);
    \draw[thin,dashed] (-2,0) -- (2,0);
    \draw[thin,dashed] (-2,0.15) -- (2,0.15);

	\draw[domain=-1.5:1.5, smooth, variable=\x, blue, very thick] plot ({\x}, {\x*\x});
	\draw[very thick, blue] (0,-2) -- (0,2);
    \filldraw[blue] (0,0) circle (2pt);
\end{tikzpicture}
\hspace{2cm}
\begin{tikzpicture}[scale=0.8]
	\draw[thick,->] (-2.5,-2.5) -- (-2.5,2.5);
	\node at (-2.5,3) {$\alpha$};
	\draw[thin,dashed] (-2,1.15) -- (2,-1.85);
    \draw[thin,dashed] (-2,1.33) -- (2,-1.67);
    \draw[thin,dashed] (-2,1.5) -- (2,-1.5);
	\draw[domain=-1.5:1.5, smooth, variable=\x, blue, very thick] plot ({\x}, {\x*\x});
	\draw[very thick, blue] (0,-2) -- (0,2);
    \filldraw[blue] (0,0) circle (2pt);
    \filldraw[blue] (-0.35,0.1) circle (2pt);
\end{tikzpicture}

\caption{In spherical symmetry a pitchfork bifurcation (left) can split into a saddle-note bifurcation followed by a transcritical (right) when the foliation is changed.}
\label{fig:PFtilt}
\end{figure}

A more involved situation, which we will see concretely in Section~\ref{sec:imperfect}, occurs when there are multiple parameters that can be varied independently. In this paper we only consider one-parameter bifurcations, so in the multi-parameter case we are just varying along a curve in the phase space. In that case the function $\theta(r,\alpha)$ does not see the full phase space. Equivalently, slices of constant $\alpha$ (with all other parameters fixed) do not foliate the space  $\Sigma \times \mathbb R^d$, where $d$ is the number of parameters, but instead just sweep out a submanifold of codimension $d-1$. As a result, deforming the slicing $\Sigma \times \mathbb R$ in the full space $\Sigma\times\mathbb R^d$ can result in dramatically different behaviour in the bifurcation diagrams, since there may be spherical MOTSs that are not contained in any of the deformed slices.

% \begin{figure}
% \begin{tikzpicture}[scale=0.8]
% 	\draw[thin,->] (0,-2.5) -- (0,2.5);
% 	\node at (0,3) {$\alpha$};
% 	\draw[thin,->] (-2,0) -- (2,0);
% 	\node at (2.5,0) {$r$};
% 	\draw[blue, ultra thick] (-1.9,-1.9) -- (1.9,1.9);
% 	\draw[ultra thick, blue] (0,-2) -- (0,2);
% \end{tikzpicture}
% \hspace{1cm}
% \begin{tikzpicture}[scale=0.8]
% 	\draw[thin,->] (0,-2.5) -- (0,2.5);
% 	\node at (0,3) {$\alpha$};
% 	\draw[thin,->] (-2,0) -- (2,0);
% 	\node at (2.5,0) {$r$};
% 	\draw[domain=0.01:2, smooth, samples=100, variable=\x, red, ultra thick] plot ({\x}, {\x - 0.02/\x});
% 	\draw[domain=-2:-0.01, smooth, samples=100, variable=\x, red, ultra thick] plot ({\x}, {\x - 0.02/\x});
% \end{tikzpicture}
% \hspace{1cm}
% \begin{tikzpicture}[scale=0.8]
% 	\draw[thin,->] (0,-2.5) -- (0,2.5);
% 	\node at (0,3) {$\alpha$};
% 	\draw[thin,->] (-2,0) -- (2,0);
% 	\node at (2.5,0) {$r$};
% 	\draw[domain=0.01:2, smooth, samples=100,variable=\x, red, ultra thick] plot ({\x}, {\x + 0.02/\x});
% 	\draw[domain=-2:-0.01, smooth, samples=100,variable=\x, red, ultra thick] plot ({\x}, {\x + 0.02/\x});
% \end{tikzpicture}
% \caption{Two possible deformations of a transcritical bifurcation (left): two smooth curves with no bifurcations (centre) or a pair of saddle-node bifurcations (right).}
% \label{fig:imperfect:trans}
% \end{figure}

We will give some explicit examples of this
in Section~\ref{sec:imperfect}. For instance, in Figure~\ref{fig:TransPert} we see that a transcritical bifurcation can split into a pair of smooth, non-bifurcating curves or a pair of saddle-node bifurcations, whereas in Figure~\ref{fig:PitchPert} we see that a pitchfork can split into a smooth non-bifurcating curve and a saddle-node bifurcation.

These are referred to in the literature as ``imperfect bifurcations" \cite{GSimperfect,strogatz2015}, since they occur when the underlying equation itself is perturbed, rather than the parameter $\alpha$. We expect to see similar behaviour when considering bifurcations of non-spherical MOTSs, since in that case the space of allowable deformations is infinite dimensional. This provides a possible explanation for why saddle-node bifurcations are most commonly seen in numerical simulations, whereas transcritical and pitchfork bifurcations have only been observed in precisely tuned exact solutions or in problems with axial symmetry, as we will see below.

\section{Spherically symmetric bifurcations: Reissner--Nordstr\"om--de Sitter}
\label{sec:sphere_stab}

Our first set of examples comes from the 
 Reissner--Nordstr\"om--de Sitter family of spacetimes.  
In this section we restrict our attention to spherical symmetry and bifurcations with a vanishing principal eigenvalue, so all calculations can be done exactly. We will see that in this case the spherically symmetric MOTSs and their bifurcations
can also be explicitly understood as a one-dimensional dynamical system. 

\subsection{Stability spectrum in spherical symmetry}

For a spherically symmetric MOTS the induced
metric is that of a sphere of some areal radius $R_o$, with Gauss curvature $\mathcal{K} = 1/R_o^2$. If we further assume 
that all other quantities appearing in the stability operator are similarly symmetric,
then $\omega^{l_+} = 0$, $\sigma_{l^+} = 0$, 
and all scalar quantities are constants over $\mathcal{S}$. 
In this case the stability operator (\ref{eq:varSigma}) reduces to
% \begin{equation}
%     L_{\mathcal{S}} (\ell, X) \psi  = \frac{e^\mu B }{R_o^2}  \left( - \triangle + 1 - R_o^2 \left[\frac{A}{B} G_{l^+ l^+} + G_{l^+ l^-} \right]  \right) \psi \label{eq:SSstab}
% \end{equation}
\begin{equation}
    L_{\mathcal{S}} (\ell, N) \psi  = \frac{e^\mu}{R_o^2}  \left( - \triangle + 1 - R_o^2 \left[\frac12 G_{l^+ l^+} + G_{l^+ l^-} \right]  \right) \psi \label{eq:SSstab}
\end{equation}
where $\triangle$ is the Laplacian on the unit 
sphere. 
The eigenfunctions are spherical 
harmonics $Y_{np}(\theta, \phi)$  with corresponding eigenvalues
% \begin{equation}
%     \lambda_{nm} = \frac{e^\mu B }{R_o^2}  \bigg(  n(n+1) + 1  - R_o^2 \left[\frac{A}{B} G_{l^+ l^+} + G_{l^+ l^-} \right] \bigg)  
%     \label{eq:SSeig2} \; . 
% \end{equation}
\begin{equation}
    \lambda_{np} = \frac{e^\mu}{R_o^2}  \bigg(  n(n+1) + 1  - R_o^2 \left[\frac12 G_{l^+ l^+} + G_{l^+ l^-} \right] \bigg)  
    \label{eq:SSeig2} \; . 
\end{equation}
From (\ref{eq:SSeig2}) it is clear that in the absence of matter fields all eigenvalues are positive and so a spherically symmetric MOTS in vacuum is always strictly stable. 

All of the examples that we consider in this section will be isolated horizons so that $G_{l^+ l^+} = 0$, in which case
\begin{equation}
    \lambda^{\mathrm{iso}}_{np} = \frac{e^\mu}{R_o^2}  \bigg(  n(n+1) + 1  - R_o^2  G_{l^+ l^-}  \bigg)  \; . 
    \label{eq:SSeigiso}
\end{equation}
 %Note too that for these cases, it is manifest that the stability classification does not depend on the scaling of $X$, $\ell$ or $(l_+,l_-)$. 
Physically,  (\ref{eq:SSeigiso}) can be understood as showing that sufficiently dense matter induces instability.

\subsection{Spherical MOTSs in RNdS spacetime}

In infalling Eddington--Finkelstein coordinates the RNdS metric takes the form
\begin{equation}
   g_{\alpha \beta} \dd x^\alpha \dd x^\beta = - \frac{f(r)}{r^2}  \dd v^2 + 2 \dd v \dd r + r^2 \dd \Omega^2 
   \label{eq:metricRNds}
\end{equation}
with 
\begin{equation}
    f(r) = - \frac{\Lambda}{3} r^4  + r^3 - 2mr  + q^2 \; . \label{eq:f} 
\end{equation}
We assume $\Lambda > 0$ so that the spacetime is asymptotically de Sitter. We also assume $q \neq 0$.
%By the structure of this polynomial, with $\Lambda > 0$ there is always at least one positive and one negative root ($f(0)>0$ but $f(r)$ necessarily becomes negative as $r \rightarrow \pm \infty$). 
%However, there will be up to three positive roots and these (in ascending order) correspond to the inner black hole, outer black hole and cosmological horizons. 

For null normals
\begin{equation}
    l_+ = \frac{\partial}{\partial v} + \frac{f}{2r^2}  \frac{\partial}{\partial r} \quad \mathrm{and} \quad
     l_- = -   \frac{\partial}{\partial r}
\end{equation}
the outward expansion is
\begin{equation}
    \theta_{l^+} = \frac{f}{r^3} \; ,
\end{equation}
thus there are MOTSs at all radii $r_o$ for which $f(r_o)=0$.
 When $\Lambda > 0$ and $q \neq 0$ there is always at least one positive and one negative root, since $f(0)>0$ and $f(r)$ becomes negative as $r \rightarrow \pm \infty$. However, there can be up to three positive roots and these (in ascending order) correspond to the inner black hole, outer black hole and cosmological horizon. These are all isolated horizons.

Now
\begin{equation}
    G_{l^+ l^-} = \frac{1}{r^{2}}-\frac{f_{r}}{r^{3}}+\frac{f}{r^{4}}\; . 
\end{equation}
Then, by (\ref{eq:SSeigiso}) the 
principal eigenvalue of a MOTS with $r = r_o$ is
\begin{equation}
    \lambda_o = \frac{e^\mu}{r_o^2}  \bigg( 1  - r_o^2  G_{l^+ l^-}  \bigg)  =  \frac{e^\mu}{r_o^3}
      f_r(r_o) \; ,
\end{equation}
so the sign is determined by the slope of $f(r)$ at $r_o$. 
In particular, this means the principal eigenvalue of the outermost MOTS is always non-positive, since 
$\lim_{r \rightarrow \infty} f(r) = - \infty$. Thus if there are three positive roots, we recover the usual result: the  inner black hole and cosmological horizons are unstable and the outer black hole horizon is strictly stable.

\subsection{Bifurcations in RNdS}
\label{sec:RNdS_Bif}

From the metric, it is straightforward to see that radial null curves satisfy either $v = \mbox{constant}$ for inward-oriented curves or 
\begin{equation}
    \frac{ \dd r}{\dd v} = \frac{f}{2 r^2}  
    \label{eq:RNdS_DS}
\end{equation}
for outward-oriented curves. We focus on the outward-oriented curves. Then (\ref{eq:RNdS_DS}) is a one-dimensional dynamical system and the fixed points are the radial locations of the 
inner, outer and cosmological horizons. 

As noted in Section \ref{sec:BifTheory}, there is a mismatch in terminology between dynamical and geometrical notions of stability. Geometrically 
a horizon with $f_r > 0$ is stable and outward-oriented, spherically symmetric null rays evolve away from it. If $f_r<0$ it is unstable and outward-oriented, spherically symmetric null rays evolve towards it. However, viewed as a dynamical system, these would respectively be labelled unstable and stable. 

We now consider how the fixed points change if we promote the parameters in (\ref{eq:f}) to functions $\Lambda=\Lambda(\alpha)$, $m=m(\alpha)$ and $q=q(\alpha)$ and then evolve through the phase space of solutions. These examples are finely tuned in the sense that we carefully choose functions so that the transcritial and pitchfork examples have a persistent fixed point at a fixed value of $r$. 
However, these are still useful examples as they clearly demonstrate the close relationship between MOTS evolution and bifurcation theory in a familiar exact solution setting. The following three bifurcations occur when the principal eigenvalue vanishes and a previously strictly stable MOTS becomes unstable.

\begin{figure}
\includegraphics{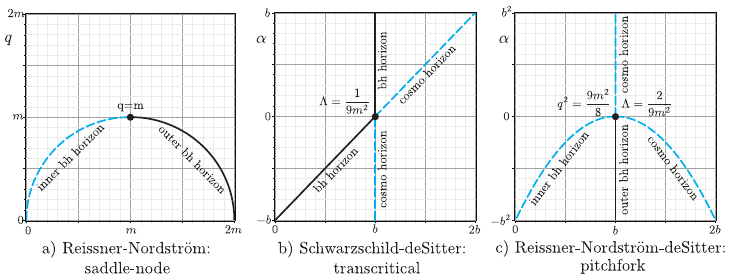}
    \caption{Bifurcations in the Reissner--Nordstr\"om--de Sitter family of spacetimes. The horizontal axis is the areal radius of the MOTS and the vertical axis is the spacetime deformation parameter. The solid dots are the bifurcation points with the physical parameters of those points listed. %\ivan{labels: to match Graham changes plus make $q/Q$ and $m/m$ consistent}
    }
    \label{fig:RNdS}
\end{figure}

\subsubsection{Saddle-node}
\label{sec:RNdS_SN}

Our first example is very well known. Setting $\Lambda = 0$ and so restricting to Reissner--Nordstr\"om and fixing $m=\mbox{constant}$, we vary $q$. For $q< m$ there are two roots: the inner and outer black hole horizons, which are respectively unstable and strictly stable. At $q=m$ the horizons become degenerate, $f_r(m)=0$, and the resulting extremal horizon has a vanishing principal eigenvalue. For $q>m$ there are no longer any MOTSs and we are left with a naked singularity. 

This pair annihilation of horizons, shown in Figure \ref{fig:RNdS}a, is an example of a saddle-node bifurcation.

\subsubsection{Transcritical}
\label{sec:RNdStrans}

Setting $q=0$ and so restricting to Schwarzschild--de Sitter,
we parameterize $f(r)$ as
\begin{equation}
f(r) = - r\left(\frac{\Lambda  r}{3}+a \right) \left(r -b \right) \left(r -b -\alpha \right) \label{eq:ftc}
\end{equation}
and assume that $a, b > 0$. Then the negative root is at
$r= -3a/\Lambda$, while the outer black hole and cosmological horizons are respectively at
$r=b$ and $r=b+ \alpha$ if $\alpha>0$ (or vice versa if $\alpha < 0$). If $\alpha = 0$, there is a degenerate root at $r=b$ at which 
$f_r(b)=0$ and the principal eigenvalue vanishes.

Fixing $b$ for scale and varying $\alpha$, there is a transcritical bifurcation at $\alpha=0$. Starting with $\alpha < 0$ and increasing it, the physical picture is that the black hole 
horizon moves out towards a cosmological horizon of area $4 \pi b^2$. At $\alpha=0$ they become degenerate and then as $\alpha$ continues to increase, the cosmological horizon increases in area while the black hole horizon remains with constant area $4 \pi b^2$. This is shown in Figure \ref{fig:RNdS}b). 

Matching (\ref{eq:f}) and (\ref{eq:ftc}), we find expressions for $\Lambda$, $a$ and $m$ in terms of $b$ and $\alpha$:
\begin{align}
    \frac{\Lambda}{3} & = \frac{1}{3 b^{2}+3 b \alpha +\alpha^{2}} \label{SdS1}  \\
    a & = \frac{\Lambda}{3} ( 2 b + \alpha ) \\
    m & = \frac{b}{2} \left( 1 - \frac{\Lambda}{3} b^2 \right) \; \label{SdS3}. 
\end{align}
Thus $\Lambda$ is always positive, $a$ is positive if $\alpha > - 2b$ and, after some calculation, it can be seen that $m$ is positive if $\alpha > - b$. Hence the bifurcation happens for reasonable physical parameters. 
See \cite{Senovilla:2022bsn,Senovilla:2023npd}
for a nice discussion of a physical version of this process
(though with a slightly different perspective than we use here).

\subsubsection{Pitchfork} 
\label{sec:RNdS_pitch}

For the pitchfork we need all physical parameters non-zero. We parameterize
\begin{equation}
    f(r) = - \left(\frac{\Lambda  r}{3}+a \right) \left(r -b \right) \left((r -b)^2 +\alpha) \right) \, , \label{eq:fpf} 
\end{equation}
and assume that $a, b > 0$. 
Then the negative root is again at $r = -3a/\Lambda$ and there is always a root 
at $r=b$. For $\alpha > 0$ these are the only two roots and physically the $r=b$ root is the cosmological horizon in a naked singularity spacetime. For
$\alpha < 0$ there are three positive roots: $b - \sqrt{-\alpha}$, $b$ and $b + \sqrt{-\alpha}$, which are the inner black hole, outer black hole and cosmological horizon, respectively. 

Fixing $b$ for scale, we find a pitchfork bifurcation at $\alpha = 0$. For $\alpha > 0$, $r=b$ is the only MOTS and it is unstable. At $\alpha = 0$ it becomes triply degenerate with a vanishing principal eigenvalue. Then for $\alpha < 0$, $r=b$ becomes the black hole horizon, which is strictly stable, while $r = b \pm \sqrt{-\alpha}$ are the inner black hole and cosmological 
horizon, which are both unstable.

Matching (\ref{eq:f}) and (\ref{eq:fpf}), we find expressions for $\Lambda$, $a$, $m$ and $q^2$ in terms of $b$ and $\alpha$:
\begin{align}
   \frac{\Lambda}{3} &=  \frac{1}{6 b^{2}-\alpha} \label{eq:RNdSpitch1} \\
a & = \Lambda b \\
m &= \frac{\Lambda b}{3} \left(4 b^{2}+\alpha \right) \\
q^2 & = \Lambda b^{2} \left(b^{2}+\alpha \right) \;. \label{eq:RNdSpitch4}
\end{align}
For $b>0$ and $\alpha > -b^2$ all of these quantities are positive. Hence, the pitchfork bifurcation happens for well-defined physical parameters (though  it does involve a naked singularity). 
This process is illustrated in Figure \ref{fig:RNdS}c).

\subsection{Perturbed bifurcations}
\label{sec:imperfect}

As noted, these examples are finally tuned. Hence they are  easily generalized to examples of imperfect bifurcations. Almost any perturbation of $f$ would do but for demonstration purposes we choose $m \rightarrow m (1 + \epsilon)$. That is,
\begin{equation}
    f_{\mathrm{\tiny{pert}}} = f - 2\epsilon m r \; . 
\end{equation}

\subsubsection{Perturbed saddle-node}

As would be expected for the saddle-node bifurcation, the perturbation moves but does not change the character of the bifurcation point. We have
\begin{equation}
    f_{\mathrm{\tiny{pert}}} = r^2 -2 (1+\epsilon)m r +q^2 \; . 
\end{equation}
For a path parameterized by 
increasing $q$ as in Section \ref{sec:RNdS_SN}, there will still be a double root 
at 
\begin{equation}
    q = (1+\epsilon) m \; . 
\end{equation}
The $(1,0)$-unstable MOTS corresponding to a saddle-node pair annihilation moves but  remains. 

\subsubsection{Perturbed transcritical}

\begin{figure}
    \centering
    \includegraphics[width=0.95\linewidth]{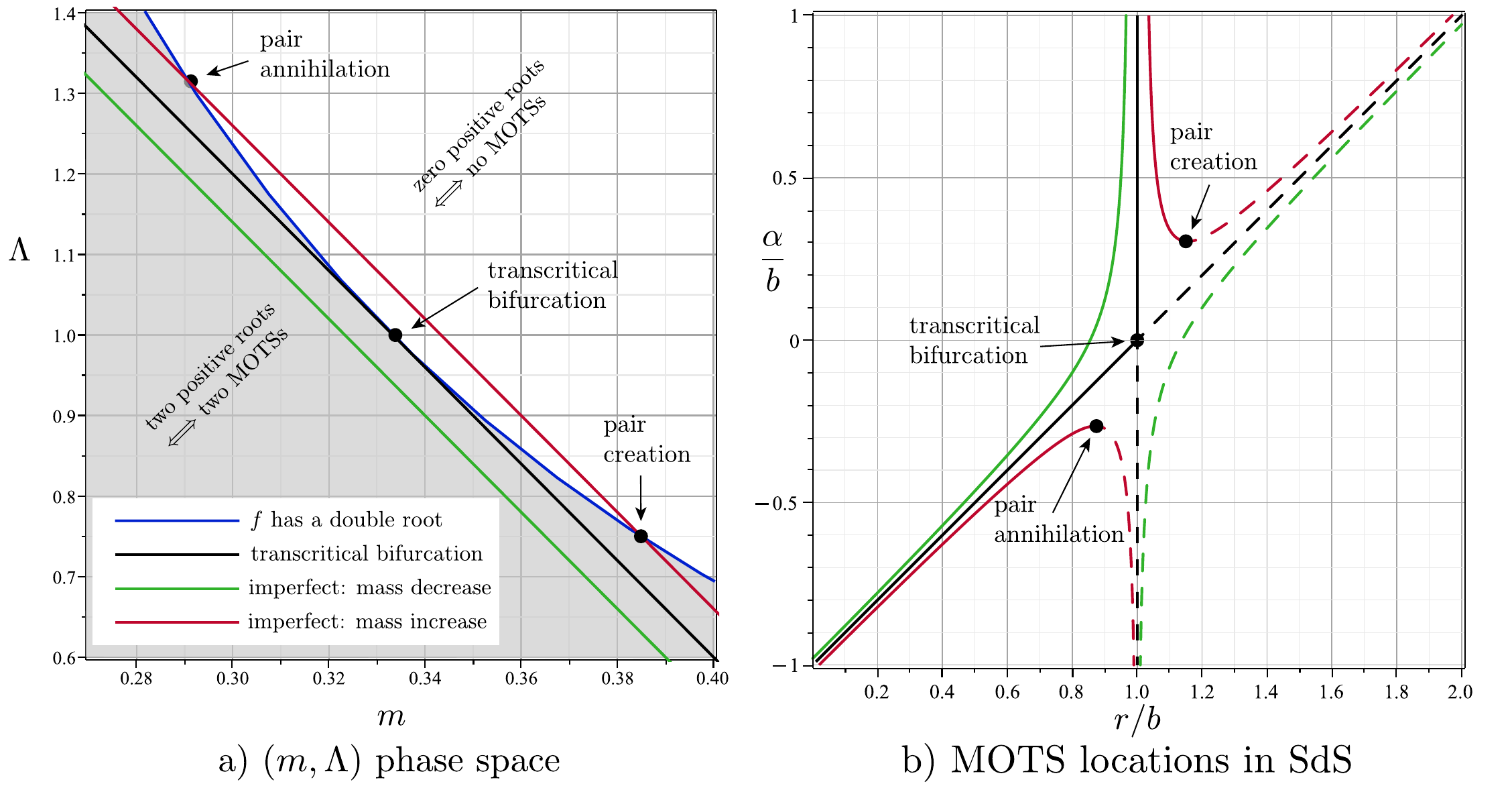}
    \caption{Transcritical bifurcation in SdS  perturbed into imperfect bifurcations. The left-hand figure shows the ranges of $m$ and $\Lambda$ for which $f(r)$ has two positive roots (gray) or no positive roots (white) along with the values for the transcritical bifurcation (black) and its perturbations (red and green). The right-hand figure shows the corresponding MOTS locations in $r$ as $\alpha$ scales.} 
    \label{fig:TransPert}
\end{figure}

For the transcritical bifurcation of Section \ref{sec:RNdStrans}
\begin{equation}
    f_{\mathrm{\tiny{pert}}} = - \frac{\Lambda}{3} r^4 + r^2 - 2m(1+\epsilon) r \; . 
\end{equation}
Then Figure \ref{fig:TransPert} shows how this perturbation affects the transcritical bifurcation. The left-hand side of that figure is the $(m, \Lambda)$ phase space of SdS solutions 
along with the paths traced through that space  by the (\ref{SdS1})-(\ref{SdS3}) parameter curve and its perturbations. 

A negative mass perturbation  (the green curves) pushes the path fully into the region where there is both a black hole and cosmological horizon while a positive mass perturbation pushes it partly into the naked singularity region. This is reflected in the right-hand side of the figure where the bifurcation becomes imperfect. For a negative mass perturbation there is no longer any bifurcation at all while for a positive one it degenerates into two saddle-nodes. 

\subsubsection{Perturbed pitchfork}

\begin{figure}
    \centering   \includegraphics[width=0.95\linewidth]{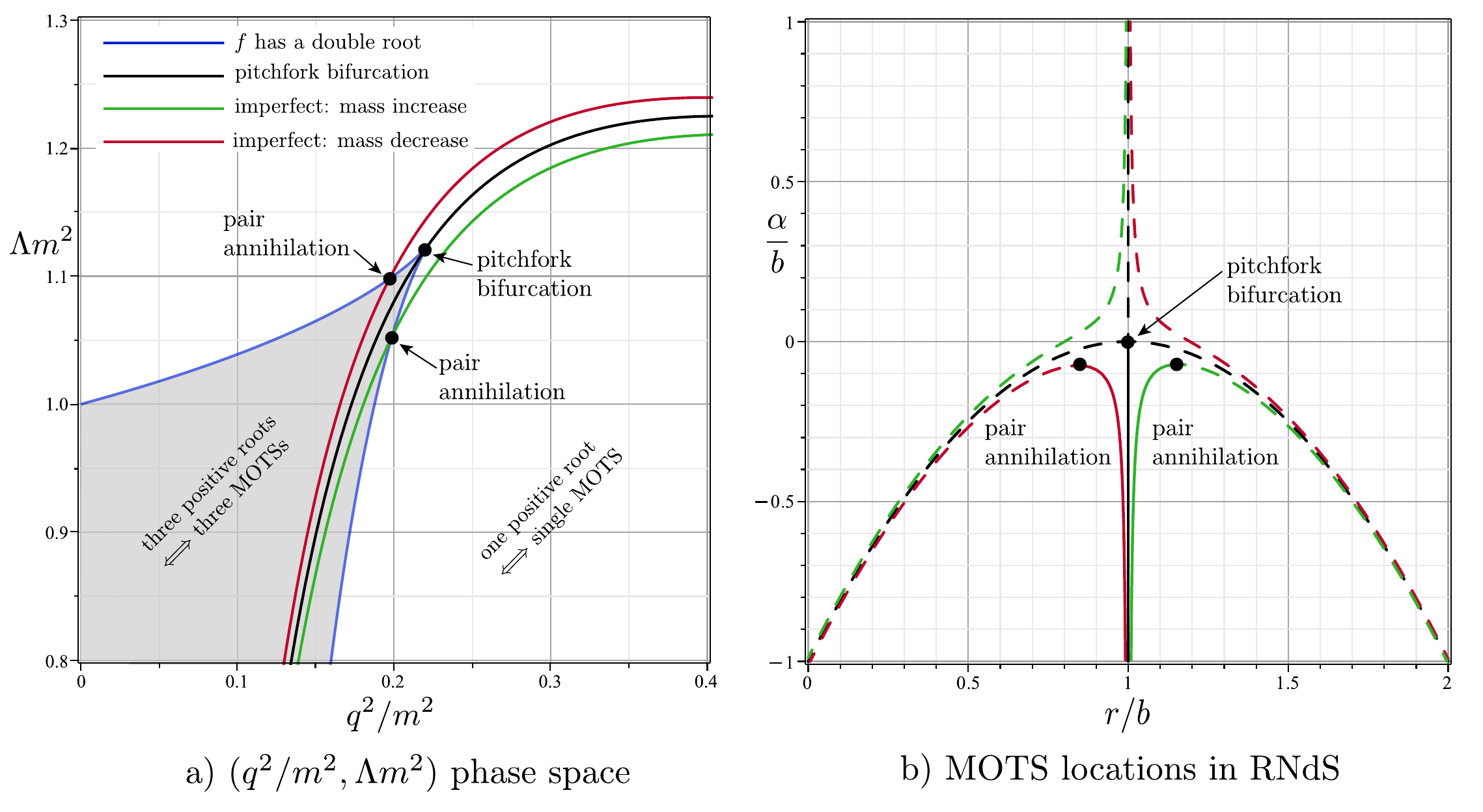}
    \caption{Pitchfork perturbation in RNdS perturbed into imperfect bifurcations. The left-hand figure shows the ranges of $q^2/m^2$ and $\Lambda m^2$ for which $f(r)$ has three positive roots (gray) or only a single positive root (white) along with the values for the pitchfork bifurcation (black) and its perturbation (red and green). The right-hand figure shows  shows the corresponding MOTS locations in $r$ as $\alpha$ scales.  }
    \label{fig:PitchPert}
\end{figure}

For the pitchfork bifurcation of Section \ref{sec:RNdS_pitch},
\begin{equation}
    f(r) = - \frac{\Lambda}{3} r^4 + r^2 -2m(1+\epsilon)r+q^2  \; . 
\end{equation}
Then Figure \ref{fig:PitchPert} shows how this perturbation affects the pitchfork bifurcation. This time the phase space of solutions is three-dimensional $(m, q, \Lambda)$ but for display purposes we show a cross-section where $q^2$ and $\Lambda$ are appropriately scaled by $m^2$ to be dimensionless. The left-hand side shows that cross-section along with the paths traced through the phase space by (\ref{eq:RNdSpitch1})-(\ref{eq:RNdSpitch4}) and its perturbations. 

The right-hand side of the figure shows that both positive and negative perturbations result in an imperfect pitchfork. For a positive mass perturbation, the inner black hole horizon disconnects from the bifurcation and the rest of it becomes a saddle-node bifurcation, where the outer black hole and cosmological horizon annihilate. For a negative mass perturbation, the cosmological horizon disconnects from the bifurcation and the rest of it becomes a saddle-node, where the inner and outer horizons annihilate.

\section{Axisymmetric bifurcations I:  Reissner--Nordstr\"om
%of a stability operator
} 
\label{sec:axiRN}

We now turn from spherical symmetry to consider an example in which axisymmetric MOTSs bifurcate from the inner horizon of a Reissner--Nordstr\"om black hole. Except for the transition from extremality, which was seen in the previous section, all of these bifurcations correspond to vanishing non-principal eigenvalues.

\subsection{Vanishing eigenvalues on the inner horizon}
Using the techniques from Section \ref{sec:sphere_stab} we calculate the (axisymmetric) eigenvalues of the inner horizon as
\begin{align}
    \lambda_{n0}  &=  \frac{1}{P_{\mathrm{in}}^2} \left( n(n+1) + 1 - \frac{q^2}{P_{\mathrm{in}}^2} \right) \\
     & = \frac{1}{P_{\mathrm{in}}^2} \left( n(n+1) - \frac{2 \sqrt{m^2 - q^2}}{m - \sqrt{m^2 - q^2} }\right)  
\end{align}
and so, as claimed in (\ref{eq:RNvanish}), the 
zeroth, first and second eigenvalues vanish for $q=m$, $q=\frac{\sqrt{3}}{2} m$ and 
$q = \frac{\sqrt{7}}{4} m$, respectively. It is around these parameter values that we will look for bifurcations of the inner horizon into other MOTSs.

\subsection{The metric and coordinates}
\label{sec:RNmetric}

% This time we go beyond the principal eigenvalue and consider bifurcations associated with other vanishing eigenvalues. Only the 
% extremality-to-naked-singularity transition corresponds to a loss of boundary between trapped and untrapped regions. When higher eigenvalues vanish
% we  see examples of transcritical and pitchfork bifurcations from the inner horizon.

We work in generalized Painlev\'{e}--Gullstrand coordinates \cite{Hennigar:2021ogw, Martel:2000} so that the spacetime metric is
\begin{equation}
    g_{\alpha \beta} \dd x^\alpha \dd x^\beta =  -F \dd T^2 + 2 \sqrt{1 - pF} \dd T \dd r + p \dd r^2 + r^2 \! \dd \Omega^2 
\end{equation}
where $F = F(r) = 1 - \frac{2m}{r} + \frac{q^2}{r^2}$, $p = p(r)$ is a free function of $r$ 
%$d \Omega^2 = \dd \theta^2 + \sin^2 \!  \theta \,  \dd \phi^2 $ is the metric on a unit two-sphere and the 
and the stress-energy tensor is 
\begin{equation}
    T_{\alpha \beta} \dd x^\alpha \dd x^\beta =
    \frac{q^2}{r^4} \left( F \dd T^2 - 2 \sqrt{1 - pF} \dd T \dd r -p  \dd r^2 + r^2 \dd \Omega^2  \right)  \; . 
\end{equation}

Then surfaces of constant $t$ have induced metric
\begin{equation}
    h_{ij} \dd x^i \dd x^j = p \dd r^2 + r^2 d \Omega^2   \, , 
\end{equation}
forward-in-time oriented unit normal
\begin{equation}
    u^\alpha \left( \frac{\partial}{\partial x^\alpha} \right)  = \sqrt{p} \frac{\partial}{\partial T} - \sqrt{\frac{1 - pF}{p}} \frac{\partial}{\partial r}
\end{equation}
and extrinsic curvature 
\begin{equation}
    K_{ij} \dd x^i \dd x^j =  \frac{F_r p^2 + p_r}{2  \sqrt{p(1-pF)}} \dd r^2 - r \sqrt{\frac{ 1-pF }{p}} \dd \Omega^2 \; . 
\end{equation}
The energy density seen by observers moving perpendicular to these surfaces is 
\begin{equation}
    \rho = T_{\alpha \beta} u^\alpha u^\beta = \frac{q^2}{8 \pi r^4}  \; . 
\end{equation}

The free function $p(r)$ adapts this coordinate system so that the radial acceleration of a worldline with tangent vector $u$ has magnitude \cite{Hennigar:2021ogw}
\begin{equation}
    \| a \| =\left|   \frac{p'}{2p^{3/2}} \right| \, , 
\end{equation}
where the prime indicates a derivative with respect to $r$. Constant $p$
corresponds to freely falling observers as discussed in \cite{Martel:2000}.
For Reissner--Nordstr\"om such  observers do not reach $r=0$: they turn around at $r = q^2/2m$. However, this can be overcome by adapting the coordinate system to accelerating observers. There are many choices for which the coordinates run all the 
way to the origin. In this paper we choose
\begin{equation}
    p(r) = \frac{r^2}{r^2 + q^2}  \quad \Longrightarrow \quad \| a \| = \frac{q^2}{r^2 \sqrt{r^2+q^2}} \; . 
\end{equation}
Note that this is a different $p$ than was presented in \cite{Hennigar:2021ogw}. However, there are only small quantitative differences in the resulting MOTSs. 

\subsection{Finding MOTSs}

Non-spherically symmetric MOTSs are found numerically. 
By the discussion of Section \ref{sec:MOTSodesic}, the MOTSodesic equations are
\begin{align}
    \ddot{P} & = - \frac{q^2}{P(P^2+q^2)} \, \dot{P}^2 + \frac{P^2+q^2}{P} \, \dot{\Theta}^2 + \kappa \sqrt{P^2 + q^2} \dot{\Theta} \label{eq:RNddP}\\
    \ddot{\Theta} & = - \frac{2}{P} \, \dot{P}  \dot{\Theta} - \frac{\kappa}{\sqrt{P^2+q^2}} \, \dot{P} \label{eq:RNddT}
\end{align}
where dots are derivatives with respect to the arclength $s$ and
\begin{equation}
\label{eq:kappaRN}
    \kappa =   -
     \sqrt{\frac{2m}{P^3}} + \frac{ \sqrt{mP} }{\sqrt{2}(P^{2}+q^2)} \,  \dot{P}^2 - \frac{\cot \! \Theta }{\sqrt{P^2+q^2}}  \,  \dot{P}  -  \sqrt{2 mP }  \,  \dot{\Theta}^2+\frac{\ \sqrt{P^2+q^2} }{P}  \,  \dot{\Theta} \; . 
\end{equation}
The arclength constraint $\| T \| = 1$ is 
\begin{equation}
    \frac{P^2}{P^2 + q^2} \left(\dot{P}^2 + (P^2 + q^2) \dot{\Theta}^2 \right) = 1 \; . \label{eq:RNarclength}
\end{equation}

%For this paper, we restrict our attentions to 
%bifurcations from the spherically symmetric inner horizon MOTS found at $P_{\mathrm{in}} = M - \sqrt{M^2-q^2}$. 

\subsection{Stability spectra}

% Using the techniques from Section \ref{sec:sphere_stab} we calculate the eigenvalue spectrum of the inner horizon as
% \begin{align}
%     \lambda_{nm}  &=  \frac{1}{P_{\mathrm{in}}^2} \left( n(n+1) + 1 - \frac{q^2}{P_{\mathrm{in}}^2} \right) \\
%      & = \frac{1}{P_{\mathrm{in}}^2} \left( n(n+1) - \frac{2 \sqrt{M^2 - q^2}}{M - \sqrt{M^2 - q^2} }\right)  
% \end{align}
% and so, in agreement with (\ref{eq:RNvanish}), the 
% zeroth, first and second eigenvalues
% respectively vanish for $q=m$, $q=\frac{\sqrt{3}}{2} m$ and 
% $q = \frac{\sqrt{7}}{4} m$. It is around these MOTSs that we will look for bifurcations of the inner horizon into other MOTTs. 

Using the MOTSodesic equations, we will find axisymmetric MOTSs bifurcating from the inner horizon. Their stability is determined by the operator $\tilde L$ in (\ref{eq:L1d}), so
% are axisymmetric and so, following the discussion earlier in this section, 
we calculate
\begin{equation}
    \mathcal{G} = -\frac{\dot{P}}{P}+ 3 \sqrt{\frac{ 2mP}{P^2+q^{2}}} \dot{P} \dot{\Theta}- (\cot \! \Theta)  \dot{\Theta} 
\end{equation}
and
\begin{align}
    \mathcal{H} = &  -\frac{\ddot{P}}{P} + \left(  - \ddot{\Theta}   +   \frac{2 \dot{P} \dot{\Theta} }{P}  \right) \cot \! \Theta - \frac{2 \dot{P}^2 \cot^2 \! \Theta}{P^2+q^2} - \left( 1 + \frac{2 q^2}{P^2}  \right)  \dot{\Theta}^{2} \nonumber - \frac{q^2}{P^4} \nonumber \\
    & \quad + \sqrt{ \frac{9m}{{2P^3 (P^2 + q^2)}}} \Bigg( P^2 \left(\ddot{P} \dot{\Theta} + \ddot{\Theta} \dot{P} \right)  + \dot{P} \left(\frac{8}{3} + P^2 \dot{\Theta}^2  \right) \cot{\Theta} \nonumber \\
    & \quad + P \dot{P}^2 \dot{\Theta} \left(\frac{1}{2} + \frac{q^2}{P^2 + q^2} \right) - \frac{8\dot{\Theta}}{3P} (P^2 + q^2)  \Bigg)  -  m \left(\frac{4}{P^3} + \frac{9P \dot{P}^2 \dot{\Theta}^2}{2(P^2 + q^2)} \right) \; . 
\end{align}
Using the numerical methods of \ref{app:spectra}, 
calculating the first 10 eigenvalues for the $L$ defined by these coefficients can be done accurately using the first 20 Chebyshev polynomials as a basis and takes on the order of one second on a typical laptop.

\subsection{The first bifurcations}

We examine the first three bifurcations encountered in the Reissner--Nordstr\"om spacetime as the charge is decreased from $q=m$.

\subsubsection{Transition from extremality: a saddle-node bifurcation}
\label{sec:RN_innerSN}

For $q=m$ there is a single spherical MOTS: the extremal horizon with vanishing principal eigenvalue. For $q < m$ this splits into outer and inner horizons, 
as previously discussed in Section \ref{sec:RNdS_SN} and explicitly shown in 
Figure \ref{fig:RN_combo}c). In that figure (and all others for this section) $q$ increases from left to right and so the extremal horizon is on the far right.
%with the bifurcation progressing to the left. 

This bifurcation is also shown in the top part of Figure \ref{fig:BifRN}c), which plots the intersection of the inner and outer MOTSs with the positive $z$-axis as a function of $q$. As noted earlier, this is a saddle-node bifurcation of spherical MOTSs.

\begin{figure}
    \centering
    \includegraphics[scale=0.95]{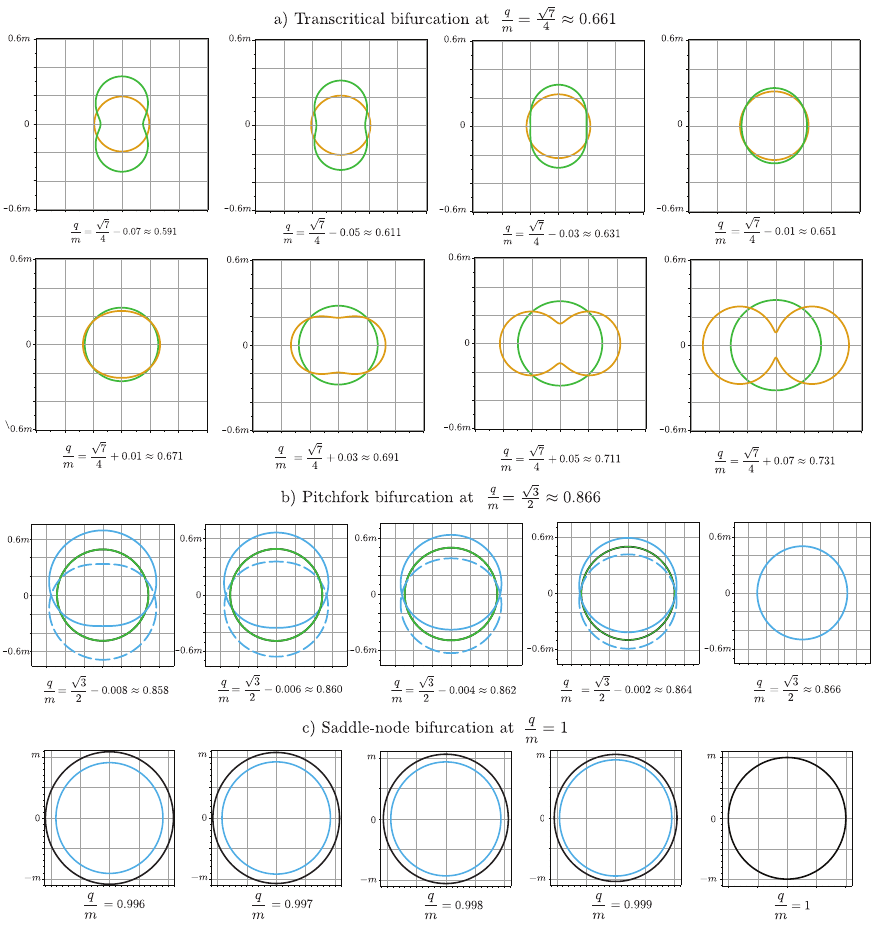}
    \caption{Axisymmetric MOTSs bifurcating from the Reissner--Nordstr\"om inner horizon. Charge is increasing from the top left to the bottom right in this figure. Black, blue, green and gold MOTSs respectively have 0, 1, 2 or 3 negative eigenvalues. For the pitchfork, the two non-spherical MOTSs are shown solid and dashed so that they are easier to distinguish. }
    \label{fig:RN_combo}
\end{figure}

The lower part of Figure \ref{fig:BifRN}c) plots the eigenvalues of the two MOTSs. As would be expected at $q=m$, the point of bifurcation, an eigenvalue vanishes. In this case it is the principal eigenvalue. All other eigenvalues remain positive in the range of this graph.

\begin{figure}
    \centering
    \includegraphics[scale=0.95]{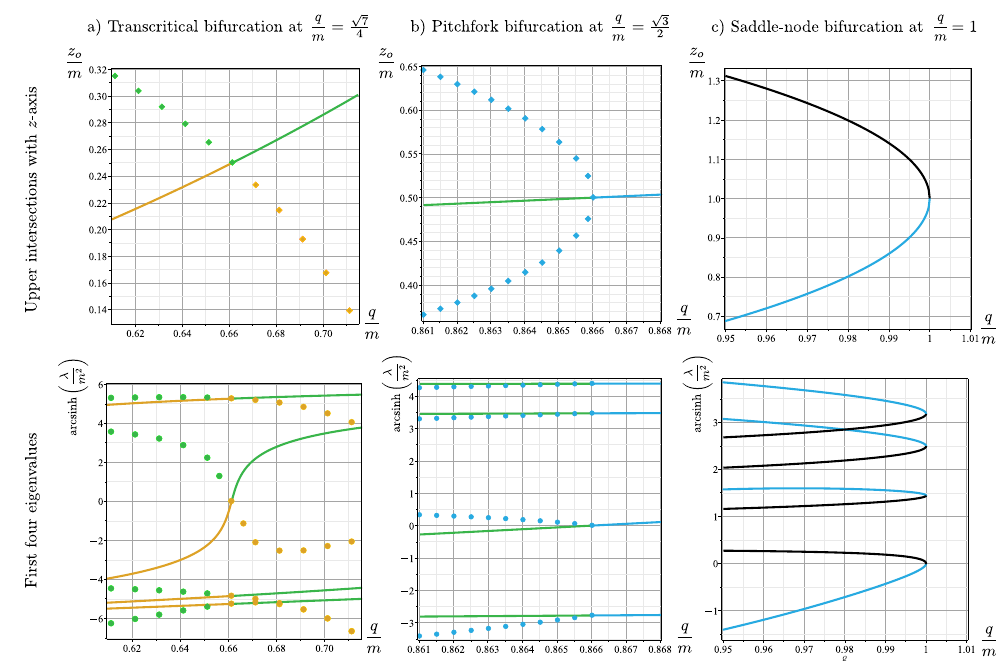}
    \caption{The first three bifurcations of Reissner--Nordstr\"om. 
    %In all diagrams charge increases from left to right. 
    The top row shows the intersection of the MOTSs (from Figure \ref{fig:RN_combo}) with the $z$-axis. The bottom row plots the first four eigenvalues of the various MOTSs. As expected, each bifurcation corresponds to a vanishing eigenvalue. The solid lines track the inner and outer horizon (and are calculated exactly) while the dots are the numerically calculated axisymmetric MOTSs. As in Figure \ref{fig:RN_combo}, black, blue, green and gold curves/dots respectively correspond to MOTSs with 0, 1, 2 or 3 negative eigenvalues. %\ivan{fix labels}
    }
    \label{fig:BifRN}
\end{figure}

\subsubsection{Transition to two negative eigenvalues: a pitchfork bifurcation}
\label{sec:RN_innertrans}

More interesting is the MOTS evolution starting from $q = \frac{\sqrt{3}}{2} m$ that is shown in Figure \ref{fig:RN_combo}b). 
For $q>\frac{\sqrt{3}}{2} m$ there are no MOTSs close to the inner horizon. However, as the $n=1$ eigenvalue vanishes for $q=\frac{\sqrt{3}}{2} m$ and then becomes negative for $q<\frac{\sqrt{3}}{2} m$, two MOTSs split from the inner horizon and evolve away from it. Geometrically they are reflections of each other through  the equatorial plane. 

The top part of Figure \ref{fig:BifRN}b) plots the positive $z$ intersections as a function of $q$ and it is clear that this takes the form of a pitchfork bifurcation. During this transition, the inner horizon starts with two negative eigenvalues (green) and then switches to a single negative eigenvalue (blue) as the charge increases to $q = \frac{\sqrt{3}}{2} m$. At that point, the two non-spherical MOTSs join and beyond that point cannot be found. 

The bottom part of Figure \ref{fig:BifRN}b) plots the eigenvalues for these MOTSs through this process. The two non-spherical MOTSs are geometrically identical (just reflections of each other) and so have the same eigenvalues. As expected at the bifurcation point, all MOTSs have a vanishing eigenvalue. In this case it is the second one and the principal eigenvalue is negative throughout, and so all of the involved MOTSs are unstable. 

As might be expected for a pitchfork, this bifurcation is associated with a symmetry of the system (reflection across the equatorial plane). Note too that the two bifurcating MOTSs straddle the unstable Killing horizon. These are our first examples of an unstable MOTS that is crossed by a marginally stable MOTS, as allowed by Lemma \ref{lem:maxmin}. Here $\mathcal{N}_+$ runs from the inner to the outer horizons while $\mathcal{N}_-$ runs from the inner horizon to the singularity.

\subsubsection{Transition to three negative eigenvalues: a transcritical bifurcation}
\label{sec:RN_innerpf}

Finally, in Figure \ref{fig:RN_combo}a), the top row shows a (green) MOTS that is partially inside and partially outside the inner horizon but which approaches it as $\frac{q}{m} \rightarrow \frac{\sqrt{7}}{4}$ from below. For $\frac{q}{m} > \frac{\sqrt{7}}{4}$ a (gold) MOTS then departs from the inner horizon. 

From the top and bottom of Figure \ref{fig:BifRN}a)  it can be seen that this is a transcritical bifurcation, with the inner horizon starting with three negative eigenvalues and finishing with two, while the bifurcating axisymmetric MOTS started with two and finished with three. This evolution was observed previously in Figure 16 of \cite{Hennigar:2021ogw}. However, it wasn't recognized in that paper as an example of a transcritical bifurcation. 

The bifurcating MOTS that straddles the inner horizon provides a second demonstration that an unstable MOTS is no longer a barrier to weakly (un)trapped surfaces. 

\subsubsection{Further bifurcations}

Preliminary calculations indicate that higher vanishing eigenvalues correspond to pitchfork (for $n$ odd) or transcritical (for $n$ even) bifurcations. %We have also tracked the non-spherical MOTSs from the pitchfork and transcritical bifurcations as they evolve away from the inner horizon and seen them go through further bifurcations. 
However, we will leave detailed investigations and classifications of these events for future work.

\section{Axisymmetric bifurcations II: Weyl-distorted Schwarzschild}
\label{sec:weyl}

We now turn to our final, most computationally involved, set of examples. Weyl-distorted Schwarzschild spacetimes are black holes that are distorted by arbitrarily strong external fields \cite{Geroch:1982}. 
They are exact solutions with the possible distortions decomposed as  multipole expansions. There is a persistent Killing horizon at $r=2m$ no matter how strong the distortions become. 

For small distortions the Killing horizon is stable, but it was shown in \cite{Pilkington:2011aj} that it does not always remain stable.\footnote{This was not shown directly using the stability operator but rather a somewhat complicated argument that linked the sign of deformations in the inward null direction to that of the inward null expansion and then showed that that expansion was not everywere positive.} In this section we return to these solutions for a more detailed study of the bifurcations that occur as $r=2m$ becomes progressively more unstable.

\subsection{Weyl-distorted Schwarzschild in static coordinates}

In the standard static coordinate system, the Weyl-distorted Schwarzschild spacetime takes the form 
\begin{equation}
\dd s^2 = -e^{2U} F \dd t^2 + e^{-2U + 2V}\left[\frac{\dd r^2}{F} + r^2\,\dd \theta^2\right] + e^{-2U}\,r^2\,\sin^2\!\theta\,\dd \phi^2 \, , 
\label{metric}
\end{equation}
where $F(r) = 1 - \frac{2m}{r}$.
From the vacuum Einstein equations, from $R_{tt} = 0$ (and $R_{\phi \phi} = 0$) 
the distorting potential $U(r,\theta)$ satisfies 
\begin{equation}
  \frac{\partial}{\partial r} \left(r^2 F U_r \right)  + \frac{1}{\sin \theta} \frac{\partial}{\partial \theta} \Big(\sin \! \theta \, U_\theta \Big) = 0  \label{lap}
\end{equation}
where subscripts $r$ and $\theta$ indicate partial derivatives. Then, first-order equations for $V(r,\theta)$ in terms of $U(r,\theta)$ can be found by  
solving $R_{rr}=0$, $R_{r \theta}=0$ and $R_{\theta \theta}=0$\footnote{Solve for $V_r$ and $V_\theta$ from  $r(r-2m) R_{rr} -R_{\theta \theta}=0$ and $R_{r \theta}=0$, which are linear in $V_r$ and $V_\theta$.} to obtain
\begin{align}
V_r & = \frac{1}{\Xi}
\big( (r-m) A \sin \! \theta + B \cos \! \theta \big) \label{Vr}  \\
V_\theta & =  \frac{1}{\Xi} \big( (r-m)B \sin \! \theta - r^2 F A \cos \! \theta \big)  \label{Vt}
\end{align}
where 
\begin{align}
    \Xi &= (r-m)^2-m^2 \cos^2 \! \theta \\
    A & =  \left(r^2 F U_r^2 + 2 m U_r - U_\theta^2\right) \sin \! \theta\\
    B & = \left(2r^2 F U_r + 2m \right) U_\theta \sin \! \theta \; . 
\end{align}
There is a degeneracy in the components of the Einstein equations: all five non-trivial components of the Ricci tensor vanish if 
(\ref{lap}), (\ref{Vr}) and (\ref{Vt}) are satisfied.

Equation (\ref{lap}) has an exact series solution
\cite{Breton:1997,Frolov:2007}:
\begin{equation}
\label{eq:U}
U(r,\,\theta) =  \sum_{n = 1}^{\infty} \alpha_n \left( \frac{R}{m} \right)^n \mathcal{P}_n  \, ,  
\end{equation}
where the $\alpha_i$ are multipole coefficients, 
\begin{equation}
R =  \left[ \left(1 - \frac{2m}{r}\right)r^2 + m^2\cos^2\!\theta \right]^{1/2} \, , 
\end{equation}
and the $\mathcal{P}_n$ are defined in terms of Legendre 
polynomials $P_n$,
\begin{equation}
\mathcal{P}_n  = P_n\left(\frac{(r - m)\cos\theta}{R}\right) \, . 
\end{equation}
Given this form of $U$, the differential equations (\ref{Vr}) and (\ref{Vt}) for $V$ can 
also be integrated to obtain \cite{Breton:1997,Frolov:2007}
\begin{align}
\label{eq:V}
V(r,\,\theta) &=  \sum_{i = 1}^{\infty} \frac{i\alpha_i^2 }{2} \left( \frac{R}{m}\right)^{2i} \left(P_i^2 - P_{i - 1}^2\right)\\
& \quad  - \frac{1}{m} \sum_{i = 1}^{\infty}\alpha_i\,\sum_{j = 0}^{i - 1}  \Big[(-1)^{i + j} \left(r - m\left(1 - \cos\theta\right)\right)  
 + r - m\left(1 + \cos\theta\right) \Big] \left( \frac{R}{m} \right)^j P_j \; . \nonumber  
\end{align}

These expressions are complicated but take simple forms on $r=2m$. It is straightforward to check that 
$R(2m,\theta) = m \cos \theta$ and so $\mathcal{P}_n(2m, \theta) = 1$. Then 
\begin{equation}
    \bar{U}(\theta) := U(2m, \theta) = \sum_{n=1}^\infty \alpha_n \cos^n \! \theta
\end{equation}
and 
\begin{equation}
    \bar{V}(\theta) := V(2m, \theta) = 2 \bar{U}(\theta) - 2 u_o  \label{Vbar}
\end{equation}
where 
\begin{equation}
    u_o = \sum_{n=1}^\infty \alpha_{2n} \label{u0} \; . 
\end{equation}
To avoid conical singularities along the $z$-axis, the odd coefficients must also satisfy
\begin{equation}
        \sum_{n=1}^{\infty} \alpha_{2n - 1} = 0\; . \label{multipoleCon}
\end{equation}
Note that 
these conditions restrict the possible 
distortions that we can consider. While we 
can study pure $2n$-pole solutions (such as 
quadrupole), pure dipole, octopole or other
odd-poles all have conical singularities. Odd-pole distortions must include
at least two terms to avoid conical singularities.

Thus the simplest possible distortion that we can consider is pure quadrupole. Then $u_o = \alpha_2$ is non-zero but all other
expansion coefficients vanish. Dropping the subscript, the quadrupole distortion potential is
\begin{equation}
    U_{\mathrm{quad}}(r,\theta) = \alpha \left(  \cos^2 \! \theta 
+\frac{ r \left(r-2m\right) }{2m^2  }  \left(3 \cos^{2}\! \theta -1\right) \right) \; . 
    \label{UQ}
\end{equation}
This relatively simple $U_{\mathrm{quad}}$ is perhaps 
surprising given the apparently complicated form of (\ref{eq:U}). However it turns out that in this case (and in fact all other cases) the square root terms all cancel and only a polynomial in $r$ and $\cos \theta$ remains. 

The pure quadrupole term for $V(r,\theta)$ is also perhaps
simpler than might be expected. From (\ref{eq:V}) we have
\begin{align}
    V_{\mathrm{quad}} (r,\theta) 
     &= - 2 \alpha \left( \frac{r-m}{m} \right)  \sin^2 \! \theta \label{VQ}\\
     & \qquad + \alpha^2 \sin^2 \!\theta \left( \frac{r(r-2m)}{4m^4} \right) 
     % \nonumber \\
     % & \quad \times 
     \Big(
     -(4m-3r)(2m-3r) \cos^2 \!\theta  + r(r-2m)\Big) \; .  
   \nonumber
\end{align}

In the next section we return to general potentials, but when the time comes to focus on specific solutions, we will use the quadrupole distortions described by (\ref{UQ}) and (\ref{VQ}).

\subsection{Switching to Painlev\'e--Gullstrand-like coordinates}
The above form
of the metric has a coordinate singularity at $r=2m$, just like undistorted Schwarzschild. 
This singularity may be removed
with a generalized Painlev\'e--Gullstrand-like coordinate 
transformation. It would also be possible to do this in Eddington--Finkelstein-like coordinates, as in 
\cite{Fairhurst:2001, Pilkington:2011aj}, but
since we need spacelike time 
slices in order to use the MOTSodesic method, it will be more
convenient to use the PG-like coordinate system here. 

Following \cite{Hennigar:2021ogw}, we define 
a new coordinate
\begin{equation}
    T = t + e^{-2u_o} \int_{2m}^r {\frac{\sqrt{1-p(r) F}}{F}} \, \dd r
\end{equation}
where $p(r)$ is an arbitrary function. Then (\ref{metric}) becomes
\begin{align}
\mathrm{d} s^2 =&  - e^{2U}F  \dd T^2 
+ 2 e^{2U-2u_o} \sqrt{1-pF}  \dd T \dd r \label{PGmetric} \\
& + e^{2V-2U} (G \,\dd r^2 + r^2 \, \dd \theta^2) 
 + r^2e^{-2U} \sin^2 \! \theta \,\dd \phi^2 
\nonumber
\end{align}
where 
\begin{equation}
   G =  {e}^{4 U -2 V -4 u_{o} } p +\frac{1-{e}^{4 U -2 V -4 u_{o} }}{F} \; .  
\end{equation}
The denominator of the second term vanishes at $r=2m$ but by (\ref{Vbar}) so does the numerator. 
 Thus to find the limit we need to apply l'H\^{o}pital's rule:
\begin{align}
    \bar{G} = \! \lim_{r \rightarrow 2m} \! G \label{grr}
   & =  p(2m)  +  \lim_{r \rightarrow 2m} \left( \frac{r( 1-e^{4U-2V-4u_o}}{r-2m} \right)  \\
    & = p(2m) - \lim_{r \rightarrow 2m} \left( r (4U_r - 2 V_r) e^{4U-2V-4u_o}\right) \nonumber\\
    & = p(2m) + 8 \cot \! \theta \bar{U}_\theta - 4 \bU_\theta^2  
   \nonumber
\end{align}
where between the second and third lines we have applied (\ref{lap}) and (\ref{Vr}) to find $U_r$ and $V_r$ on $r=2m$. Hence, $g_{rr}$ is well-defined at $r=2m$ and there is no coordinate singularity. 

With spherical symmetry broken, the interpretation of $p(r)$ is not so physical as in Section \ref{sec:RNmetric}, but it remains a function that characterizes the behavior of the $\Sigma_T$-slicing.% rather than anything physical. 

In future subsections we will be working with pure quadrupole distortions with values of the distortion parameter $-9 < \alpha < 5$. As we shall see, at the outer ranges the geometry is far from being spherically symmetric and it is necessary
to carefully choose non-zero $p(r)$ to keep $\Sigma_T$ spacelike at least in some neighbourhood of $r=2m$. Happily for the ranges in which we are interested, this can be managed with constant $p(r)=p_o$. Depending on the exact problem we will have $15 \leq p_o \leq 40$. The expression for $g_{rr}$ is fairly complicated so it is easiest to confirm that these are appropriate values by numerical experimentation: simply plotting $g_{rr}$.

\subsection{Horizon geometry and stability operator}
\label{sec:HGSO}

The metric components of (\ref{PGmetric}) are independent of $T$
and so it is immediate that $\frac{\partial}{\partial T}$ is a
Killing vector field. 
Now consider the three-surface $H$ defined by $r=2m$. 
$\frac{\partial}{\partial T}$ is tangent to $H$ and is also null on that surface since 
$F(2m) = 0$, hence $H$ is a Killing horizon and will remain so for any value of $\alpha$. 
It follows that it is also an isolated horizon and that any cross-section 
of $H$ is a MOTS
\cite{Ashtekar:2000a, Ashtekar:2000b, Ashtekar:2001jb}.

We now examine the geometry of these MOTSs. Let $\mathcal{S}$ be any compact, axisymmetric,
two-dimensional slice of 
$r=2m$ that can be parameterized as
\begin{equation}
 (T,r,\theta, \phi) = (S(\vartheta) , 2m, \vartheta, \varphi)
 \quad \mbox{for} \quad 0 < \vartheta < \pi, \ - \pi < \varphi < \pi  \; . 
\end{equation}
Tangents to this surface push forward to vectors in 
the full spacetime as
\begin{equation}
e_\vartheta = \frac{\partial}{\partial \vartheta} =  S_\vartheta \frac{\partial}{\partial T} + \frac{\partial}{\partial \theta} 
\quad \mbox{and} \quad
e_\varphi=  \frac{\partial}{\partial \varphi} = \frac{\partial}{\partial \phi} \, ,
\end{equation}
and so the induced metric on $\mathcal{S}$
is 
\begin{align}
q_{AB} \dd \vartheta^A  \dd \vartheta^B & = g_{ab} e^a_A e^b_B  \dd \vartheta^A  \dd \vartheta^B \\
& =   4m^2 e^{-2\bU} \! \! \left(e^{4\bU-4u_o} \dd \vartheta^2 + \sin^2 \! \vartheta \, \dd \varphi^2  \right)   \; . \nonumber
\end{align}
As would be expected for an isolated horizon, the induced metric is independent of the choice of cross-section.

Using the tangent vectors $e_\vartheta$ and $e_\varphi$, it is straightforward to see that the normal space to $\mathcal{S}$  is spanned by $\dd T - S_\vartheta \dd \theta$ and  $\dd r$. From these we can construct the general form of the axisymmetric future-pointing null normal vector fields
\begin{equation}
    l_+ = e^{\beta} \frac{\partial}{\partial T} \, ,
   \qquad  l_- = e^{-\bU-\beta} \left( \frac{e^{4 u_o}  S_\theta^2 + 4 \bar{G} m^2}{8 m^{2}}  \frac{\partial}{\partial T} 
    - e^{2u_o} \frac{\partial}{\partial r}
+ \frac{ e^{4 u_{o}} S_\theta}{4 m^{2}}  \frac{\partial}{\partial \theta} \right) \, , 
\end{equation}
where $\beta= \beta(\vartheta)$ is a free function. 
% They form a basis for the normal space $T^\perp \mathcal{S}$ to $\mathcal{S}$. 
% They are cross-normalized so that $l_+ \cdot l_- = -1$. 
By direct calculation it is  straightforward to confirm that
\begin{equation}
    \theta_+ = 0 \; \; \mbox{and} \quad \sigma^+_{AB} = 0 
\end{equation}
irrespective of $\beta(\vartheta)$ or $S(\vartheta)$
(as they must be for an isolated horizon). 

The connection on the normal bundle is
\begin{equation}
    \omega^{l_+} = \left(\beta_\vartheta + \bar{U}_\vartheta + \frac{e^{2u_o} S_\vartheta}{4m} \right) \dd \vartheta \; . 
\end{equation}
This depends on both $\beta$ and the foliation $S(\vartheta)$, however
it is clearly a total derivative. In the dyad defined by 
\begin{equation}
    \beta = - \bU - \frac{e^{2u_o} S}{4m}  \label{betatriv}
\end{equation}
the connection is trivial, so we conclude that the stability of the cross-section  is independent of $S(\vartheta)$. This point was first noted in the less systematic discussion of \cite{Pilkington:2011aj}. An analogous result was shown for general Killing horizons in \cite{Mars:2012sb}: Proposition 4 from that paper says that the stability operators for different slices of a Killing horizon are similar (and so have the same eigenvalues).

Next, we consider the stability operator for these MOTSs.
We write the $q_{AB}$-compatible covariant derivative as $\mathcal{D}_A$. Then for a function $\psi(\vartheta,\varphi)$, the Laplacian on $\mathcal{S}$ is
\begin{equation}
\mathcal{D}^2 \psi = \frac{1}{4m^2}\left(  e^{-2 \bar{U} + 4 u_o} \left(\psi_{\vartheta \vartheta} + (\cot \theta - 2 \bU_\vartheta) \psi_\vartheta  \right) + \frac{e^{2 \bU} }{\sin^2 \vartheta} \psi_{\varphi \varphi} \right) \label{D2}
\end{equation}
and 
the Gauss curvature is
\begin{equation}
    \mathcal{K} = \frac{e^{-2\bar{U} + 4 u_o}}{4m^2} \left(1 + \bar{U}_{\vartheta \vartheta}  + 3  \bU_\vartheta \cot \vartheta - 2 \bar{U}_\vartheta^2  \right) \; .  \label{cK}
\end{equation}
All such $\mathcal{S}$ are slices of an isolated horizon and this is a vacuum spacetime. Hence, by the discussion of
Section \ref{sec:nonex_stab} the stability of $\mathcal{S}$ is fully 
determined by the self-adjoint operator
\begin{equation}
    L_\cS(l_+, -l_-) = - \mathcal{D}^2 \psi + \mathcal{K} \psi \;. \label{eq:KillStab}
\end{equation}
%which takes this simplest form and is explicitly self-adjoint when expressed relative to the (\ref{betatriv})-dyad.  
Other choices of $\ell$ or $X$ will change the numerical values of the eigenvalues but not their signs: the $(n_0,n_-)$ stability classification is invariant. 

% This subsection demonstrated that not only is the $(n_0,n_-)$ stability classification of the Killing horizon independent of the normal dyad 
% $(l_+, l_-)$, $\ell$  and direction of deformation $X$, but it is also 
% independent of the particular $S(\vartheta)$ slicing as well as being 
% completely indifferent to any foliation chosen (or not) for the full spacetime. 

\subsection{Stability of the quadrupole distorted Weyl--Schwarzschild Killing horizon}

 For all the examples up to this point we have considered bifurcations from spherical MOTSs and so it has been easy to identify the points of bifurcation. The current example is not spherically symmetric and so we must numerically calculate the spectrum of (\ref{eq:KillStab}). Here we use both the pseudo-spectral techniques of Appendix \ref{app:spectra} as well as an independent spectral method. 

 We recall from Section \ref{sec:slice}  that a non-rotating MOTS is said to be
 \emph{$(n_0, n_-)$-unstable} if its stability operator has $n_0$ vanishing eigenvalues and $n_-$ negative eigenvalues. 

%We have independently confirmed these results with a variety of other calculations including finite element \cite{MaryThesis} as well as using other scalings of the $(l_+, l_-)$ dyad, directions $X$ and null normals $\ell$.  In the interest of (relative) brevity, most of these are omitted. 

\subsubsection{Simplified stability operator}

From, (\ref{UQ}), (\ref{D2}) and (\ref{cK}), (\ref{eq:KillStab}) becomes
\begin{align}
\label{L:full}
     L(l_+, -l_-) & = - \mathcal{D}^2 \psi + \mathcal{K} \psi  \\
     & =  \frac{e^{-2\bU + 4u_o}}{4m^2} \bigg(
 -\psi_{\vartheta \vartheta} - (\cot \theta - 2 \bU_\vartheta) \psi_\vartheta 
  + \left(1 + \bar{U}_{\vartheta \vartheta}  + 3  \bU_\vartheta \cot \vartheta - 2 \bar{U}_\vartheta^2  \right) \psi  \bigg)\nonumber \\
 & \qquad  - \frac{e^{2 \bU}}{4m^2 \sin^2 \! \theta} \psi_{\varphi \varphi} \nonumber \; . 
\end{align}
We are only worried about the $(n_0, n_-)$-stability classification, and so can equally well study the spectrum of the slightly simpler
\begin{align}
     L(\ell, -l_-) & =
 -\psi_{\vartheta \vartheta} - (\cot \theta - 2 \bU_\vartheta) \psi_\vartheta 
 \nonumber  \\ &
  \qquad + \left(1 + \bar{U}_{\vartheta \vartheta}  + 3  \bU_\vartheta \cot \vartheta - 2 \bar{U}_\vartheta^2  \right) \psi   - \frac{e^{4 \bU-4u_o}}{ \sin^2 \! \theta} \psi_{\varphi \varphi}  \, ,  
\end{align}
where $\ell = 4m^2 e^{2\bU - 4u_o} l_+$. Calculations can (and have) been done keeping the exponential term on the main part of the operator. However, it increases the difficulty of the numerical calculations, as for larger $\alpha$ that term varies by several orders of magnitude between the poles and the equator.

Our focus is on the axisymmetric eigenfunctions $\psi(\theta)$ and their associated eigenvalues. Then, further specializing to quadrupolar distortions
\begin{equation}
    \bar{U} = \alpha \cos^2 \theta  \qand u_o = \alpha \, ,
\end{equation}
we study the spectrum of 
\begin{align}
    \tilde{L} = &   -\psi_{\vartheta \vartheta} - \Big(\cot \vartheta + 4 \alpha \sin \vartheta \cos \vartheta \Big) \psi_\vartheta \label{eq:Ltilde}\\
&  \quad + \Big(1 + 2 \alpha (1 - 5 \cos^2 \vartheta) - 8 \alpha^2 \sin^2 \vartheta \cos^2 \vartheta \Big) \psi \; .  \nonumber
\end{align}

\subsubsection{Quadrupolar geometry}

To get a feeling for the geometry of these quadrupolar distortions,
note that 
on $r=2m$ the equatorial circumference is independent of $\alpha$,
\begin{equation}
    C = 4 \pi m \,,
\end{equation}
while the north to south pole meridian length is 
\begin{equation}
    L = 2m e^{-2 \alpha} \int_0^\pi   e^{\alpha \cos^2 \vartheta } \dd \vartheta  \; . 
\end{equation}
A graph of the ratio of these quantities is shown in Figure \ref{fig:WeylDef}. The general effect is that negative $\alpha$ stretches the MOTS in the north-south direction while positive $\alpha$ flattens it. For $\alpha = -4$, twice the north-south distance is almost $1000$ times greater than the circumference, while for $\alpha = 4$ it is less than $\frac{1}{100}$ times smaller.

The area of the MOTS is
\begin{equation}
    A = 16 \pi m^2 e^{-2 \alpha} \,,
\end{equation}
thus negative/positive $\alpha$ respectively 
expand/shrink the total area. The Gauss curvature is
\begin{equation}
    \mathcal{K}_{\mathrm{quad}} = \frac{e^{\alpha  \left(3-\cos \! 2 \vartheta \right)}}{4m^2}\left(-2 \alpha^{2}\sin^{2} 2 \vartheta  - \alpha \left( 5   \cos \! 2 \vartheta + 3 \right)  +1\right) \; . 
\end{equation}
For small magnitude $\alpha$ this is everywhere positive but for $\alpha > \frac{1}{8}$ a negative curvature region develops at the poles. This is also the $\alpha$ at which the MOTS can no longer be embedded in $\mathbb{R}^3$: the north-south distance at the poles must increase faster than can be accommodated in flat space. For $\alpha < - \frac{1}{2}$ a negative curvature region develops at the equator. See Figure \ref{fig:WeylDef} for representative MOTSs embedded in $\mathbb{R}^3$. 

\begin{figure}
    \centering
\includegraphics{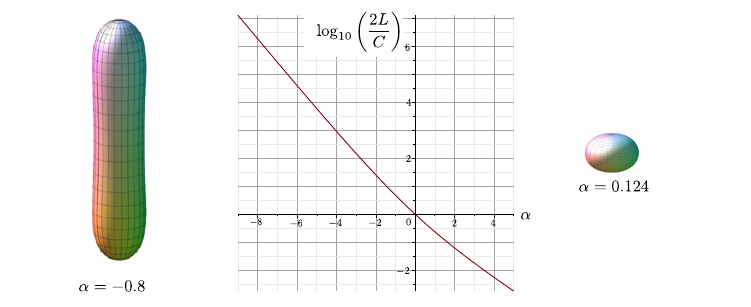}
    \caption{Deformation of the MOTS at $r=2m$ for quadrupolar Weyl--Scharzschild. The log of the ratio of twice the north-south meridian length to the equatorial circumference is plotted. Embedding diagrams for exemplary negative and positive $\alpha$ are respectively shown to the left and right (approximately to scale). For $\alpha \geq \frac{1}{8}$ the MOTS can no longer be isometrically embedded in Euclidean $\mathbb{R}^3$. }
    \label{fig:WeylDef}
\end{figure}

\subsubsection{Pseudo-spectral eigenvalue calculation}

The lowest eigenvalues in the  spectrum of  $\tilde{L}$, calculated using the methods \ref{app:spectra}, are shown in Figure \ref{fig:Ltilde}. The numerical calculation used 40 Chebyshev polynomial basis elements and the associate spectra were calculated at a resolution of $\Delta \alpha = 0.02$. 
As a check on the accuracy of this approximation, the zeros of the eigenvalue curves shown differ from those calculated with 20 basis elements by at most $0.02$ (for the outermost four\,---\,the inner four are indistinguishable). As we are mainly interested in the number of negative eigenvalues, this is sufficient for our purposes. 

Thus, within an accuracy of about $0.02$ we can identify the first few regions of $(n_o,n_-)$ stability for the Weyl--Schwarzschild Killing horizon. These are identified by colour in Figure \ref{fig:Ltilde}. 
We expect bifurcations from $r=2m$ 
at the points where eigenvalue curves cross zero.

\begin{figure}
    \centering
    \includegraphics[width=1.0\linewidth]{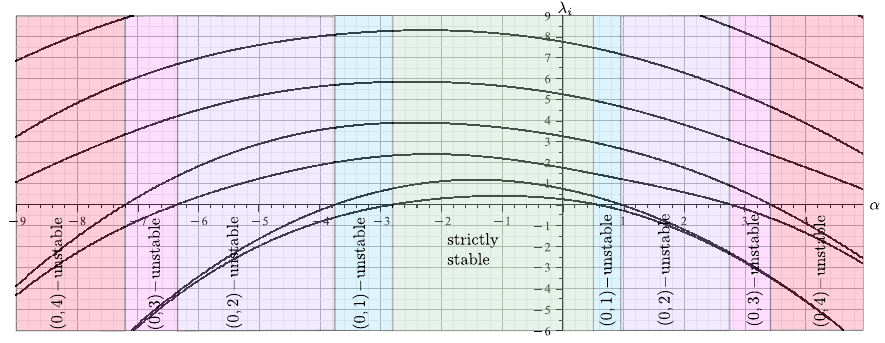}
\caption{Eigenvalues and stability ranges for the axisymmetric stability operator $\tilde{L}$ (\ref{eq:Ltilde}).
40 Chebyshev polynomial basis functions were used in the calculation, which was run at a resolution of $0.02$ in $\alpha$. 
}
\label{fig:Ltilde}
\end{figure}

\subsubsection{Vanishing eigenvalues: spectral calculation}

We could improve the accuracy of the location of the zeros by improving the resolution of the search. However, there is a more algebraic way to find them in which numerics are only needed to calculate the roots of a large polynomial.

We proceed with a spectral method, which is in spirit the same as the pseudo-spectral methods. However, in this case we expand functions with a cosine basis,
\begin{equation}
    \psi (\theta) = \sum_{q=0}^\infty \mathcal{C}_q \cos (q \theta) \, ,  
\end{equation}
and then do all integrals exactly using a computer algebra package (we used Maple \cite{maple}). This time the basis functions are orthogonal:
\begin{equation}
\mathcal{I}_{pq} = \int_0^\pi \cos (p\theta) \cos (q \theta) \dd \theta = \left\{\begin{array}{ll} 
    \pi &  \mbox{if } p = q = 0\\
    \frac{\pi}{2}  & \mbox{if } p=q \neq 1 \\
 0 & \mbox{otherwise}
    \end{array}  \right.
\end{equation}
and so the elements of the $N$-dimensional matrix approximation to $\tilde{L}$ are
\begin{align}
    \tilde{L}_{pq} =  \frac{a}{\pi} \int_0^\pi \! \! \! \!   
     \cos (p \vartheta) \Bigg( & q^2 \cos (q \vartheta) + q \Big(\cot \vartheta + 4 \alpha \sin \vartheta \cos \vartheta \Big) \sin (q \vartheta) \nonumber\\
& \ + \Big(1 + 2 \alpha (1 - 5 \cos^2 \vartheta) - 8 \alpha^2 \sin^2 \vartheta \cos^2 \vartheta \Big) \cos (q \vartheta)  \Bigg) \dd \vartheta \,, \label{eq:Ltildepq}
\end{align}
where $a=1$ for $p = 0$ and otherwise $a=2$. Doing these integrals in closed form is not quick (on a fast 2020 laptop it took a couple of days to do them all) but it does avoid numerical issues associated with integrating rapidly oscillating functions and ultimately returns each element of  $\tilde{L}_{pq}$ as a polynomial in $\alpha$. 

%The determinant of this matrix of polynomials is zero if and only if one or more eigenvalues vanishes. But, that determinant 
The determinant of this matrix is a $2N$-degree polynomial in $\alpha$ and its zeros are easily found by the root-solving algorithms built into any computer algebra package. A second advantage of this method is that once the $N \times N$ matrix of polynomials has been found, it is easy to calculate the estimates for $m \times m$ with $m < N$: this is just a subset of the larger matrix. Hence it is easy to examine the approach of the zeros to particular $\alpha$ as $N$ is increased. 
This limiting behaviour for the 
first four values of $\alpha$ corresponding to a vanishing eigenvalue on either side of $\alpha = 0$ is shown in Table \ref{tab:zeros}. They are clearly consistent with the zeros seen in Figure \ref{fig:Ltilde} but identified with considerably higher precision. 

\begin{table}
    \small
    \centering
    \begin{tabular}{|c|rrrrrrrr|}
    \hline
     N & \multicolumn{7}{c}{Values of $\alpha$ for which an eigenvalue of $\tilde{L}_{pq}$ vanishes} & \\
    \hline
      4 & -11.953951& - 8.848182& -2.699315& -2.101266 & 0.530775 & 0.941374& 1.418673& 1.711892\\
      8 & -34.113475& -28.258243 & -3.695085 &  -2.804439 & 0.557777 & 0.956615 & 1.870187 & 3.352650  \\
      12 & -7.392968 & -6.031000 & -3.744000 & -2.798910 & 0.557798 & 0.956652 &  2.762204 & 3.422071\\
      16 &  -7.096823 & -6.301235 & -3.741615 & -2.798086 & 0.557798 & 0.956652 & 2.762154 & 3.422026\\
      20 & -7.196978 & -6.349868 & -3.741561 &  -2.798084 & 0.557798 & 0.956652 &  2.762146 &3.422014 \\
      24 & -7.200478 & -6.349351 & -3.741562 & -2.798084 & 0.557798 & 0.956652 & 2.762146 & 3.422014\\
      28 & -7.200255 & -6.349264 & -3.741562 & -2.798084 & 0.557798 & 0.956652 & 2.762146 & 3.422014\\
      32 & -7.200250 & -6.349264 & -3.741562 & -2.798084 & 0.557798 & 0.956652 & 2.762146 & 3.422014\\
      36 & -7.200250 & -6.349264 & -3.741562 & -2.798084 & 0.557798 & 0.956652 & 2.762146 & 3.422014\\
  %    40 & -7.200250 & -6.349264 & -3.741562 & -2.798084 & 0.557798 & 0.956652 & 2.762146 & 3.422014\\
%      44 & -7.200250 & -6.349264 & -3.741562 & -2.798084 & 0.5577977 & 0.9566516 & 2.762146 & 3.422014\\
      \hline
    \end{tabular}
    \caption{Values of $\alpha$ for which $\det (\tilde{L}_{pq}) $ vanishes. In such cases there is a vanishing eigenvalue. The left-hand column records the number of basis functions used in the approximation. After $N=36$, these $\alpha$ are unchanging up to six decimal places (and for the smaller magnitude ones many more). }
    \label{tab:zeros}
\end{table}

With the vanishing eigenvalues found, we turn to looking for nearby MOTSs. Given that the 
Weyl--Schwarzschild solution is significantly more complicated than previous examples, this is also more complicated. 

\subsection{MOTSodesic equations}
We start with the geometry of the slices. 
Surfaces of constant $T$ have induced metric
\begin{equation}
    h_{ij} \dd x^i \dd x^j =   e^{2V-2U} (G \,\dd r^2 + r^2 \, \dd \theta^2) 
 + r^2e^{-2U} \sin^2 \! \theta \,\dd \phi^2     \; .   
\end{equation}  
The forward-in-time oriented unit normal to these surfaces is
\begin{equation}
    u_\alpha \dd x^\alpha = - \frac{e^U}{\sqrt{G}} \dd T
\end{equation}
as a one-form or 
\begin{equation}
    u^\alpha \left( \frac{\partial}{\partial x^\alpha} \right)  =  e^{-U} \sqrt{G} \frac{\partial}{\partial T} - e^{3U-2V-2u_o}\sqrt{\frac{1-p_o F}{G}}  \frac{\partial}{\partial r}
\end{equation}
as a vector. Then, the extrinsic curvature on a surface of constant $T$ is
\begin{align}
    K_{ij} \dd x^i \dd x^j &=  e_i^\alpha e_j^\beta \nabla_\alpha u_\beta \\
   &=   \frac{e^{U-2u_o}}{2\sqrt{G}} \Big(
    \left(-6 hG U_{r}  +2 hG V_{r}  + h G_{r}  -2 G h_{r}\right) \dd r^2 \nonumber \\
    & \qquad \qquad \quad +2 h\left(-4 G U_{\theta}+2 G V_{\theta}+G_{\theta}\right)  \dd r \,\dd \theta \nonumber \\
    & \qquad \qquad \quad +  2 r h \left(U_{r} r -V_{r} r -1\right) \dd \theta^2 + \left(r U_{r}  -1\right) \dd \phi^2  \Big)\nonumber \; . 
\end{align}

In the usual way we look for MOTSodesics  parameterized by arclength: $r=P(s)$ and $\theta = \Theta(s)$ so that 
\begin{equation}
    g_{ij} T^i T^j = e^{2V-2U} (G \dot{P}^2 + P^2 \dot{\Theta}^2) = 1 \; . 
\end{equation}
Here and in the rest of this subsection  it is understood that any $U$, $V$, $G$, $h$ or any derivatives of these quantities should be evaluated at $(r,\theta) = (P(s),\Theta(s))$.
The unit normal to the corresponding surface of revolution is
\begin{equation}
    N = \frac{P \dot{\Theta}}{\sqrt{G}} \frac{\partial}{\partial r} - \frac{\sqrt{G} \dot{P}}{P} \frac{\partial}{\partial \theta} \;  . 
\end{equation}

From (\ref{eq:kappa}) the MOTSodesic acceleration is
\begin{align}
\kappa   &= -   N^i a_i^{\hat{e}_\phi}  + 
    q^{ij} K_{ij} \\
     &= \left( \frac{\left(2 h G V_{r}+G_{r} h -6 h G U_{r} -2 G h_{r}\right) {e}^{U -2 u_{o}} }{2 \sqrt{G}} \right) \dot{P}^2 \nonumber \\
   & \qquad +\left(\frac{ h \left(G_{\theta} + 2G\left(V_\theta - 2U_{\theta} \right)  \right) {e}^{U -2 u_{o}} }{\sqrt{G}} \right) \dot{P} \dot{\Theta} + \left( \frac{\sqrt{G}\, \left(U_{\theta}-\cot \! \Theta \right)}{P}\right) \dot{P} \nonumber \\
  & \qquad + \left( \frac{P h \,{e}^{U -2 u_{o}} \left(P U_{r}  -P V_{r}  -1\right) }{\sqrt{G}} \right) {\dot{\Theta}}^{2} -\frac{\left(P U_{r}  -1\right) \dot{\Theta}}{\sqrt{G}}+\frac{h \,{e}^{3 U -2 u_{o} -2 V} \left(P U_{r}  -1\right)}{r \sqrt{G}} \nonumber
\end{align}
and the relevant contracted Christoffel symbols are
\begin{align}
\Gamma^{r}_{\alpha \beta} T^{\alpha} T^\beta  &=  \left(\frac{ G V_{r} -G U_{r}+ \frac{1}{2} G_{r}}{G} \right) \dot{P}^{2}
+ 2 \left( \frac{ G V_{\theta} - G U_{\theta}  + \frac{1}{2} G_{\theta}  }{G} \right)\dot{P} \dot{\Theta} \nonumber \\
& \qquad +P \left( \frac{ PU_{r}  -P V_{r}  -1 }{G} \right) {\dot{\Theta}}^{2}
\end{align}
and
\begin{align}
\Gamma^{\theta}_{\alpha \beta} T^{\alpha} T^\beta &= -\left( \frac{ G V_{\theta}-G U_{\theta}  +\frac{1}{2} G_{\theta}}{P^{2}} \right) \dot{P}^{2} -2 \left( \frac{ PU_{r}  -PV_{r}  -1 }{P}\right)  \dot{P} \dot{\Theta}  +\left(V_{\theta}-U_{\theta}\right) {\dot{\Theta}}^{2} \; . 
\end{align}
Substituting in $U_{\mathrm{quad}}$ and $V_{\mathrm{quad}}$, we solve these numerically  via the methods of \ref{app:MOTSo}.

\subsection{MOTSodesics near the bifurcation points}

We then look for MOTSodesics near these bifurcation points using the methods of \ref{app:MOTSo}. 

\begin{figure}
    \centering
    \includegraphics[scale=0.88]{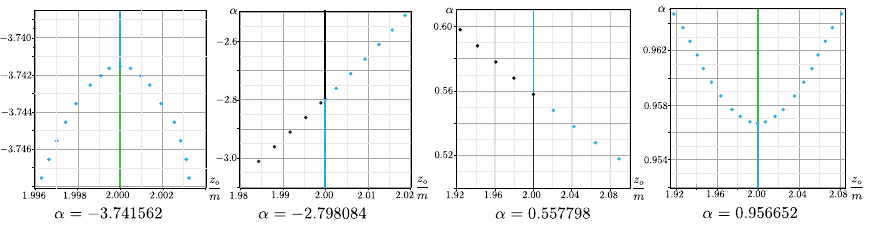}
    \caption{The first two bifurcations from the Weyl--Schwarzschild horizon for $\alpha>0$ and $\alpha< 0$. The first one on either side of zero is a transcritical bifurcation while the second is a pitchfork. $z_o$ is the point at which the MOTS intersects the $z$ axis. The color scheme is the same as in Figure \ref{fig:BifRN}: black, blue and green denote MOTSs which have respectively zero, one and two negative eigenvalues of the stability operator. }
    \label{fig:Weyl4}
\end{figure}

\subsubsection{The $(1,0)$-unstable bifurcations}
\label{sec:weyltrans}

We begin with the bifurcation at $\alpha \approx 0.557798$ and for this choose $p_o = 15m$. This keeps the surfaces of constant $T$ spacelike in the range $1.4 m < r < 2.6m$ for that value of $\alpha$. 

The nearby MOTSs are shown in Figure \ref{fig:WeylTrans} and maximum $r$ intersections of those MOTSs with the $z$-axis are plotted in Figure \ref{fig:Weyl4}. As $\alpha$ increases, a
$(0,1)$-unstable MOTS approaches from the outside, becomes $(1,0)$-unstable when it meets $r=2m$ at $\alpha \approx 0.557798$ and then proceeds to smaller $r$ as a stable MOTS. Correspondingly, $r=2m$ starts out stable and transitions to become $(0,1)$-unstable. This is a transcritical bifurcation. 

While the MOTSs look spherical in the diagram, these are just coordinate representations. In reality they are very flattened and are not isometrically embeddable in Euclidean $\mathbb{R}^3$ (Figure \ref{fig:WeylDef}).

Figure \ref{fig:Weyl4} also shows the transcritical bifurcation at $\alpha \approx -2.798084$. In this case $p_o = 40$ keeps the constant $T$ slice spacelike over a wide range of $r$: at least $ m< r <3m $. However, over that range the metric quantities can vary by over 40 orders of magnitude! This causes numerical issues and so in practice we can only identify MOTSs much closer to $r=2m$. For $ 1.9m< r <2.1m $ there are about three orders of magnitude. 

That said, as $\alpha$ decreases a $(0,1)$-unstable MOTS again approaches from the outside, crosses $r=2m$ and becomes stable. The $r=2m$ MOTS stability evolves in the opposite way. The MOTS diagrams are very similar to Figure \ref{fig:WeylTrans} but omitted as they stay so close to $r=2m$ that they cannot be seen well.

\begin{figure}
    \centering
    \includegraphics{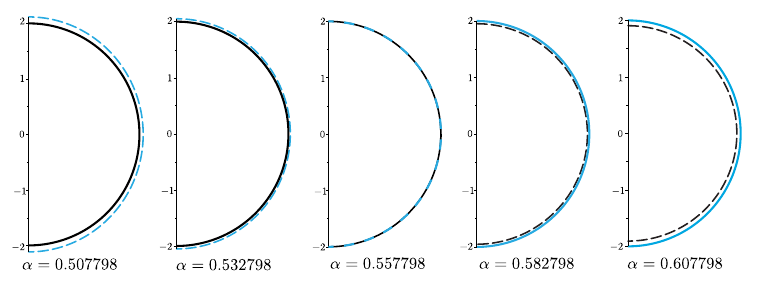}
    \caption{The transcritical bifurcation at $\alpha \approx 0.557798$. As $\alpha$ increases, a
    $(0,1)$-unstable MOTS approaches from outside, swaps stability with $r=2m$ and then continues to smaller $r$ as a stable MOTS. In this figure $r=2m$ is the solid line while the dashed line is the other MOTS. Color indicates stability in as in earlier figures. }
    \label{fig:WeylTrans}
\end{figure}

\subsubsection{The $(1,1)$-unstable bifurcations}
\label{sec:weylpitch}

To study the bifurcation at $\alpha \approx 0.956652$ we choose $p_o=25m$, which keeps the surfaces of constant $T$ spacelike in the range $1.85m < r < 2.6m$ for that value of $\alpha$. 

The nearby MOTSs are shown in Figure \ref{fig:WeylPitch} and maximum $r$ intersections with the $z$-axis are plotted in Figure \ref{fig:Weyl4}.
For $\alpha < 0.956651$, $r=2m$ is the only MOTS in its neighbourhood. It is $(0,1)$-unstable. At $\alpha \approx 0.956652$ it becomes 
$(1,1)$-unstable and a bifurcation occurs. For $\alpha > 0.956652$ it is $(0,2)$-unstable and two $(0,1)$-unstable MOTSs diverge from it. This is a pitchfork bifurcation. 

As in the previous case, while these MOTSs look nearly spherical, geometrically they are very flattened and are not isometrically embeddable in Euclidean $\mathbb{R}^3$ (Figure \ref{fig:WeylDef}).

Figure \ref{fig:Weyl4} also shows a pitchfork bifurcation at $\alpha \approx -3.741562$. In this case $p_o =40$ keeps the surfaces of constant $T$ spacelike in the range $m < r < 3m$, but again this is over many orders of magnitude for the metric components. Restricting to $1.9 m < r < 2.1 m$ keeps that to a range of less than four orders of magnitude. 

This time, as $\alpha$ decreases the $(0,1)$-unstable MOTS at $r=2m$ is alone until $\alpha \approx -3.741562$, at which point another pitchfork bifurcation occurs. The MOTS diagrams are very similar to Figure \ref{fig:WeylPitch} but again are omitted as they stay so close to $r=2m$ that they cannot be seen well.

\begin{figure}
    \centering
    \includegraphics{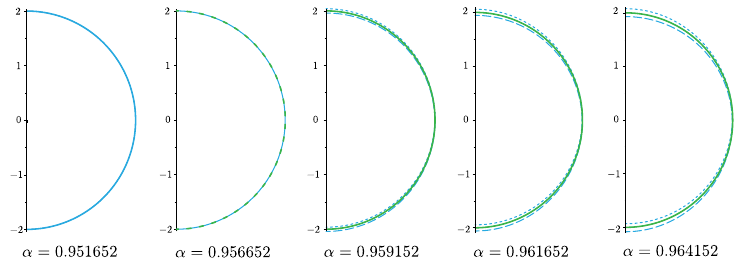}
    \caption{The pitchfork bifurcation at $\alpha \approx 0.956652$. $r=2m$ is initially $(0,1)$-unstable, becomes $(1,1)$-unstable at the bifurcation point and afterwards is $(0,2)$-unstable. The MOTSs that bifurcate off are $(0,1)$-unstable. In this figure $r=2m$ is the solid line while the other MOTSs are dashed  lines. Color indicates stability as in earlier figures. }
    \label{fig:WeylPitch}
\end{figure}

\section{Discussion}

In this paper we examined how a MOTS changes under general geometric deformations of the time slice that contains it. Possible deformations include time evolutions, arbitrary changes of parameters in exact solutions, and changes to the time slice within the full spacetime. In all cases,
we have seen that the bifurcation theory for solutions of differential equations also applies to MOTSs. This brings further order to the complicated behaviour of ``exotic'' MOTSs that have been catalogued in both time and phase space evolutions over the last few years. 

%Restricting attention to $(3+1)$-gravity, this theory also guarantees that the existence of such MOTS is not restricted to particular foliations. 
As long as a MOTS has an invertible stability operator, it will continue to exist after a sufficiently small deformation. If there are vanishing eigenvalues there may be bifurcations, but even in this case the possibilities are governed by the theory. Thus, though we have worked exclusively in axisymmetry both in this paper and in the head-on merger simulations of \cite{Pook-Kolb:2021gsh,Pook-Kolb:2021jpd}, the same phenomena should be observable in more general foliations and spacetimes. 

In particular, starting from our existing examples, any MOTS with an invertible stability operator will continue to exist in nearby spacetimes, whether or not they are axisymmetric.  ``Exotic'' MOTS are not artifacts of axisymmetry. On the other hand, for an axisymmetric MOTS that undergoes a bifurcation, a perturbation away from axisymmetry could change the type of bifurcation that occurs. This scenario is described by the theory of imperfect bifurcations \cite{GSimperfect,strogatz2015}; some simple examples with spacetime perturbations were considered in Section \ref{sec:imperfect}. 
There may be extreme situations where there are no observable MOTS (in analogy with the lack of trapped surfaces in \cite{iyerwald}) but they will be the exception rather than the rule. 
Similar considerations apply when we deform the spacetime foliation containing our MOTS: this can change the stability and hence the resulting bifurcation diagram, as discussed in Section~\ref{sec:slicedependence}. 
%and~\ref{sec:imperfect}.

In the example sections we have seen saddle-node, transcritical and pitchfork bifurcations associated with the vanishing of both principal and non-principal eigenvalues. They occur in full spherical symmetry (Section 4), for bifurcations departing from spherically symmetry (Section 5) and for situations where neither the background nor any of the MOTSs have that symmetry (Section 6). This demonstrates that the bifurcations can occur in a wide variety of situations while also providing simple models for testing future ideas. In particular, we demonstrate imperfect bifurcations in the simplest models (Section 4).

While the RNdS and RN examples are good toy models, the Weyl--Schwarzschild example is perhaps more physically relevant. Ultimately we are interested in a full understanding of how the geometry evolves during black hole mergers and Weyl--Schwarzschild, though still static, allows us to examine how MOTSs evolve under the application of external deforming fields. In particular, we saw that sufficiently strong deformations can cause a MOTS to lose stability. Without stability, a MOTS also loses its barrier property and so this is an alternate way in which an apparent horizon could dissolve into an external gravitational field (the other way is just to disappear in a pair annihilation, i.e., a saddle-node bifurcation). During a merger, the loss of the original apparent horizons is a (final?) step after
which the original black hole identities are lost inside the new, merged black hole. 

Even for the exact solutions considered here, one generally has to use numerical techniques to both find MOTSs and classify their stability (discussed in Appendix C). However, understanding that the possible bifurcations are constrained significantly simplifies numerical searches for MOTSs and we expect that that will continue to be useful in searches during full numeric merger calculations. 

Over the decades a variety of stability operators have appeared in the literature. Appendix A presents them in a unified modern language. It shows how stability depends on the direction of deformation while also examining the various changes which \emph{do not} modify the stability classification (which depends on the number of negative eigenvalues but not their particular values).

Bifurcation theory is certainly not specific to Einstein's equations and so the same types of bifurcations should appear in any theory of gravity with MOTSs or indeed in any theory that has generalizations of MOTSs. As long as they are surfaces which are solutions to a set of differential equations, the possible evolutions will be governed by bifurcation theory. 
%
%Above we noted that bifurcation theory is very general and independent of Einstein's equations: it equally well applies to other theories. 
However, it is likely that the particular bifurcations that can occur and how they occur will be restricted during time evolutions by those equations. With respect to time evolution, the Einstein equations (together with a suitable energy condition) may place restrictions on the direction from which an unstable MOTS can approach a stable one during a transcritical bifurcation, or whether a pitchfork is supercritical or subcritical. 
We leave discussion of this  for later work;  some results in this direction can  be found in \cite{BusseyCoxKunduri}.

While this new understanding further increases the utility of the MOTS stability operator in understanding black holes and their internal structure, parts of the theory still require work. In this paper we have essentially identified bifurcations ``by hand'': when an eigenvalue vanishes we have manually searched for nearby MOTS and so identified the type of bifurcation. However, they can also be classified in a more rigorous way by calculating higher derivatives (deformations) of the stability operator. While straightforward in principle, the calculations tend to be rather complicated.
%We have made partial progress on this but chosen to omit it from this paper. 
This approach is discussed systematically in \cite{BusseyCoxKunduri}.

The theory presented here does not cover all phenomena that have been observed in simulations. For example:
\begin{enumerate}
    \item 
It is not uncommon for MOTSs in simulations to develop cusps. During a merger, when the two original apparent horizons first touch, a third (unstable) MOTS ``hugs'' tightly around them with cusps at the point of contact (see, for example, Figure 3 in \cite{Pook-Kolb:2019iao}). Subsequently, the two original MOTSs pass through each other and the third develops self-intersections. During this process, the third MOTS is  effectively a combined deformation of the two other ones.
Some interesting results on how physical properties of MOTS evolve through the formation of cusps can be found in \cite{kastha2026cuspformationmergingblack} but connections to bifurcation theory have not yet been explored.

\item Toroidal MOTS have been observed to form during departures from time symmetry, where a minimal surface splits into two oppositely oriented MOTSs plus a toroidal MOTS between them \cite{Pook-Kolb:2021jpd,Sievers:2023zng}. The toroidal MOTS cannot be viewed as a smooth deformation of the original minimal surface, as it has different topology.
\end{enumerate}
We expect that the theory can be extended to cover both of these cases but that is also left for future work.

\section*{Acknowledgements}
IB and CMO were supported by NSERC grant 2018-0473. GC acknowledges the support of NSERC grants RGPIN-2017-04259 and RGPIN-2025-06186. 
%IB would like to thank Robie Hennigar, Badri Krishnan and Ryan Unger for useful discussions. 
The authors would like to thank Liam Bussey, Robie Hennigar, Badri Krishnan, Hari Kunduri,  Ryan Unger and Yucun Xie for useful discussions. 

\vspace{1cm}

\appendix
%\addtocontents{toc}{\fixappendix}

\section{The generalized stability operator}
\label{sec:stabop}

The stability operator defined in Section \ref{sec:slice} measures changes in the expansion $\theta_\ell$ when a MOTS $\cS$ is deformed within a spacelike slice. That is, the deformation is proportional to $N$, the normal to $\cS$ in the slice. Here we develop a more general stability operator, allowing for deformations proportional to any vector field $X$ that is normal to $\cS$ in the full spacetime.

\subsection{Definition and basic properties}
\label{sec:defprop}

Let $X$ be a vector field in $M$ that is normal to $\cS$, and consider a (local) congruence of curves $\Gamma$
that are tangent to $X$ at the point where they intersect $\mathcal{S}$. Then, given a function $\psi$, choose a parameter $\zeta$ on those curves so that
\begin{enumerate}
    \item $\zeta =0$ on $\mathcal{S}$ and 
    \item $\frac{\partial}{\partial \zeta} = \psi X$ on $\mathcal{S}$.
\end{enumerate}
If we Lie-drag $\mathcal{S}$ along the $\Gamma$ to define a family of surfaces $\mathcal{S}_\zeta$ and calculate the geometric 
properties of those surfaces, the deformation of $\theta_\ell$ by an amount $\psi$
in the direction $X$ is
\begin{equation}
    \delta_{\psi X} \!  \theta_\ell  = \frac{\mathrm{d} \theta_\ell}{\mathrm{d} \zeta} \; . 
\end{equation}
This derivative turns out to be independent of the details of $\Gamma$ away from $\mathcal{S}$. Further, in practice it is not necessary to even construct the congruence: we can calculate this derivative of $\theta_\ell$ entirely from surface quantities. Over the decades the expression for this quantity has been derived many times
(for example \cite{Newman_1987,Hayward:1994,  Andersson:2005, cao2011}). Here we will use the form found in \cite{Booth:2007}. For a general future-outward null vector $\ell$ we have\footnote{This is actually a slight generalization of the \cite{Booth:2007} calculation. In that paper 
the authors fixed $l^+ = \ell$. Here we have generalized to 
allow for the expression to be calculated with an arbitrary
reference dyad. 
%Setting $ \ell = l^+$, it is straightforward to recover the original form.
}
\begin{equation}
    \delta_{\psi X} \theta_\ell = e^\mu \bigg( 
    -  \mathcal{D}^2_{l^+} \left( {B} \psi \right) 
     + \Big( {B}  \left[ \mathcal{K} - %G_{_{+ -}} 
     G_{l^+ l^-} %G(l^+,l^+) %G_{ab} l_+^a l_+^b 
     \right]  - {A} \left[ \|\sigma_{+}\|^2 +G_{l^+ l^+} \right]\Big) \,   \psi \bigg) \label{eq:var} 
\end{equation}
where, for a cross-scaled null normal reference dyad $(l^+, l^-)$,
\begin{equation}
    e^{\mu} = - \ell \cdot l^- 
    \; ,  \quad {A} = - X \cdot l^-
    \quad \mbox{and} \quad {B} = X \cdot l^+ \, ,
\end{equation}
so that $X = {A} l^+ - {B} l^-$; see Figure \ref{AB}. As in Section \ref{sec:slice}, $\mathcal{K}$ is the Gauss curvature of $\mathcal{S}$, 
$G_{l^+ l^\pm} = G_{ab} l_+^a l_\pm^b$ for the Einstein tensor $G_{ab}$, $\|\sigma_{l^+}\|^2 = \sigma_{AB}^{l^+} \sigma^{AB}_{l^+}$
for the outward null shear
\begin{equation} 
\sigma^{l_+}_{AB} = k^{l_+}_{AB} - \frac{1}{2}  \theta_{l^+} q_{AB} = 
e_A^a e_B^b \nabla_A l^+_B - \frac{1}{2} \theta_{l^+} q_{AB}
\end{equation}
and
\begin{align}
\mathcal{D}_{l^+}^2 F &  = (\mathcal{D}^A - \omega_{l^+}^A)(\mathcal{D}_A - {\omega}_A^{l^+}) F \, , 
% \\
% & = \mathcal{D}^2 F - 2 {\omega}_{l^+}^A \mathcal{D}_A F + (- \mathcal{D}_A {\omega}_{l^+}^A + {\omega}^{l^+}_A {\omega}_{l^+}^A) F \, ,
% \nonumber
\end{align} 
where $\mathcal{D}$ is the covariant derivative on $\mathcal{S}$ and ${\omega}^{l_+}_A = - e_A^a l^-_b \nabla_a l_+^b$
is the connection on the normal bundle to $\mathcal{S}$, expressed with respect to the reference dyad $(l^+, l^-)$.

Under a rescaling 
of the dyad $(l^+, l^-)$ a straightforward calculation shows that 
\begin{equation}
    {\omega}^{e^\beta l_+}_A = \omega^{l_+}_A + \mathcal{D}_A \beta \; . 
    \label{eq:tomgauge}
\end{equation}
% In the preceding and following 
% equations we
% move the $l^+$ (or $l_+$)  label up and down as is notationally convenient. 
Using this, it can be shown that the right-hand side of (\ref{eq:var}) does not depend on the choice of dyad; cf. the proof of Lemma 1 in  \cite{Jaramillo_2015}. We can thus define a stability operator without reference to the normal dyad.

\begin{defn}
\label{stabop}
Let $\mathcal{S}$ be a MOTS with 
an outward null vector field $\ell$ and
$X$ a vector field normal 
to $\mathcal{S}$ with $X \cdot \ell > 0$.
The \textbf{stability operator in 
the direction $X$} acts on a function $\psi$ as
\begin{equation}
    L_\mathcal{S}(\ell, X) \psi := \delta_{\psi X} \theta_\ell \; .  
\end{equation}
\end{defn}

\begin{figure}
\begin{center}
    \includegraphics[]{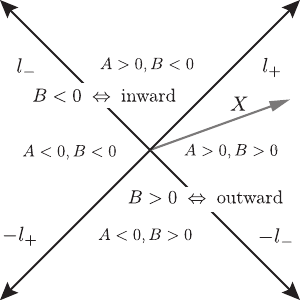}
    \caption{For $A>0$ and $B>0$, $X$ is spacelike and outward oriented. }
    \label{AB}
\end{center}
\end{figure}

We have required $X \cdot \ell = e^\mu B > 0$ to restrict our attention to outward-pointing $X$, as shown in 
Figure \ref{AB}, thus $\psi > 0$ will generate an outward deformation and $\psi < 0$ will generate an inward one. Under this assumption, each $L_\mathcal{S}  (\ell,  X)$ is a linear %second-order 
elliptic differential operator. To see this we need to find the leading (i.e., second-order) part of the stability operator. Expanding the derivative term in (\ref{eq:var}), we find
% \begin{equation}
%     L_\mathcal{S}(\ell, X) = -e^\mu B \mathcal{D}^2 \psi
%     \ +\  \textrm{terms with $\mathcal D \psi$ and $\psi$} \; .
% \end{equation}
\begin{align}
    L_{\mathcal{S}} (\ell, X) \psi  &= 
    -e^\mu B \mathcal{D}^2 \psi - 2e^\mu (\mathcal{D}^C_{l^+} B) 
      \mathcal{D}_C \psi \label{eq:def2}\\
    & \qquad + e^\mu \left(-  \mathcal{D}^2_{l^+} B + B [\mathcal{K}    -G_{l_+ l-}]  - 
   A [ \| \sigma_{l_+} \|^2 + G_{l_+ l_+} ]\right)  \psi \; . \nonumber
\end{align}
The leading part $-e^\mu B \mathcal{D}^2 \psi$ is elliptic because $B>0$, hence $L_\mathcal{S}(\ell, X)$ is elliptic.

As a linear elliptic operator defined on a compact manifold, $L_{\mathcal{S}} (\ell, X)$ has a discrete spectrum. In general it is not self-adjoint, and hence can have complex eigenvalues; see \cite{BCK-Kerr} for an explicit example. Even when $L_\mathcal{S} (\ell, X)$ is not self-adjoint, however, the \emph{principal eigenvalue} $\lambda_o$ (the eigenvalue with the smallest real part) is real and simple, and the corresponding eigenfunction $\psi_0$ has no zeros, so it can be chosen to be positive everywhere \cite{Andersson:2005}. We then say that $\cS$ is \emph{strictly stable (in the direction $X$)} if $\lambda_o > 0$, \emph{marginally stable} if $\lambda_o = 0$, and \emph{unstable} if $\lambda_o< 0$.

In general this classification depends on the direction $X$: examples of how the same MOTS can be strictly stable, marginally stable or unstable depending on the choice of $X$ are shown in Figure \ref{fig:zoom}. %In the next section we show that this does not depend on the scaling of $\ell$ or $X$.

\begin{figure}
    \centering
    \includegraphics[width=0.7\linewidth]{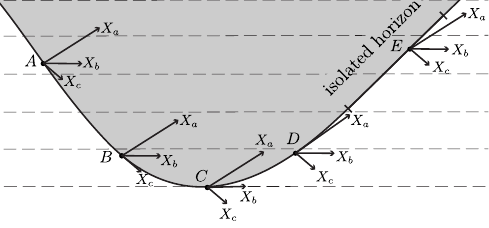}
    \caption{Magnification of part of Figure \ref{fig:mtt}, demonstrating the $X$ dependence of the stability operator $L_{\mathcal{S}} (\ell, X)$. 
    %In this spherically symmetric case it is sufficient to consider similarly symmetric deformations. 
    All $X$ shown in the figure are spacelike and outward-oriented. The individual MOTSs are unstable with respect to deformations into the gray  outer trapped region, marginally stable when tangent to the MOTT and stable with respect to deformations into the white  outer untrapped region.  }
    \label{fig:zoom}
\end{figure}

\subsection{Scaling invariance and self-adjointness}
\label{sec:rescaling}

We now show that the classification of $\cS$ as strictly stable, marginally stable or unstable does not depend on the scaling of $\ell$ and $X$.

\begin{prop}
    \label{prop:scaling-stab}
     Let $\cS$ be a MOTS with $\ell$ and $X$ as above. For any functions $\mu$ and $\gamma$, the principal eigenvalues of $L_{\mathcal{S}} (e^\mu \ell, e^\gamma X)$ and $L_{\mathcal{S}} (\ell, X)$ have the same sign.
\end{prop}

\begin{proof}
Let $\lambda_o$ and $\tilde\lambda_o$ denote the principal eigenvalues of $L_{\mathcal{S}} (\ell, X)$ and $L_{\mathcal{S}} (e^\mu \ell, e^\gamma X)$, respectively, and let $\psi_o > 0$ be the eigenfunction for $\lambda_o$. Using (\ref{eq:var}) we calculate
\begin{equation}
    L_{\mathcal{S}} (e^\mu \ell, e^\gamma X)(e^{-\gamma}\psi_o) = e^\mu L_{\mathcal{S}} (\ell, X)\psi_o = \lambda_o e^\mu \psi_o.
\end{equation}
Now consider the function $\psi = e^{-\gamma}\psi_o > 0$. If $\lambda_o > 0$, then $L_{\mathcal{S}} (e^\mu \ell, e^\gamma X) \psi > 0$ and it follows from Lemma 2 in \cite{Andersson:2005} that $\tilde\lambda_o > 0$. If $\lambda_o = 0$, then $L_{\mathcal{S}} (e^\mu \ell, e^\gamma X) \psi = 0$ and we conclude that $\tilde\lambda_o = 0$.
Reversing the roles of the two operators, the same argument shows that $\lambda_o$ is positive (zero) if $\tilde\lambda_o$ is positive (zero). We then conclude that $\lambda_o < 0$ if and only if $\tilde\lambda_o < 0$, completing the proof.
\end{proof}

When $L_{\mathcal{S}} (\ell, X)$ is self-adjoint, and hence has real eigenvalues, we get more. In particular, the number of negative and zero eigenvalues is invariant under rescalings of $\ell$ and $X$. To this end, first recall that 
% two operators $L$ and $\tilde{L}$ are \emph{similar} if
% % \begin{prop}
% % \label{prop:simcong}
% % Let $L$ and $\tilde{L}$ be linear operators. Then they are said to be
%     % \begin{enumerate}
%     %     \item \textbf{similar} and we write $L \sim \tilde{L}$, if for some function $\gamma$:
% \begin{equation}
%     {L} = e^\gamma \tilde{L} e^{-\gamma} \; \;  \mbox{so that for any $\psi$} \, ,  \;   {L} \psi = 
%     e^\gamma \tilde{L} ( e^{-\gamma} \psi) \; . 
% \end{equation}
% Similar operators clearly share a common eigenvalue spectrum (with eigenfunctions rescaled by $e^\gamma$). Then we have the following. 
%
if ${L} = e^\gamma \tilde{L} e^{-\gamma}$ for some function $\gamma$, then $L$ and $\tilde L$ are similar\footnote{This is a special case of similarity, where the transformation is multiplication by the smooth, positive function $e^\gamma$.} and hence have the same eigenvalues.

%         \item \textbf{congruent} and we write $L \simeq \tilde{L}$, if for some 
% function $\gamma$:
% \begin{equation}
%     {L} = e^\gamma \tilde{L} e^{\gamma} \; \;  \mbox{so that for any $\psi$} \, ,  \;   {L} \psi = 
%     e^\gamma \tilde{L} ( e^{\gamma} \psi) \; . 
% \end{equation}
% Congruent self-adjoint operators have the same number of negative eigenvalues and the same number of vanishing eigenvalues. 
%     \end{enumerate}
% \end{prop}
% \begin{proof} 
% This is straightforward: (i) is easy and  (ii) just quotes a standard result: 
%     \begin{enumerate}
%     \item It is immediate that if $\psi_n$ is an eigenfunction of $L$ then 
% $e^\gamma \psi_n$ is an eigenfunction of $\tilde{L}$ and they share 
% a common eigenvalue $\lambda_n$. Hence the operators have the same eigenvalue 
% spectrum. 
%   \item  It is immediate that if $\psi_0$ is an eigenfunction of $L$ with eigenvalue
% $0$ then $e^{-\gamma} \psi_0$ is an eigenfunction of $\tilde{L}$ with 
% eigenvalue $0$. Hence they have the same number of vanishing eigenvalues. More generally the eigenvalues and functions don't have such a simple relationship, but a 
% generalization of Sylvester's Law of Inertia (Lemma 2.3 in \cite{inertia}), tells us that congruent self-adjoint operators have the same number of negative eigenvalues in their spectra. 
%     \end{enumerate}
% \end{proof}

% We have seen that the eigenvalue spectrum depends on the $X$. However the next result shows that only the direction of $X$ matters. The scaling doesn't affect the spectrum. 
%Then it is straightforward to show the following:
\begin{prop}
    \label{prop:scaling}
     Let $L_{\mathcal{S}} (\ell, X)$ be a self-adjoint MOTS stability operator with $n_0$ vanishing  and $n_-$ negative eigenvalues. For any functions $\mu$ and $\gamma$, the operator $L_{\mathcal{S}} (e^\mu \ell, e^\gamma X)$ has  $n_0$ vanishing and $n_-$ negative eigenvalues.
\end{prop}

In general the rescaled operator $L_{\mathcal{S}} (e^\mu \ell, e^\gamma X)$ will not be self-adjoint, but it is similar to a self-adjoint operator, namely (\ref{congop}), therefore it has real eigenvalues and the numbers $n_0$ and $n_-$ are well defined.

\begin{proof}
As in the proof of Proposition \ref{prop:scaling-stab}, we use (\ref{eq:var}) to calculate
\begin{equation}
    e^{(\gamma-\mu)/2} L_{\mathcal{S}} (e^\mu \ell, e^\gamma X) \big(e^{(\mu - \gamma)/2} \psi\big) = e^{(\mu+\gamma)/2} L_{\mathcal{S}} (\ell, X) \big(e^{(\mu+\gamma)/2}\psi \big).
\end{equation}
This shows that $L_{\mathcal{S}} (e^\mu \ell, e^\gamma X)$ is similar
%\footnote{Recall that operators $A$ and $B$ are \emph{similar} if $B = S^{-1}AS$ for some invertible $S$, and \emph{congruent} if $B = S^TAS$.}
to the operator
\begin{equation}
\label{congop}
    e^{(\mu+\gamma)/2} L_{\mathcal{S}} (\ell, X) \big(e^{(\mu+\gamma)/2}\psi \big)
\end{equation}
 on the right-hand side, thus they have the same eigenvalues. Moreover, the operator in (\ref{congop}) is congruent to $L_{\mathcal{S}} (\ell, X)$, and so by Sylvester's Law of Inertia (Lemma 2.3 in \cite{inertia}) it has the same number of zero and negative eigenvalues. 
\end{proof}

We next give a sufficient condition for the stability operator to be self-adjoint.

\begin{prop}
\label{prop:rotation}
If $\cS$ is a non-rotating MOTS, then $L_\cS(\ell,X)$ is similar to a self-adjoint operator. Moreover, there is an $\ell_*$ (depending on $X$) for which $L_\cS(\ell_*,X)$ is self-adjoint.
\end{prop}

\begin{proof}
If $\omega^{l_+}_A = D_A F$, then by 
(\ref{eq:tomgauge}) we can pick a new gauge $\beta = -F$ for which 
$\omega_A^{e^\beta {l}_+} = 0$. We thus assume the dyad $(l^+,l^-)$ has been chosen to have $\omega^{l_+}_A = 0$. As noted in Section \ref{sec:defprop}, the stability operator is independent of this choice. From (\ref{eq:var}) we have
\begin{equation}
    L_\cS(\ell,X)\psi = e^\mu \bigg( 
    -  \mathcal{D}^2 \left( {B} \psi \right) 
     + \Big( {B}  \left[ \mathcal{K} - %G_{_{+ -}} 
     G_{l^+ l^-} %G(l^+,l^+) %G_{ab} l_+^a l_+^b 
     \right]  - {A} \left[ \|\sigma_{+}\|^2 +G_{l^+ l^+} \right]\Big) \,   \psi \bigg) \;.  
\end{equation}
The zeroth-order part is automatically self-adjoint, so we only need to check the $\mathcal{D}^2$ term. Integrating by parts, for an arbitrary test function $\phi$ we have
\begin{equation}
    \left< \phi, e^\mu\mathcal{D}^2 \left( B \psi \right) \right>
    = \left< e^\mu \phi, \mathcal{D}^2 \left( B \psi \right) \right>
    = \left< \mathcal{D}^2 \left(e^\mu \phi\right) , B \psi \right>
    = \left< B\mathcal{D}^2 \left(e^\mu \phi\right) ,  \psi \right> \,,
\end{equation}
where $\left<\cdot,\cdot\right>$ denotes the standard $L^2$ inner product with respect to the volume form on $\cS$. This will be self-adjoint if $e^\mu = B$, so we choose $\ell_* = B l^+$.

Any other $\ell$ is of the form $\ell = e^\gamma \ell_*$. We then calculate, as in the proof of Proposition \ref{prop:scaling},
\begin{equation}
    e^{-\gamma/2} L_\cS(\ell,X)\big(e^{\gamma/2}\psi\big) 
    = e^{-\gamma/2} L_\cS(e^\gamma \ell_*,X)\big(e^{\gamma/2}\psi\big)
   = e^{\gamma/2} L_\cS(\ell_*,X)\big(e^{\gamma/2}\psi\big) \;.
\end{equation}
It follows that $L_\cS(\ell,X)$ is similar to the operator $\psi \mapsto e^{\gamma/2} L_\cS(\ell_*,X)\big(e^{\gamma/2}\psi\big)$. This latter operator is congruent to the self-adjoint operator $L_\cS(\ell_*,X)$ and hence is self-adjoint.
\end{proof}

Hence, given a non-rotating MOTS $S$ and a normal direction $X$, we may invariantly classify its stability by the number of zero and negative eigenvalues, without reference to the scaling of the reference dyad $(l^+, l^-)$, $\ell$ or $X$.

\begin{defn}
    \emph{(Stability Classification)} A non-rotating MOTS $\mathcal{S}$ is \textbf{$\bm{ (X,n_0, n_-)}$-unstable} if  $L_{\mathcal{S}} (\ell, X)$ has $n_0$ vanishing eigenvalues and $n_-$ negative eigenvalues. 
\end{defn}

By this classification, a  MOTS that is strictly stable in the $X$ direction is $(X,0,0)$-unstable, while a marginally stable MOTS is $(X,1,0)$-unstable. 

\subsection{Stability in a spacelike hypersurface}

The most common appearance of the stability operator (and the one used in most of this paper) is that for an $\mathcal{S}$ that is embedded in a spacelike hypersurface $\Sigma$. This is the version that first appeared in \cite{Andersson:2005}.

Then it is natural to start with  the timelike normal $u$ to $\Sigma$ and the spacelike normal $N$ to $\mathcal{S}$ in $\Sigma$. 
From there one defines null normals
\begin{equation}
    l_+ = u + N, \qquad
    l_- = \frac{1}{2} \left( u - N\right)
\end{equation}
or some constant multiples of these quantities (depending on scaling choices). It is then also natural to work with the stability operator which acts as
\begin{equation}
    L_\mathcal{S} (l_+, N) \psi = \delta_{\psi N} \theta_{l^+} \; . 
\end{equation}
This is the one that we work with through most of this paper. 

However, consider now another hypersurface $\bar{\Sigma}$ that  also intersects $\mathcal{S}$. The associated timelike $\bar{u}$ and spacelike $\bar{N}$ normals can be written as
\begin{align}
    \bar{u} &  = ( \cosh \gamma ) u + (\sinh \gamma) N \label{eq:tilt} \\
    \bar{N} &  = ( \sinh \gamma ) u + (\cosh \gamma) N \nonumber
\end{align}
for some function $\gamma$, whence
\begin{align}
    \bar{l}_+ &= e^\gamma l_+ \\
    \bar{l}_- &= e^{-\gamma} l_- 
    \nonumber \; . 
\end{align}
Then it is straightforward to see that relative to this new slice
\begin{equation}
    L_\mathcal{S} (\bar{l}_+, \bar{N} ) = L_\mathcal{S} (e^\gamma l_+, e^{-\gamma} N) + L_\mathcal{S} \big(e^\gamma l_+, (\cosh \gamma)l_+ \big) \; . \label{eq:SigmaSim}
\end{equation}
If the second term on the right-hand side vanishes (as it will in the next subsection) then this is a similarity relation  and so they have the same eigenvalue spectrum.

\subsection{Non-expanding horizons}
\label{sec:nonex_stab}

A non-expanding horizon (NEH) \cite{Ashtekar:2000a,Ashtekar:2000b,Ashtekar:2001jb,Ashtekar:2004a} is a null three-surface of topology 
$\mathbb{R} \times S^2$  on which 
$\theta_\ell = 0$ and the null dominant energy condition holds ($-T_{ab} \ell^b$ is 
causal and future directed). 
The best known examples of non-expanding horizons are the Killing horizons (inner, outer and cosmological) in the standard exact solutions such as the Kerr--Newman--(A)dS family. However, NEHs are more general than Killing horizons as they do not require global ``time'' Killing vector fields.\footnote{In generality, isolated horizons are intermediate between non-expanding and Killing horizons. Roughly, a non-expanding horizon has fixed intrinsic 
geometry while an isolated horizon also has a fixed extrinsic geometry. All Killing horizons are isolated horizons and all isolated horizons are non-expanding horizons.}

In this case the stability of $\cS$ does not depend on the direction $X$ of deformation.

\begin{prop}
    Let $\mathcal{S}$ be a non-rotating MOTS that is a slice of a non-expanding horizon. If $\mathcal{S}$ is $(X,n_0, n_-)$-unstable in an outward normal direction $X$, then it is $(\tilde{X},n_0, n_-)$-unstable in any other outward normal directions $\tilde{X}$. 
\end{prop}

\begin{proof}
From the Raychaudhuri equation it follows that for a non-expanding horizon
    \begin{equation}
    \| \sigma_{l^+} \|^2 + G_{l^+ l^+}  = 0 \; . \label{eq:raycon}
    \end{equation}
From (\ref{eq:var}) we see that $L_\cS(\ell,X)$ does not depend on the value of $A$ (where $X = Al^+ - Bl^-$), thus
\begin{equation}
    L_{\mathcal{S}}(\ell, X) = L_{\mathcal{S}}(\ell, -B l_-) \; .   \label{eq:NEHstab1}
\end{equation}
It then follows from Proposition \ref{prop:scaling} that $L_{\mathcal{S}}(\ell, -B l_-)$ has the same number of zero and negative eigenvalues as
\begin{equation}
    L_{\mathcal{S}}(l_+, -l_-) = -  \mathcal{D}_{l^+}^2  +   \left[ \mathcal{K} - G_{l^+ l^-} \right] \; .   \label{eq:NEHstab2}
\end{equation}
This means all $L_{\mathcal{S}}(\ell, X)$ have the same number of zero and negative eigenvalues. 
\end{proof}
\noindent The comment following (\ref{eq:SigmaSim}) is special case of this result. From (\ref{eq:var}) we have
\begin{equation}
    L_\mathcal{S} \big(e^\gamma l_+, (\cosh \gamma)l_+ \big) 
    = - e^\gamma \cosh \gamma (
     \| \sigma_{l^+} \|^2 + G_{l^+ l^+}) = 0 \label{eq:sigSim2}
\end{equation}
for a non-expanding horizon. 

Thus, for a non-expanding horizon, we can talk about $(n_0,n_-)$-instability without specifying a direction of deformation. As noted in Figure \ref{fig:zoom}, this is not true for a general MOTS. 

For a non-rotating, vacuum Killing horizon it is immediately clear that the (\ref{eq:NEHstab2}) is also independent of the particular MOTS chosen on which to calculate the stability (the induced two-metric is slice independent). However, this is also true more generally. 
In Proposition 4 of \cite{Mars:2012sb}, it was shown that the stability operators on different (MOTS) slices of \emph{non-evolving} horizons (a variant of an isolated horizon) are similar and so all have same stability classification. This follows by essentially the same argument as (\ref{eq:SigmaSim}) and (\ref{eq:sigSim2}): for non-evolving horizons the extrinsic geometry is also invariant. Then all slices have the same geometry up to ``tiltings'' like those in (\ref{eq:tilt}).

Combined, these results
greatly simplify the study of stability for non-expanding horizons: all component MOTS have the same stability classification which is also independent of the deformation direction $X$. 
Often it will be most convenient to calculate  via the right-hand side of (\ref{eq:NEHstab2}). All NEHs considered in this paper are non-rotating and so we can choose a gauge for which $\omega_A^{l^+} = 0$. Then we can understand the (in)stability of $L_{\mathcal{S}}(\ell, X)$ 
by calculating the spectrum of
\begin{equation}
    L_\mathcal{S} \big( l_+,  - l_-  \big)  =    
    -  \mathcal{D}^2   +   \left[ \mathcal{K} - G_{l_+ l_-} \right]   \; . \label{eq:simplerL}
\end{equation}

\section{Finite-dimensional bifurcation theory}
\label{app:bif}

We now review bifurcation theory for finite-dimensional dynamical systems, giving examples of the three types of bifurcations (saddle-node, transcritical and pitchfork) that appear in this work. For more details see, for example, \cite{strogatz2015} as an introduction or \cite{kuznetsov1998elements,wiggins2006} for a more advanced discussion.

\subsection{The one-dimensional case}

Consider a one-dimensional dynamical system
\begin{equation}
    \dot{y} = F ( y) 
\end{equation}
where dot is the time derivative. The \emph{fixed points} are values $y_o$ for which $F(y_o) = 0$. Writing $y = y_o + \Delta y$ near a fixed point, the perturbation $\Delta y$ evolves according to 
\begin{equation}
    \Delta\dot{y} \approx F_y(y_o) \Delta y \label{eq:1Ddynsys} \, , 
\end{equation}
where the subscript is a $y$-derivative. If $F_y(y_o) < 0$ then $y_o$ is \emph{stable}: nearby $y$ evolve towards it as time increases. If $F_y(y_o) > 0$ then $y_o$ is
\emph{unstable}: nearby $y$ evolve away from it. If $F_y(y_o) = 0$ then $y_o$ is \emph{marginally stable} and it is necessary to consider higher derivatives of $F$ to determine what happens to nearby $y$.

Now consider a parameterized one-dimensional 
dynamical system
\begin{equation}
    \dot{y} = F(y, \alpha) \label{eq:1Ddynsysparam} \; ,
\end{equation}
where $\alpha$ is a real parameter. If $F(y_0, \alpha_o) = 0$, then $y_o$ is a fixed point for the parameter value $\alpha_o$. Bifurcation theory describes how fixed points can appear or disappear, as well as change stability, as $\alpha$ varies.

This amounts to describing the set of solutions to $F(y,\alpha) = 0$. If $F_y(y_o, \alpha_0) \neq 0$, then by the implicit function theorem we can solve for $y$ to get a smooth curve of fixed points, $y=Y(\alpha)$, with $Y(\alpha_o) = y_o$. Every solution to $F(y,\alpha) = 0$ near $(y_o, \alpha_o)$ must lie on this curve, so there will be no additional fixed points and hence no bifurcation occurs. This shows that (in the one-dimensional case) a bifurcation can only occur when $y_o$ is marginally stable.

To understand the possible types of bifurcations, one uses a change of variables to reduce $F$ to a \emph{normal form}. The three most common are the following.

\begin{description}
    \item[Saddle-node] For
    \begin{equation}
    \label{SN:normal}
        F(y,\alpha) = \alpha + y^2
    \end{equation}
    we have two fixed points for $\alpha < 0$, one for $\alpha=0$, and none for $\alpha > 0$, so the number of fixed points changes as $\alpha$ passes through $0$. Writing the fixed points for $\alpha < 0$ as $y_\pm = \pm \sqrt{-\alpha}$, we calculate
    \begin{equation}
        F_y(y_\pm,\alpha) = 2y_\pm = \pm 2 \sqrt{-\alpha} \; ,
    \end{equation}
    which means $y_-$ is stable and $y_+$ is unstable. We thus have stable and unstable fixed points colliding and annihilating one another, as in  Figure \ref{fig:RNdS}a). 

    \quad This is a \emph{saddle-node bifurcation}; its normal form is (\ref{SN:normal}). At the point $(0,0)$ we have $F_y = 0$, $F_{yy} \neq 0$ and $F_\alpha \neq 0$. Any function $F(y,\alpha)$ that has a fixed point satisfying these conditions can be reduced to either $\alpha + y^2$ or $\alpha - y^2$ by a suitable coordinate transformation, so the corresponding dynamical system undergoes a saddle-node bifurcation at that point. Intuitively this is saying that the zero set of $F(y,\alpha)$ behaves like the zero set of its second-order Taylor polynomial, provided the conditions $F_{yy} \neq 0$ and $F_\alpha \neq 0$ are both satisfied.

\item[Transcritical] 
In a saddle-node bifurcation fixed points are created (or annihilated) in pairs. However, it is also possible for new fixed points to bifurcate from existing ones. The simplest example of this is a transcritical bifurcation, which has normal form
\begin{equation}
    F(y,\alpha) = \alpha y - y^2.
\end{equation}
Here there are two curves of fixed points, $y=0$ and $y = \alpha$, which intersect at $\alpha=0$. In terms of stability we calculate $F_y = \alpha - 2y$ and hence
\begin{equation}
    F_y(0,\alpha) = \alpha \; , \qquad F_y(\alpha, \alpha) = - \alpha \;.
\end{equation}
For $\alpha<0$ we see that $y=0$ is stable and $y = \alpha$ is unstable, and vice versa for $\alpha > 0$. That is, the two curves exchange stability at $\alpha=0$; cf. Figure  \ref{fig:RNdS}b).

\quad In practice a transcritical bifurcation usually occurs when there is a known family of fixed points (in this case $y=0$) that changes stability at some point, resulting in a new family bifurcating from it. The key properties of the normal form are that $F(0,\alpha) = 0$ for all $\alpha$, and at the point $(0,0)$ we have $F_y = 0$, $F_{yy} \neq 0$ and $F_{y\alpha} \neq 0$. As in the saddle-node case, any function satisfying these conditions can be reduced to the normal form by a change of variables. Moreover, if the known family of fixed points is a curve $y_o(\alpha)$, rather than $y=0$, we can reduce to the above case by simply replacing $F(y,\alpha)$ with $F\big(y + y_o(\alpha), \alpha\big)$.

\item[Pitchfork] Finally, suppose that $y=0$ is a fixed point for all $\alpha$, with $F_y = 0$ and $F_{y\alpha} \neq 0$ but $F_{yy} = 0$, so the above discussion does not apply. We then need to look at higher derivatives of $F$. If $F_{yyy} \neq 0$ we have a \emph{pitchfork bifurcation}. There are two possible normal forms, depending on the sign of $F_{yyy}$:
\begin{equation}
    F(y,\alpha) = \alpha y + y^3
\end{equation}
if $F_{yyy} > 0$ (subcritical), and 
\begin{equation}
    F(y,\alpha) = \alpha y - y^3
\end{equation}
if $F_{yyy} < 0$ (supercritical).

\quad In either case we have $F_y(0,\alpha) = \alpha$, so $y=0$ is stable for $\alpha<0$ and unstable for $\alpha>0$. In the subcritical case there are two additional fixed points when $\alpha<0$, namely $y = \pm\sqrt{-\alpha}$, which are both unstable. 
At $\alpha=0$ all three fixed points coincide, and only $y=0$ persists for positive $\alpha$. Similarly, the supercritical case has two additional fixed points, $y = \pm\sqrt{\alpha}$, for $\alpha>0$, which are both stable, and do not persist for negative $\alpha$. An example of a subcritical pitchfork bifurcation is shown in Figure \ref{fig:RNdS}c). 

\quad Pitchforks often occur in the presence of symmetry: if $F(y,\alpha)$ is an odd function of $y$, then $F(0,\alpha) = 0$, meaning $y=0$ is automatically a fixed point  and $F_{yy}(0,\alpha) = 0$. This is the most common way for this bifurcation to occur. 

\end{description}

Saddle-node, transcritical and pitchfork (enforced by symmetry) are the three standard one-parameter bifurcations. Other bifurcations are possible, but only with more vanishing Taylor coefficients. These requires fine-tuning (often via additional parameters) and so are less commonly observed in applications.

\subsection{Systems of equations}

For a system of equations the situation is much the same. Let
\begin{equation}
    \dot{y}^i = F^i(y^j) \label{eq:nDdynsysparamMB}
\end{equation}
be an $n$-dimensional dynamical system, with $y^j \in \mathbb{R}^n$ and $F^i: \mathbb{R}^n \rightarrow \mathbb{R}^n$. The \emph{fixed points} $y_o^i$ satisfy $F^i ( y_o^j) = 0$ and their stability is determined by the Jacobian matrix $[DF_o]_j^i = \partial_j F^i (y_o)$. If all of the eigenvalues have negative real part then $y_o^i$ is \emph{stable}, and if at least one eigenvalue has positive real part then $y_o^i$ is \emph{unstable}.

For a one-parameter family of $n$-dimensional 
dynamical systems
\begin{equation}
    \dot{y}^i = F^i(y^j,\alpha) \label{eq:nDdynsysparam}
\end{equation}
we thus need to understand the zero set of $F^i(y^j,\alpha)$. If $y^i_o$ is a fixed point for $\alpha = \alpha_o$ and the Jacobian is invertible at that point, then there will be a smooth curve of fixed points, $y^i =Y^i(\alpha)$, with $Y^i(\alpha_o) = y^i_o$; otherwise a bifurcation may occur.

Unlike the one-dimensional case, where we simply had $F_y(y_o,\alpha_o) = 0$, here the Jacobian $[DF_o]_j^i$ is an $n\times n$ singular matrix, which could be quite complicated. However, except in the presence of symmetry, the nullspace is usually one dimensional. It is then possible (using Lyapunov--Schmidt reduction, see \cite{GSbifurcation}) to reduce to a one-dimensional bifurcation problem, where the discussion of normal forms in the previous section  applies.

\section{Numerical tools}
\label{sec:numerics}

In general neither the MOTSodesic equations nor the eigenvalue equations can be solved exactly. In this appendix we discuss the numerical tools used for 
these calculations.

\subsection{Solving the MOTSodesic equations}
\label{app:MOTSo}

The MOTSodesic equations (\ref{eq:PT}) are fairly easily solved.  The only real complication comes with the initial conditions.  

\subsubsection{Initial conditions}
All of the examples that we study  in this paper are expressed in  spherical-type coordinates for which there are coordinate degeneracies at $\theta = 0, \pi$. As a result (\ref{eq:PT}) are not well-defined at the endpoints where we need to impose our boundary conditions. Hence, we impose them perturbatively. We consider solutions of the form:
\begin{align}
    P(s) & = P_0 + \cancel{ P_1 s} + P_2 s^2 + P_3 s^3 + P_4 s^4 + P_5 s^5 + P_6 s^6 + \dots \\
    \Theta(s) & =  \cancel{\Theta_0} + \Theta_1 s + \Theta_2 s^2 + \Theta_3 s^3 + \Theta_4 s^4 + \Theta_5 s^5 + \Theta_6 s^6 + \dots \; .
\end{align}
$P_o$ is the intended initial position of MOTSodesic on the $z$-axis and at that point $\Theta=0$, so we have $\Theta_0=0$.  $P_1=0$ by the requirement that the MOTSodesic orthogonally intersects the $z$-axis. Then $\Theta_1$ can be determined by the requirement that $g_{ij}T^i T^j = 1$. From that point the other coefficients can be solved for order-by-order. 
For an example of such a calculation see \cite{Booth:2021sow}. 

In the MOTSodesic calculations in the main body of this paper, for each given $P_0$ we calculated $(P(0.001),\Theta(0.001))$ and 
$(\dot{P}(0.001),\dot{\Theta}(0.001))$ for a fifth order expansion and used those quantities as the initial conditions. 

\subsubsection{Integrating the MOTSodesic equations} 
 With the initial conditions established, (\ref{eq:PT}) can be integrated by any standard numerical integration technique. In particular, the built-in packages in standard symbolic algebra software are more than sufficient to solve for $(P(s), \Theta(s))$.

The shooting method was then used to search for MOTSodesics, systematically working through possible values of $P_0$ to find 
curves that (approximately) close.
We cannot expect perfect closure as the axisymmetry forces an instability into the MOTSodesic equations: if a curve approaches the $z$-axis at any deviation from a right angle, then the MOTS principal curvature associated with the rotationally symmetry will always diverge and this will be sufficient to force the MOTSodesic to turn from the axis. When searching for a MOTSodesic, the key is to look for the point where the curves switch from curving upwards to curving downwards. For more details see the discussion and diagrams in, for example, \cite{Booth:2020qhb, Booth:2021sow}. Other search strategies can be found in \cite{Hennigar:2021ogw}.

\subsection{Calculating stability spectra}
\label{app:spectra}

Calculating the spectra of the MOTS stability operator is a bit more complicated and requires more than the built-in packages of standard computer algebra systems.
We 
proceed using psuedo-spectral methods
\cite{boyd01  ,Pook-Kolb:2018igu}.

This section starts with an overview of the process and then considers its implementation.

\subsubsection{Basic idea}

We wish to calculate the eigenvalues of our axisymmetric stability operator $\tilde{L}$ acting on the space of functions $\tilde{\psi}:[0, \smax] \rightarrow \mathbb{R}$ 
with vanishing derivatives at the boundaries:
\begin{equation}
    \left. \frac{\dd {\tilde\psi}}{\dd s} \right|_{s=0} = \left. \frac{\dd {\psi}}{\dd s} \right|_{\smax} = 0 \label{eq:BC}
    \; . 
\end{equation}
That is, we wish to find pairs  $(\tilde{\psi}(s),\lambda)$ such that
\begin{equation}
    (\tilde{L} - \lambda)  \tilde{\psi} = 0 \; . \label{eq:eig} 
\end{equation}

Let $\{ \mathcal{T}_p, p \in \mathbb{N} \}$ be a set of basis functions satisfying the boundary conditions (\ref{eq:BC}). There is a non-degenerate inner product on these functions:
\begin{equation}
    \mathcal{I}_{pq} =  \langle \mathcal{T}_p, \mathcal{T}_q \rangle \; . 
\end{equation}
If they are orthonormal then $\mathcal{I}_{pq} = \delta_{pq}$ but we do not assume this. 
With respect to the basis functions we can expand
\begin{equation}
    \tilde{\psi}(s) = \sum_{p=1}^\infty \mathcal{A}_p \mathcal{T}_p(s)
\end{equation}
where the $\mathcal{A}_p$ are 
numerical coefficients. Then (\ref{eq:eig}) becomes
\begin{equation}
    \sum_{p=1}^\infty \mathcal{A}_p 
    \left( \tilde{L} - \lambda \right) \mathcal{T}_p = 0 \; . \label{}
\end{equation}

Of course, we still can't solve this in exact form but for finite $N$ we can approximate it using the Galerkin method\cite{boyd01}. If we approximate our eigenfunctions with respect to a subset of basis elements $\{ \mathcal{T}_p : 1 \leq p \leq N \}$, then given (\ref{eq:eig}) the error is the \emph{residual}
\begin{equation}
    R(\mathcal{A}_p, \lambda) = \sum_{p=N+1}^\infty \mathcal{A}_p 
    \left( \tilde{L} - \lambda \right) \mathcal{T}_p
   % =   - \sum_{p=0}^N \mathcal{A}_p 
   % \left( \tilde{L} - \lambda \right) \mathcal{T}_p 
   \; . 
\end{equation}
We then solve for $(\mathcal{A}_p, \lambda)$ such that the residual is orthogonal to our finite subset of basis elements:
\begin{equation}
    \langle \mathcal{T}_p, R(\mathcal{A}_q, \lambda) \rangle = 0 \quad \mbox{for } 1 \leq p,q \leq N \; . 
\end{equation}
Equivalently, by (\ref{eq:eig}),
\begin{equation}
    \sum_{q=1}^N \left( \tilde{\mathcal{L}}_{pq} - \lambda \mathcal{I}_{pq} \right) \mathcal{A}_q = 0 \, , 
\end{equation}
where 
\begin{equation}
    \tilde{\mathcal{L}}_{pq} = \langle \mathcal{T}_p , \tilde{L} \mathcal{T}_q \rangle \; . 
\end{equation}

Then the eigenvalues and eigenfunctions are approximated by solving the finite-dimensional linear algebra problem
% \begin{equation}
%     \sum_{q,r=1}^N \left( \mathcal{I}^{-1}_{pq}  \tilde{\mathcal{L}}_{qr} - \lambda \delta_{qr} \right) \mathcal{A}_r = 0 
% \end{equation}
\begin{equation}
    \sum_{r=1}^N \left(  \sum_{q=1}^N \mathcal{I}^{-1}_{pq}  \tilde{\mathcal{L}}_{qr} - \lambda \delta_{pr} \right) \mathcal{A}_r = 0 
\end{equation}
where $\mathcal{I}^{-1}_{pq}$ is the inverse of the matrix of  
inner products for the finite subset of basis elements. 

To implement this in our calculations we need to specify the basis elements $\mathcal{T}_p$ and the inner product, along with a procedure to calculate that inner product. 

\subsubsection{Basis functions}

As basis elements, we choose the  
Chebyshev polynomials
\begin{equation}
   \left\{ \mathcal{T}_p \left( \frac{2s}{\smax} -1 \right) \, : \,  p \in \mathbb{N} \right\}
\end{equation}
where 
\begin{equation}
    \mathcal{T}_p (y) = 
    \left\{
    \begin{array}{ll}
    T_0(y) & \mbox{if  } p=1 \\
    T_{p-1}(y) - \left( \frac{p-1}{p+1}\right)^2
    T_{p+1}(y) & \mbox{if }p >1
    \end{array} \right. 
\end{equation}
for the standard Chebyshev polynomials $T_q(y):[-1,1] \rightarrow \mathbb{R}$
defined by 
\begin{equation}
    T_0  = 1 \; , \quad
    T_1  = y \quad  \mbox{and } \; 
    T_{n+1} = 2 y T_n(y) - T_{n-1} (y)  
\end{equation}
or, equivalently,
$
    T_p(\cos(\theta)) = \cos (p \theta) 
$.
%We have shifted the numbering of the $\mathcal{T}$ to start at 1 rather than 0 to simplify later expressions. 
The $\mathcal{T}_p$ satisfy the boundary conditions (\ref{eq:BC}) but are neither normalized nor orthogonal. We could construct an orthonormal basis with the usual Gram--Schmidt process but this is computationally cumbersome. Instead we work with the $\mathcal{T}_p$ directly. 

With respect to the standard weight function 
\begin{equation}
  \langle T_p, T_q \rangle =   \int_{-1}^1 \frac{T_p(y) T_q(y) }{\sqrt{1-y^2}}  \dd y
    = \left\{\begin{array}{ll} 
    0  & \mbox{if } p \neq q\\
    \pi &  \mbox{if } p = q = 0\\
    \frac{\pi}{2} & \mbox{if } p=q \neq 0
    \end{array}  \right. 
\end{equation}
and so the inner products of the modified $\mathcal{T}_p$ are
\begin{equation}
    \mathcal{I}_{pq} := \langle \mathcal{T}_p ,  \mathcal{T}_q \rangle
    = \left\{\begin{array}{ll} 
    \pi &  \mbox{if } p = q = 1\\
    \frac{\pi}{2}  \left(1+\left(\frac{p-1}{p +1}\right)^4 \right)& \mbox{if } p=q \neq 1 \\
-\frac{\pi}{2}  \left( \frac{p-1}{p +1} \right)^2
 & \mbox{if } q=p+2 \\
 -\frac{\pi}{2} \left( \frac{q-1}{q +1} \right)^2
 & \mbox{if } p=q+2 \\
 0 & \mbox{otherwise}
    \end{array}  \right. \; . 
\end{equation}
\bigskip 
For $1 \leq p, q \leq N$, $\mathcal{I}_{pq}$ is a symmetric $N\!\times\! N$ matrix and we denote its inverse $\mathcal{I}^{-1}_{pq}$.

Then the $N^{\mathrm{th}}$-order approximation of a function $\tilde{\psi}(y)$ 
%(that also satisfies the boundary conditions) 
is
\begin{equation}
    \tilde{\psi} \approx \sum_{p=1}^N \mathcal{A}_p \mathcal{T}_p
\quad \mbox{with coefficients} \quad
    \mathcal{A}_p = \sum_{q=1}^N I^{-1}_{pq} \langle \mathcal{T}_q, \psi \rangle \; .   \label{eq:ChebExp}
\end{equation} 

\subsubsection{Calculating $\tilde{\mathcal{L}}_{pq}$}

While we can calculate $\mathcal{I}_{pq}$ exactly, the operator $\tilde{L}$  is usually given numerically, with the $\mathcal{G}$ and $\mathcal{H}$ coefficients depending on $(P(s), \Theta(s))$ as in (\ref{eq:cG}) and (\ref{eq:cH}), hence the $\tilde{\mathcal{L}}_{pq}$ will also have to be calculated numerically. 

Applying the coordinate transformation $s = \frac{\smax}{2} ( y + 1 )$, 
(\ref{eq:L1d}) becomes
\begin{equation}
    \tilde{L} = \left( - \frac{4}{\smax^2} \right) \frac{\dd^2}{\dd^2 y} + \left( \frac{2\mathcal{G}(s)}{\smax}  \right) \frac{\dd}{\dd y} + \mathcal{H}(s) \; . 
\end{equation}
Then we need a method for 
calculating 
\begin{equation}
    \tilde{\mathcal{L}}_{pq} =  \int_{-1}^1
    \frac{\mathcal{T}_p (y) \tilde{L} [\mathcal{T}_q(y)]}{\sqrt{1 - y^2}} \dd y  \; . 
\end{equation}
The natural way to do this is by Chebyshev--Gauss 
quadrature\cite{boyd01}. 
To that end we will evaluate functions at $M$ \emph{collocation points}:
\begin{equation}
    y_i = -\cos \left( \frac{2i-1}{2M} \pi \right) 
    \quad \Longleftrightarrow \quad s_i = \frac{\smax}{2} \left(1- \cos \left( \frac{2i-1}{2M} \pi \right)  \right)  . 
\end{equation}
These are the zeros of the Chebyshev polynomials. 
Then
\begin{equation}
    \tilde{\mathcal{L}}_{pq} \approx 
    \frac{\pi}{n} \sum_{i=1}^M
    \mathcal{T}_p(s_i) \left(
     - \frac{4\mathcal{T}''_q(y_i)}{\smax^2}     +  \frac{2\mathcal{G}(s_i) \mathcal{T}'_q(y_i)}{\smax}   + \mathcal{H}(s_i) \mathcal{T}_q (y_i) \right)  , 
\end{equation}
where the primes are $y$-derivatives. 

\subsubsection{Calculating eigenvalues and eigenfunctions}

Once $\mathcal{I}_{pq}$ and $\tilde{\mathcal{L}}_{pq}$ are known, it is straightforward to calculate the (approximate) eigenvalues and corresponding eigenfunctions (from the coefficients $\mathcal{A}_p$) using standard numerical linear algebra packages, including those built into any computer algebra software. 

As a rule-of-thumb, choosing $M = N+2$ appears to provide as much accuracy as you can get for a given $N$. For a given $N$, the lower eigenvalues can be expected to the most accurate with a second rule-of-thumb being that the first $N/2$ should be accurate to several decimal places. These have also been observed in other studies\cite{Pook-Kolb:2018igu,Pook-Kolb:2021jpd,Hennigar:2021ogw}. We have tested this assumption in many cases by calculating the eigenvalues and seeing how they (do not) change as $N$ increases.

\bigskip

\bibliographystyle{abbrv}
\bibliography{References}

\end{document}